\documentclass[aps,letterpaper,nofootinbib,showpacs,preprintnumbers,amsmath,amssymb]{revtex4}
\usepackage{latexsym}
\usepackage{hyperref}
\usepackage{amssymb}
\usepackage{graphicx}
\usepackage{amsmath}
\usepackage{amsthm,amssymb,amstext,amscd,amsfonts,amsbsy,amsxtra,latexsym}
\usepackage{wrapfig}
\usepackage{color}
\usepackage{pdflscape}
\usepackage{lscape}
\usepackage{tabularx}
\usepackage{lipsum}
\usepackage{rotating} 
\usepackage{pifont}
\usepackage{bbding}
\usepackage{upgreek}
\usepackage{mathrsfs}
\usepackage{wasysym}
\usepackage{phonetic}

\newcommand{\lb}{\label}
\newcommand{\be}{\begin{equation}}
\newcommand{\ee}{\end{equation}}
\newcommand{\ben}{\begin{eqnarray*}}
\newcommand{\een}{\end{eqnarray*}}
\newcommand{\bea}{\begin{eqnarray}}
\newcommand{\eea}{\end{eqnarray}}
\newcommand{\md}{{\mathrm{d}}}
\newcommand{\beq}{\begin{equation}}
\newcommand{\eeq}{\end{equation}}
\newcommand{\beqn}{\begin{equation}\nonumber}
\newcommand{\bean}{\begin{eqnarray}\nonumber}
\DeclareMathAlphabet{\mathpzc}{OT1}{pzc}{m}{it}

\def\cR{{\mathcal R}}
\def\cH{{\mathcal H}}

\def\cB{{\mathcal B}}
\def\cN{{\mathcal N}}

\def\cQ{{\mathcal Q}}

\def\cD{{\mathcal D}}
\def\cL{{\mathcal L}}
\def\cM{{\mathcal M}}

\def\cC{{\mathcal C}}

\def\cI{{\mathcal I}}

\def\bbR{{\mathbb R}}
\def\bbS{{\mathbb S}}
\def\dV{{\mbox{dVol}}}
\def\schere{\text{\ding{34}}}
\newtheorem{thm}{Theorem}[section]

\newtheorem{lem}[thm]{Lemma}
\newtheorem{prop}[thm]{Proposition}
\newtheorem{cor}[thm]{Corollary}
\newtheorem{con}[thm]{Conjecture}
\renewenvironment{proof}[1][Proof]{\begin{trivlist}
\item[\hskip \labelsep {\bfseries #1:}]}{\qed\end{trivlist}}

\newcommand{\nabb}{\mbox{$\nabla \mkern-13mu /$\,}}

\def\XXint#1#2#3{{\setbox0=\hbox{$#1{#2#3}{\int}$ }
\vcenter{\hbox{$#2#3$ }}\kern-.6\wd0}}

\makeatletter
\def\p@subsection{}
\def\p@subsubsection{}
\makeatother

\begin{document}

\begin{center}
{\bf {\Large 
BOUNDEDNESS OF MASSLESS SCALAR WAVES ON REISSNER-NORDSTR\"OM INTERIOR BACKGROUNDS}}\\

\bigskip
\bigskip
Anne T. Franzen\footnote{e-mail addresses: atfranzen@math.princeton.edu, a.t.franzen@uu.nl}

\bigskip
\bigskip
{\it Princeton University}\\ 
{\it Department of Mathematics, Fine Hall,
Washington Road, Princeton, NJ 08544,
United States}

\medskip
\medskip

{\it Universiteit Utrecht}\\
{\it Leuvenlaan 4, 3584 CE Utrecht, the Netherlands}

\end{center}

\medskip

\centerline{ABSTRACT}

\noindent
We consider solutions of the scalar wave equation $\Box_g\phi=0$, without symmetry, on fixed subextremal Reissner-Nordstr\"om backgrounds $(\cM, g)$ with nonvanishing charge.
Previously,
it has been shown that for $\phi$ arising from sufficiently regular data on a two ended Cauchy hypersurface, the solution and its derivatives decay suitably fast on the event horizon $\cH^+$. Using this, we show here that $\phi$ is in fact uniformly bounded, $|\phi| \leq C$, in the black hole interior up to and including the bifurcate Cauchy horizon $\cC\cH^+$, to which $\phi$ in fact extends continuously. 
The proof depends on novel weighted energy estimates in the black hole interior which, in combination with commutation by angular momentum operators and application of Sobolev embedding, yield uniform pointwise estimates. 
In a forthcoming companion paper we will extend the result to subextremal Kerr backgrounds with nonvanishing rotation.
\medskip

\tableofcontents
\section{Introduction\label{intro}}
The Reissner-Nordstr\"om spacetime $(\cM,g)$ is a fundamental 2-parameter family of solutions to the Einstein field equations coupled to electromagnetism, cf.~ Figure \ref{RN} for the conformal representation of the subextremal case, \mbox{$M>|e|\neq 0$}, with $e$ the charge and $M$ the mass of the black hole. 
The problem of analysing the scalar wave equation 
\bea
\lb{wave}
\Box_g \phi=0
\eea 
on a Reissner-Nordstr\"om background is intimately related to the stability properties of the spacetime itself and to the celebrated Strong Cosmic Censorship Conjecture. 
{\begin{figure}[ht]
\centering
\includegraphics[width=0.5\textwidth]{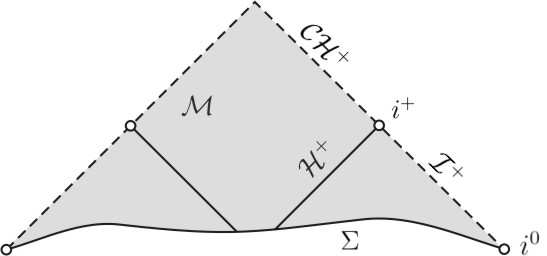}
\caption[]{Maximal development of Cauchy hypersurface $\Sigma$ in Reissner-Nordstr\"om spacetime $(\cM,g)$.}
\label{RN}\end{figure}}
The analysis of \eqref{wave} in the exterior region $J^-(\cI^+)$ has been accomplished already, cf.~ \cite{blue} and \cite{m_lec} for an overview and references therein for more details, as well as Section \ref{horizon_estimates}. The purpose of the present work is to extend the investigation to the interior of the black hole, up to and including the Cauchy horizon $\cC\cH^+$.

\subsection{Main result}

The main result of this paper can be stated as follows.
\begin{thm}
\lb{main}
On subextremal Reissner-Nordstr\"om spacetime $(\cM,g)$, with mass $M$ and charge $e$ and $M>|e|\neq 0$, let $\phi$ be 
a solution of the wave equation $\Box_g \phi=0$ arising from
sufficiently regular Cauchy data on a two-ended asymptotically flat Cauchy surface $\Sigma$. Then
\bea
\lb{maineq}
|\phi|\leq C
\eea
globally in the black hole interior, in particular up to and including the Cauchy horizon
$\cC\cH^+$, to which $\phi$ extends in fact continuously.
\end{thm}
The constant $C$ is explicitly computable in terms of parameters $e$ and $M$ and a suitable norm on initial data.
The above theorem will follow, after commuting \eqref{wave} with angular momentum operators and applying Sobolev embedding, from the following theorem, expressing weighted energy boundedness.
\begin{thm}
\lb{energythm}
On subextremal Reissner-Nordstr\"om spacetime $(\cM,g)$, with mass $M$ and charge $e$ and $M>|e|\neq 0$, let $\phi$ be 
a solution of the wave equation $\Box_g \phi=0$ arising from
sufficiently regular Cauchy data on a two-ended asymptotically flat Cauchy surface $\Sigma$. Then 
\bea
\lb{energy1}
\int\limits_{\bbS^2}\int\limits^{\infty}_{v_{fix}}\left[ v^p (\partial_v \phi)^2(u, v, \theta, \varphi) + |\nabb \phi|^2(u, v, \theta, \varphi) \right]r^2\md v\md \sigma_{\mathbb S^2}&\leq& E,\quad \mbox{for $v_{fix} \geq 1$, $u > -\infty$}\\
\lb{energy1u}
\int\limits_{\bbS^2}\int\limits^{\infty}_{u_{fix}}\left[ u^p (\partial_u \phi)^2 (u, v, \theta, \varphi)+ |\nabb \phi|^2(u, v, \theta, \varphi) \right]r^2\md u\md \sigma_{\mathbb S^2}&\leq& E,\quad \mbox{for $u_{fix} \geq 1$, $v > -\infty$}
\eea
where $p>1$ is an appropriately chosen constant, and $(u, v)$ denote Eddington-Finkelstein coordinates in the black hole interior, where by $\md \sigma_{\mathbb S^2}$ we denote the volume element of the unit two-sphere and \mbox{$|\nabb \phi|^2=\frac{1}{r^2}\left[(\partial_{\theta} \phi)^2+ \frac{1}{\sin^2 {\theta}}(\partial_{\varphi} \phi)^2\right]$}.  
\end{thm}

\subsection{A first look at the analysis}
\lb{firstlook}
The proof of Theorem \ref{main} and \ref{energythm} involves first considering a characteristic rectangle $\Xi$ within the black hole interior, whose future {\it right} boundary coincides with the Cauchy horizon $\cC\cH^+$ in the vicinity of $i^+$, cf.~ Figure \ref{alle}. 
{\begin{figure}[ht]
\centering
\includegraphics[width=0.6\textwidth]{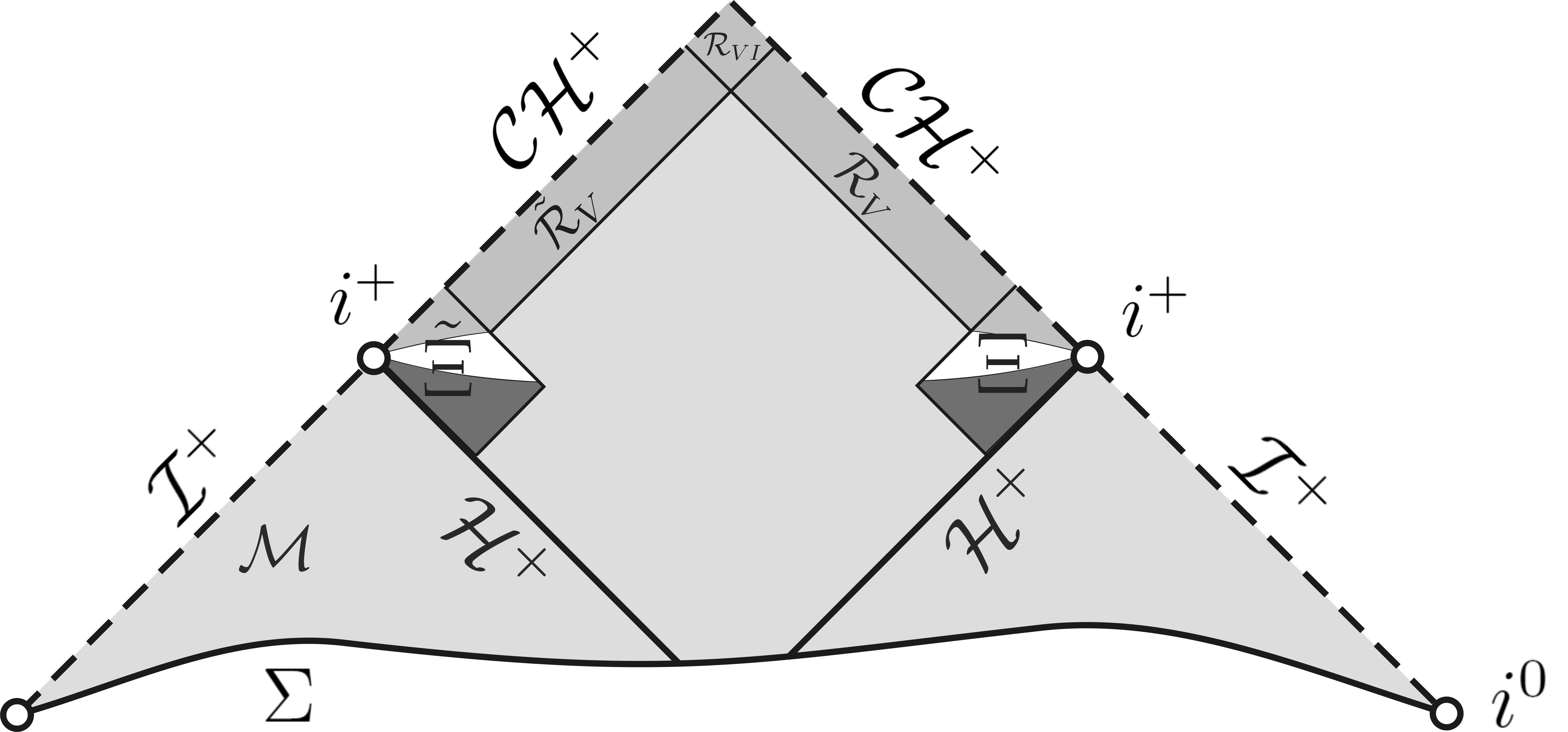}
\caption[]{Penrose diagram depicting all regions considered in the entire proof.}
\label{alle}\end{figure}}

Establishing boundedness of weighted energy norms in $\Xi$ is the crux of the entire proof. 
Once that is done, analogous results hold for a characteristic rectangle $\tilde{\Xi}$ to the {\it left} depicted in Figure \ref{alle}. Hereafter, boundedness of the energy is easily propagated to regions $\cR_V$, $\tilde{\cR}_V$ and $\cR_{VI}$ as depicted, giving Theorem \ref{energythm}. Commutation by angular momentum operators and application of Sobolev embedding then yields Theorem \ref{main}.

Let us return to the discussion of $\Xi$ since that is the most involved part of the proof. 
{\begin{figure}[ht]
\centering
\includegraphics[width=0.4\textwidth]{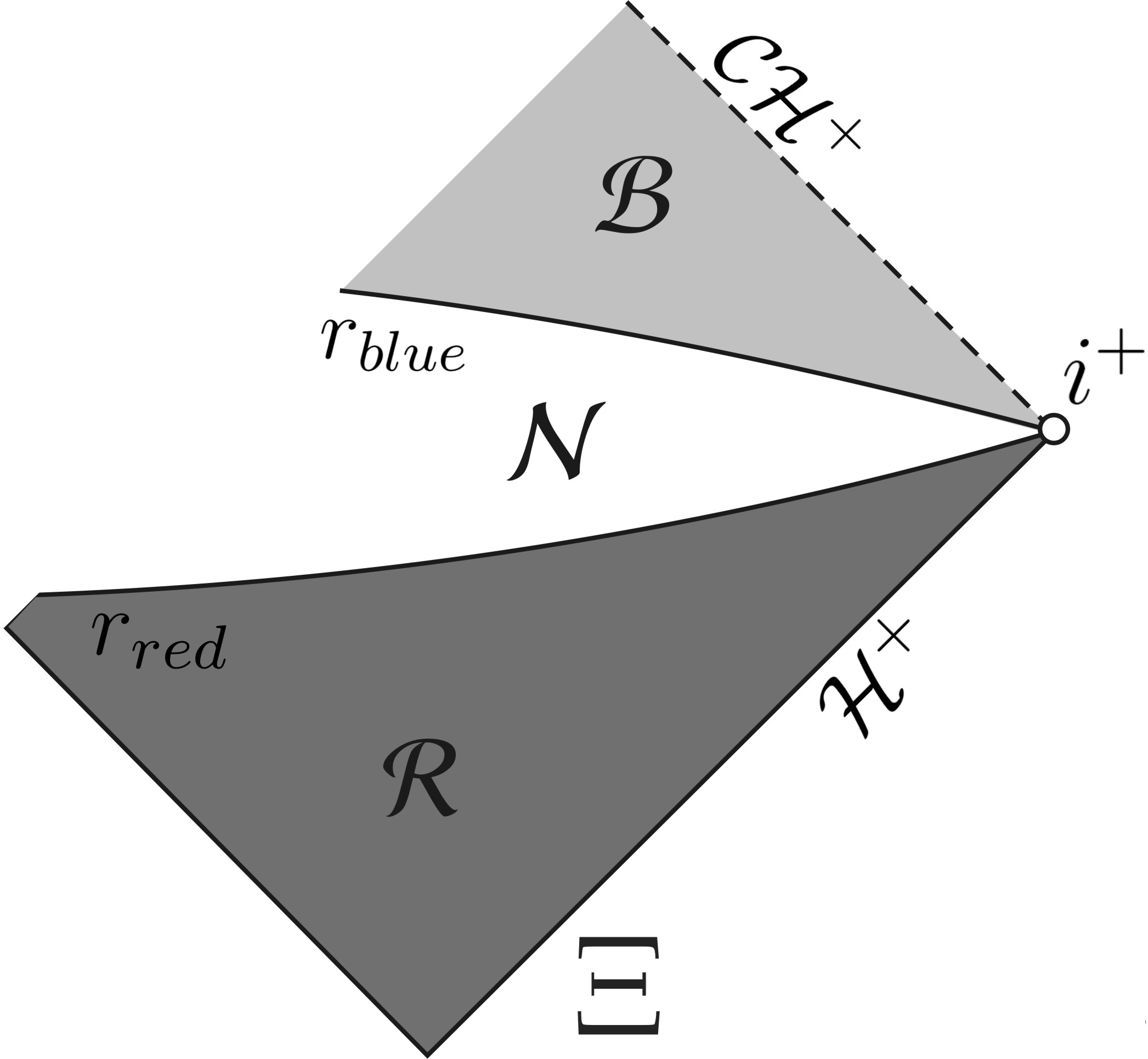}
\caption[]{Conformal representation of a characteristic regtangle $\Xi$ with redshift $\cR$, noshift $\cN$ and blueshift regions $\cB$.}
\label{bnr1}\end{figure}}

In order to prove Theorem \ref{energythm} (and hence Theorem \ref{main}) restricted to $\Xi$ we will begin with an upper decay bound for $|\phi|$ and its derivatives on the event horizon $\cH^+$, which can be deduced by putting together preceding work of Blue-Soffer, cf.~ \cite{blue}, Dafermos-Rodnianski, cf.~ \cite{m_price} and Schlue, cf.~ \cite{volker}. The precise result from previous work that we shall need will be stated in Section \ref{horizon_estimates}.

In $\Xi$ the proof involves distinguishing redshift $\cR$, noshift $\cN$ and blueshift $\cB$ regions, as shown in Figure \ref{bnr1}.

Some of these regions have appeared in previous analysis of the wave equation, especially \mbox{$\cR=\left\{r_{red}\leq r \leq r_+\right\}$}. Region \mbox{$\cN=\left\{r_{blue}\leq r\leq r_{red}\right\}$} and region \mbox{$\cB=\left\{r_{-}\leq r\leq r_{blue}\right\}$} were studied in \cite{m_interior} in the {\it spherically symmetric self-gravitating} case, but using techniques which are very special to $1+1$ dimensional  hyperbolic equations.\footnote{Let us note that the result of Theorem \ref{main} for \emph{spherically symmetric} solutions $\phi$ can be obtained by specializing \cite{m_interior, m_price, m_bh} to the uncoupled case. Restricted results for fixed spherical harmonics can be in principle also inferred from \cite{mc_b}.}
We will discuss this separation into $\cR$, $\cN$ and $\cB$ regions further in Section \ref{outline}.
One of the main analytic novelties of this paper is the introduction of a new vector field energy identity
constructed for analyses in region $\cB$.
In particular, the weighted vector field is given in Eddington-Finkelstein coordinates $(u,v)$ by
\ben
S=|u|^p\partial_u+v^p\partial_v,
\een
for $p>1$ as appearing in Theorem \ref{energythm}. This vector field associated to region $\cB$ will allow to prove uniform boundedness despite the blueshift instability.

\subsection{Outline of the paper}

The paper is organized as follows.

In the remaining Section \ref{motivation} of the introduction we will elaborate on Strong Cosmic Censorship and its relation to this work.

In Section \ref{preliminaries} we introduce the basic tools needed to derive energy estimates from the energy momentum tensor associated to \eqref{wave} and an appropriate vector field. A review of the Reissner-Nordstr\"om solution and the coordinates used in this paper will be given. Moreover, we will discuss further features of Reissner-Nordstr\"om geometry.

In Section \ref{horizon_estimates} we give a brief review of estimates obtained along $\cH^+$ from previous work, \cite{blue}, \cite{m_price} and \cite{volker}, for $\phi$ arising from sufficiently regular initial data on a Cauchy hypersurface. This is stated as Theorem \ref{anfang}.
Section \ref{outline} states our main result specialized to the rectangle $\Xi$ (see Theorem \ref{dashier}) and gives an outline of its proof. The investigation is divided into considerations within the redshift $\cR$, noshift $\cN$, and blueshift $\cB$ regions. 

The decay bound for the energy flux of $\phi$ given on the event horizon $\cH^+$, cf.~ Theorem \ref{anfang}, will be propagated through the redshift region $\cR$ up to the hypersurface $r=r_{red}$ in Section \ref{first_section}.

Thereafter, in Section \ref{zweite_region} we propagate the decay bound further into the black hole interior through the noshift region $\cN$ up to the hypersurface $r=r_{blue}$.

In Section \ref{blueshift1} a decay bound for the energy flux of $\phi$ is proven on a well chosen hypersurface $\gamma$ that separates the blueshift region into a region in the past of $\gamma$, $J^-(\gamma) \cap \cB$, and a region to the future of $\gamma$, $J^+(\gamma) \cap \cB$.
In Section \ref{innerhorizon} we will derive pointwise estimates on $\Omega^2$ to the future of $\gamma$ (in particular implying finiteness of the spacetime volume, $\operatorname{Vol}(J^+(\gamma))<C$). This will allow us to propagate our estimates into $J^+(\gamma) \cap \cB$ up to $\cC\cH^+$, yielding finally Corollary \ref{endeacht}. 

Section \ref{nineten} reveals how commutation with angular momentum operators and applying Sobolev embedding will return us pointwise boundedness for $|\phi|$. 
The necessary higher order boundedness statement is given in Theorem \ref{energythm2}. This completes the proof of Theorem \ref{dashier}.

We now must extend our result to the full interior region.

In Section \ref{leftinterior} we will state the analog of Theorem \ref{dashier} restricted to the rectangle $\tilde{\Xi}$ to the {\it left}. In Section \ref{region5_proof} and Section \ref{region_tilde5_proof} we propagate the energy estimates further along $\cC\cH^+$ in the depicted regions $\cR_V$ and $\tilde{\cR}_V$. Eventually, in Section \ref{bifurcate} we propagate the estimate to the region $\cR_{VI}$ up to the bifurcation two-sphere, and thus obtain a bound for the energy flux globally in the black hole interior (see Corollary \ref{letzteprop}) completing the proof of Theorem \ref{energythm}.

In Section \ref{global} we prove Theorem \ref{energythm3}, stating boundedness of the weighted higher order energies. Using the conclusion of this theorem, we apply again Sobolev embedding as before (using also the result of Section \ref{uni_bounded}) and thus obtain the boundedness statement of Theorem \ref{main}.
Finally, in Section \ref{continuity} we show continuous extendibility of $\phi$ to the Cauchy horizon.

An Outlook of open problems will be given in Section \ref{outlook}.
We first state an analogous result to our Theorem \ref{main} for general subextremal Kerr black holes (to appear as Theorem 1.1 of \cite{ich}). The conjectured blow up of the transverse derivatives\footnote{Note in contrast that the {\it tangential} derivatives of $\phi$ can be shown to be uniformly bounded up to $\cC\cH^+$ (away from the bifurcation sphere) from the energy estimates proven in this paper together with commuting with angular momentum operators, cf.~ Theorem \ref{energythm}.} along the Cauchy horizon for generic solutions of \eqref{wave} will also be discussed, as well as the peculiar extremal case.
Finally, we will discuss what our results suggest about the nonlinear dynamics of the Einstein equations themselves.

\subsection{Motivation and Strong Cosmic Censorship}
\lb{motivation}
Our motivation for proving Theorem \ref{main} is the Strong Cosmic Censorship Conjecture.
The mathematical formulation of this conjecture, here applied to electrovacuum, is given in \cite{christo_sing} by Christodoulou as 
\bea
\lb{sccc}
\begin{minipage}[c]{0.65\textwidth}
``\underline{Generic} asymptotically flat initial data for Einstein-Maxwell spacetimes 
 have a maximal future development which is
inextendible as a suitably regular Lorentzian manifold.''
\end{minipage}
\eea
Reissner-Nordstr\"om spacetime serves as a counterexample to the inextendibility statement since it is (in fact smoothly) extendable beyond the Cauchy horizon $\cC\cH^+$.\footnote{Outside the future maximal domain of dependence $\cD^+(\cM)$ in the future of the Cauchy horizon $J^+(\cC\cH^+)$ the spacetime shows the peculiar feature that uniqueness of the solutions of the initial value problem is lost {\it without loss of regularity}.
It is precisely the undesirability of this feature that motivates the conjecture.}
Thus, for the above conjecture to be true, this property of Reissner-Nordstr\"om must in particular be unstable.

Originally it was suggested by Penrose and Simpson that small perturbations of Reissner-Nordstr\"om would lead to a spacetime whose boundary would be a spacelike singularity as in Schwarzschild and such that the spacetime would be inextendable as a $C^0$ metric, cf.~ \cite{simpson}.  
On the other hand, a heuristic study of a spherically symmetric but fully nonlinear toy model by Israel and Poisson, cf.~ \cite{poisson}, led to an alternative scenario, which suggested
that spacetimes resulting from small perturbations would exist up to a Cauchy horizon, which however would be singular in a weaker sense, see also \cite{amos} by Ori. Considering the {\it spherically symmetric} Einstein-Maxwell-scalar field equations as a toy model, Dafermos proved that the solution indeed exists up to a Cauchy horizon and moreover is extendible as a $C^0$ metric but generically fails to be extendible as a $C^1$ metric beyond $\cC\cH^+$, cf.~ \cite{m_stab, m_interior}. For more recent extensions see \cite{kommemi, joao1, joao2, joao3}.

In this work, as a first attempt towards investigation of the stability of the Cauchy horizon under perturbations {\it without symmetry}, we employ \eqref{wave} on a fixed Reissner-Nordstr\"om background $(\cM,g)$ as a toy model for the full nonlinear Einstein field equations, cf.~ \eqref{EF}. The result of uniform pointwise boundedness of $\phi$ and continuous extension to $\cC\cH^+$ is concordant with the work of Dafermos \cite{m_stab}. This suggests that the non-spherically symmetric perturbations of the astrophysically more realistic Kerr spacetime may indeed exist up to $\cC\cH^+$.
See Section \ref{EFvacuum}.

\section{Preliminaries}
\lb{preliminaries}
\subsection{Energy currents and vector fields}
\lb{en_cur}
The essential tool used throughout this work is the so called vector field method. 
Let $(\cM, g)$ be a Lorentzian manifold. Let $\phi$ be a solution to the wave equation $\Box_g \phi=0$. 
A symmetric stress-energy tensor can be identified from variation of the massless scalar field action by
\ben
\lb{energymomentum}
T_{\mu\nu}(\phi)=\partial_\mu\phi\partial_\nu\phi-\frac12g_{\mu\nu}g^{\alpha\beta}\partial_\alpha\phi
\partial_\beta\phi,
\een
and this satisfies the energy-momentum conservation law
\begin{equation}
\label{divfree}
\nabla^\mu T_{\mu\nu}=0.
\end{equation}
By contracting the energy-momentum tensor with a vector field $V$, we define the current
\be
\lb{J}
J_\mu^V(\phi)\doteq T_{\mu\nu}(\phi) V^\nu.
\ee
In this context we call $V$ a multiplier vector field. If the vector field $V$ is timelike, then the one-form $J_\mu^V$ can be interpreted as the energy-momentum density. 
When we integrate $J_\mu^V$ contracted with the normal vector field over an associated hypersurface we will often refer to the integral as energy flux.
Note that $J^V_\mu (\phi)n^\mu_{\Sigma} \ge 0$ if $V$ is future directed timelike and $\Sigma$ spacelike, where $n^\mu_{\Sigma}$ is the future directed normal vector on the hypersurface $\Sigma$. 

Since we will frequently use versions of the divergence theorem, we are interested in the divergence of the current \eqref{J}. 
Defining
\be
\lb{K}
K^V(\phi)\doteq T(\phi)(\nabla V)=(\pi^V)^{\mu\nu}T_{\mu\nu}(\phi),
\ee
by \eqref{energymomentum} it follows that
\bea
\nabla^{\mu}J^V_{\mu}(\phi)=K^V(\phi).
\eea
Further, $(\pi^V)^{\mu\nu} \doteq \frac{1}{2} (\cL_V g)^{\mu\nu}$ is the so called deformation tensor 
of $V$.
Therefore, $\nabla^\mu J^V_\mu(\phi) =0$ if $V$ is Killing.

For a Killing vector field $W$ we have in addition the commutation relation $[\Box_g, W]=0$. In that context $W$ is called a commutation vector field.
In particular, we note already that
in Reissner-Nordstr\"om spacetime we have $\Box_g T \phi=0$ and $\Box_g \leo_i \phi=0$, where $T$ and $\leo_i$ with $i=1,2,3$ are Killing vector fields that will be defined in Section \ref{rtcoords} and \ref{angular}, respectively.\\

For a more detailed discussion see \cite{m_lec} by Dafermos and Rodnianski, \cite{sergiu} by Klainerman and \cite{christo_action} by Christodoulou.

\subsection{The Reissner-Nordstr\"om solution}
\lb{RNsection}
In the following we will briefly recall the Reissner-Nordstr\"om solution\footnote{The reader unfamiliar with this solution may for example consult \cite{haw_ellis} for a more detailed review.} which is a family of solutions to the
Einstein-Maxwell field equations 
\bea
\lb{EF}
R_{\mu \nu}-\frac12 g_{\mu \nu}R=2 T^{EM}_{\mu \nu},
\eea
with $R_{\mu \nu}$ the Ricci tensor, $R$ the Ricci scalar and the units chosen such that $\frac{8 \pi G}{c^4}=2$. The Maxwell equations are given by
\bea
\lb{maxwell}
\nabla^{\alpha} F_{\alpha \beta}=0, \qquad \nabla_{[\lambda} F_{\alpha \beta]}=0,
\eea
and the energy-momentum tensor by
\bea
\lb{tee1}
T^{EM}_{\mu \nu}&=& F^\alpha_{\mu}F_{\alpha \nu}-\frac{1}{4} g_{\mu \nu}F^{\alpha \beta}F_{\alpha \beta}.
\eea
The system \eqref{EF}-\eqref{tee1} describes the interaction of a gravitational field with a source free electromagnetic field.

The Reissner-Nordstr\"om solution represents a charged black hole as an isolated system in an asymptotically Minkowski spacetime.
The causal
structure is similar to the structure of the astrophysically more realistic axisymmetric Kerr
black holes. 
Since spherical symmetry can often simplify first investigations,
Reissner-Nordstr\"om spacetime is a popular proxy for
Kerr. 

\subsubsection{The metric and ambient differential structure}
To set the semantic convention, whenever we refer to the Reissner-Nordstr\"om solution $(\cM,g)$ we mean the maximal domain of dependence \mbox{$\cD (\Sigma)=\cM$} of complete two-ended asymptotically flat data $\Sigma$. 
The manifold $\cM$ can be expressed by \mbox{$\cM=\cQ\times \bbS^2$}, and \mbox{$\cQ=(-1, 1) \times (-1,1)$} with coordinates $U, V \in (-1,1)$ and thus
\bea
\cM=(-1, 1) \times (-1,1) \times \bbS^2. 
\eea
The metric in global double null coordinates then takes the form
\bea
\lb{metric}
g=-\Upomega^2(U,V)\md U \md V+{\mathpzc{r}}^2(U,V)\left[\md \theta^2+\sin^2\theta \md \varphi^2\right],
\eea
where $\Upomega^2$ and $\mathpzc{r}$ will be described below.

As a gauge condition we choose the hypersurface $U=0$ and $V=0$ to coincide with what will be the event horizons and we set
\bea
\lb{UVrange}
\Upomega^2(0,V)&=&\frac{1}{1-V^2},\\
\Upomega^2(U,0)&=&\frac{1}{1-U^2},
\eea
consistent with the fact that these hypersurfaces are to have infinite affine length. Fix parameters \mbox{$M>|e|\neq 0$}. The Reissner-Nordstr\"om metric \eqref{metric} in our gauge is uniquely determined from \eqref{EF}-\eqref{tee1} by setting
\bea
{\mathpzc{r}}(0,V)&=&{\mathpzc{r}}|_{{\cH_A}^+}=M+\sqrt{M^2-e^2}=r_+,\\
{\mathpzc{r}}(U,0)&=&{\mathpzc{r}}|_{{\cH_B}^+}=M+\sqrt{M^2-e^2}=r_+.
\eea
Rearranging the Einstein-Maxwell equations \eqref{EF} using \eqref{metric} we obtain the following Hessian equation 
\bea
\lb{uvr}
\partial_U \partial_V {\mathpzc{r}} &=&\frac{e^2 \Upomega^2}{4{\mathpzc{r}}^3}-\frac{\Upomega^2}{4{\mathpzc{r}}}-\frac{\partial_U {\mathpzc{r}} \partial_V {\mathpzc{r}}}{{\mathpzc{r}}},
\eea
from the $U,V$ component.
From the $\theta,\theta$ or equivalently $\phi,\phi$ component we obtain
\bea
\lb{omegalogrelation}
\partial_U \partial_V \log \Upomega^2&=&-\frac{2\partial_U \partial_V {\mathpzc{r}}}{{\mathpzc{r}}}\\
&\stackrel{\eqref{uvr}}{=}&-\frac{e^2 \Upomega^2}{2{\mathpzc{r}}^4}+\frac{\Upomega^2}{2{\mathpzc{r}}^2}+\frac{2\partial_U {\mathpzc{r}} \partial_V {\mathpzc{r}}}{{\mathpzc{r}}^2},
\eea

In fact, all relevant information about Reissner-Nordstr\"om geometry can be understood directly from \eqref{UVrange} to \eqref{omegalogrelation} without explicit expressions for $\Upomega^2(U,V)$ and ${\mathpzc{r}}(U,V)$. In particular, one can derive the Raychaudhuri equations
\bea
\lb{ray1}
\partial_U\left(\frac{\partial_U {\mathpzc{r}}}{\Upomega^2}\right)=0,\\
\lb{ray2}
\partial_V\left(\frac{\partial_V {\mathpzc{r}}}{\Upomega^2}\right)=0,
\eea
from the above.

We can illustrate the 2-dimensional quotient spacetime $\mathcal{Q}$ as a subset of an ambient $\mathbb R^{1+1}$:
{\begin{figure}[ht]
\centering
\includegraphics[width=0.6\textwidth]{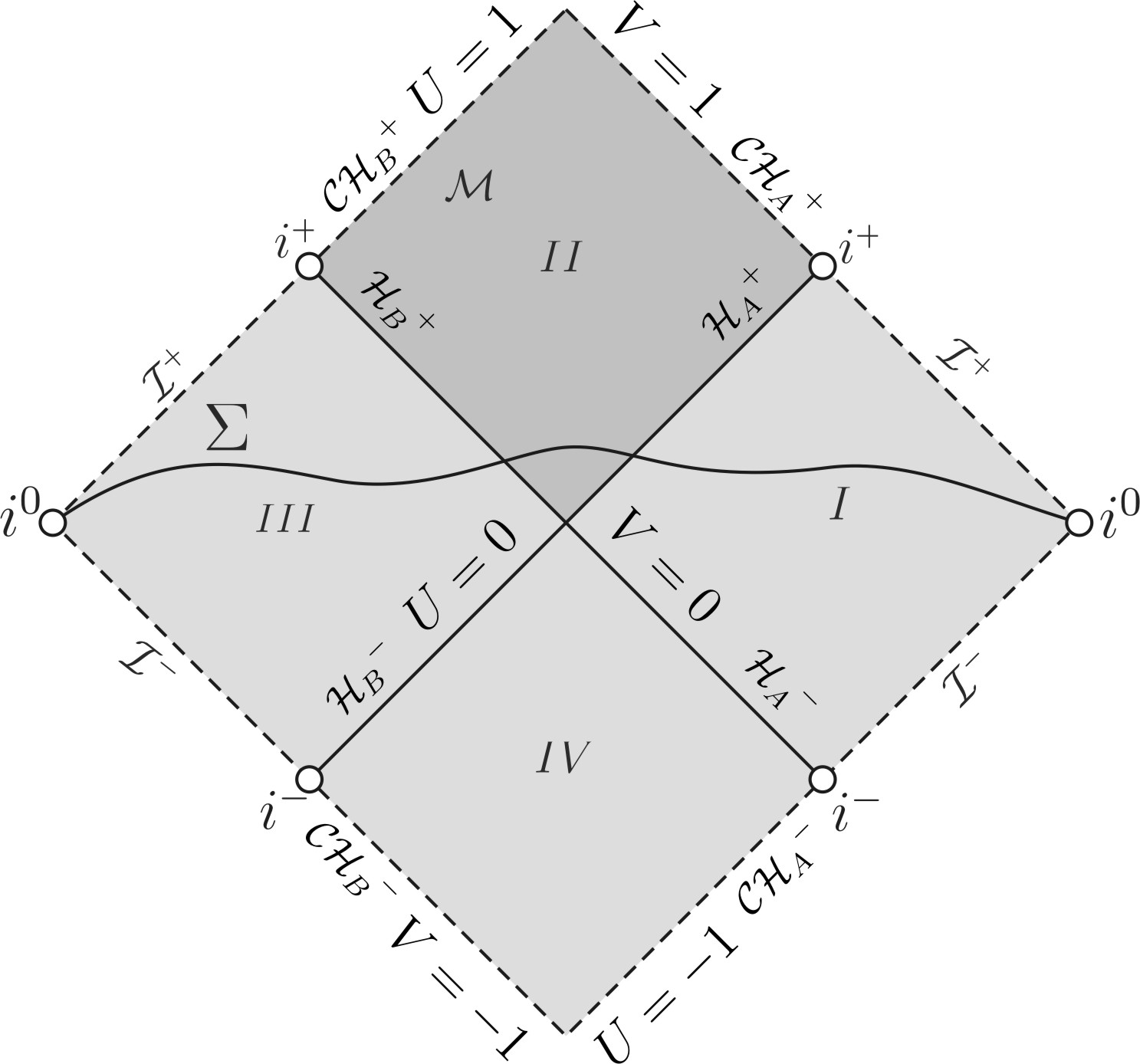}
\caption[Text der im Bilderverzeichnis auftaucht]{Conformal diagram of the maximal domain of dependence of Reissner-Nordstr\"om spacetime.}
\label{RN_ganz}\end{figure}}

Identifying $U$, $V$ with ambient null coordinates of $\mathbb R^{1+1}$, the boundary of $\cQ \subset \bbR^{1+1}$ is given by \mbox{$\pm 1\times [-1,1] \cup [-1,1]\times\pm 1$}. Let us further define the darker shaded region $II$ of Figure \ref{RN_ganz} by \mbox{$\cQ|_{II}=[0,1) \times [0,1)$}. Particularly important is \mbox{$\cC\cH^+=\cC\cH^+_A \cup \cC\cH^+_B=1\times (0,1] \cup (0,1]\times 1$}, which is the future boundary of the interior of region $II$.
We define \mbox{$\cM|_{II}=\pi^{-1}(\cQ|_{II})$}, where $\pi$ is the projection \mbox{$\pi: \cM\rightarrow \cQ$}.

\subsubsection{Eddington-Finkelstein coordinates}
It will be convenient to rescale the global double null coordinates and 
define
\bea
\lb{u_v_Edd}
u=f(U)=\frac{2r_+}{{r_+}^2-e^2}\ln\left|\ln\left|\frac{1+U}{1-U}\right|\right|,\qquad v=h(V)=\frac{2r_+}{{r_+}^2-e^2}\ln\left|\ln\left|\frac{1+V}{1-V}\right|\right|.
\eea
Note that
$u$ is the retarded and $v$ is the advanced Eddington-Finkelstein coordinate. These coordinates are both regular in the interior of $\cQ|_{II}$, cf.~ Figure \ref{RN_ganz}.
Nonetheless, we can view the whole of $\cQ|_{II}$ as
\bea
\lb{QII}
\cQ|_{II}=[-\infty, \infty) \times [-\infty,\infty), 
\eea
where we have formally parametrized by
\ben
{{\cH_A}^+}&=&\left\{-\infty\right\}\times [-\infty, \infty),\\ 
{{\cH_B}^+}&=&[-\infty, \infty)\times\left\{-\infty\right\}, 
\een 
as depicted in Figure \ref{range}, see also \eqref{u_v_Edd}.
{\begin{figure}[ht]
\centering
\includegraphics[width=0.4\textwidth]{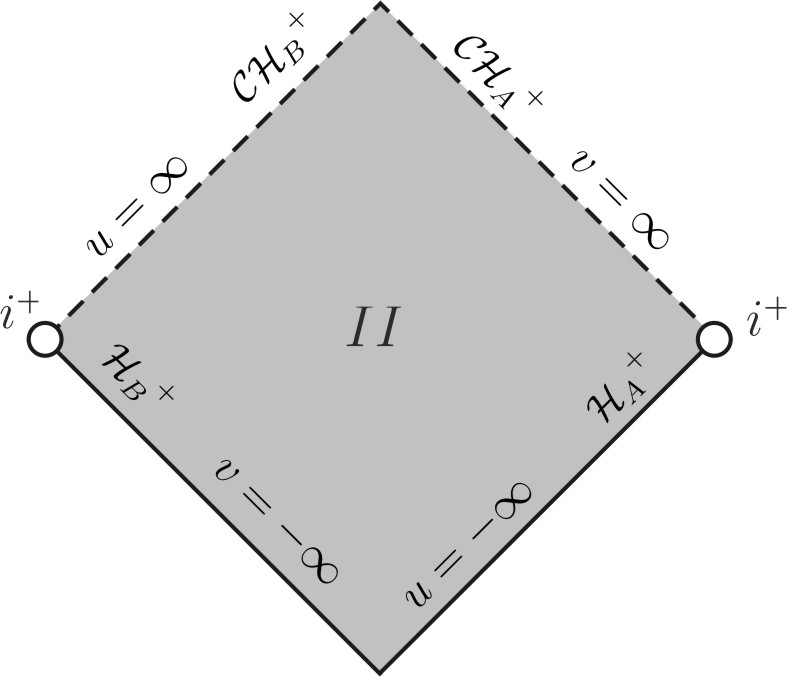}
\caption[Text der im Bilderverzeichnis auftaucht]{Conformal diagram of darker shaded region $II$, compare Figure \ref{RN_ganz}, with the ranges of $(u,v)$ depicted.}
\label{range}\end{figure}}\\

In $u$, $v$ coordinates the metric is given by 
\bea
g=-\Omega^2(u,v)\md u \md v+r^2(u,v)\left[\md \theta^2+\sin^2\theta \md \varphi^2\right],
\eea
with 
\bea
\lb{omeganull}
{\Omega}^2(u,v)=\frac{\Upomega^2(U,V)}{\partial_U f \partial_V h}=-\left(1-\frac{2M}{r}+\frac{e^2}{r^2}\right),
\eea
where the unfamiliar minus sign on the right hand side arises since all definitions have been made suitable for the interior.
We will often make use of the fact that by the choice of Eddington-Finkelstein coordinates \eqref{u_v_Edd} for the interior we have scaled our coordinates such that
\bea
\lb{def_l_n}
\frac{\partial_u r}{\Omega^2}=-\frac12, \qquad \frac{\partial_v r}{\Omega^2}=-\frac12.
\eea
(The fact that the above expressions are constants follows from the Raychaudhuri equations \eqref{ray1} and \eqref{ray2}.)
Taking the derivatives of \eqref{omeganull} with respect to $u$ and $v$ and using \eqref{def_l_n} it follows that
\bea
\lb{u-neg}
\frac{\partial_u \Omega}{\Omega}(u,v)&=&\frac{1}{2r^2}\left(M-\frac{e^2}{r}\right),\\
\lb{v-neg}
\frac{\partial_v \Omega}{\Omega}(u,v)&=&\frac{1}{2r^2}\left(M-\frac{e^2}{r}\right).
\eea

\subsection{Further properties of Reissner-Nordstr\"om geometry}
\lb{geometry}
\subsubsection{$(t,r^{\star})$ and $(t,r)$ coordinates}
\lb{rtcoords}
It is useful to define the function \mbox{$t:\mathring{\cM}|_{II} \rightarrow \bbR$} by
\bea
t(u,v)=\frac{v-u}{2},
\eea
where $\mathring{\cM}|_{II}=\cM|_{II}\setminus \partial \cM|_{II}$ is the interior of $\cM|_{II}$.
Moreover, we define the function \mbox{$r^{\star}:\mathring{\cM}|_{II}\rightarrow \bbR$} by
\bea
\lb{regge}
r^{\star}(u,v)=\frac{v+u}{2},
\eea
where $r^{\star}$ is usually referred to as the Regge-Wheeler coordinate. 
Note that for coordinates $(t, r^{\star}, \varphi, \theta)$ defined in $\mathring{\cM}|_{II}$ we have that $\frac{\partial}{
\partial t}$ is a spacelike Killing vector field which extends to the globally defined Killing vector field $T$ on $\cM$.
By $\varphi_{\tau}$ we denote a 1-parameter group of diffeomorphisms generated by the Killing field $T$. 
We can moreover relate the functions $r$ and $r^{\star}$ by
\bea
\lb{rstar1}
\md r^{\star}&=&\frac{\md r}{1-\frac{2M}{r}+\frac{e^2}{r^2}}\\ 
\lb{rstar2}
\Rightarrow r^{\star}&=&r+\frac{1}{\kappa_+}\ln{|\frac{r-r_+}{r_+}|}+\frac{1}{\kappa_-}\ln{|\frac{r-r_-}{r_-}|}+C, 
\eea
where $C$ is constant which is implicitly fixed by previous definitions,
\bea
r_-=M-\sqrt{M^2-e^2},
\eea
 and the surface gravities are given by
\bea
\lb{kappa}
\kappa_{\pm}=\frac{r_{\pm}-r_{\mp}}{2r_{\pm}^2}.
\eea
Note that $\kappa_{+}$ is the surface gravity at $\cH^+$ and $\kappa_{-}$ is the surface gravity at $\cC\cH^+$.
The function $r(u,v)$ extends continuously and is monotonically decreasing in both $u$ and $v$ towards $\cC\cH^+$ such that we have 
\bea
r(u, \infty)&=&{r}|_{{\cC\cH_A}^+}=r_-,\\
r(\infty, v)&=&{r}|_{{\cC\cH_B}^+}=r_-.
\eea

\subsubsection{Angular momentum operators}
\lb{angular}
We have already mentioned the generators of spherical symmetry $\leo_i$, $i=1,2,3$, in Section \ref{en_cur}. They are explicitly given by
\bea
\lb{angmom1}
\leo_{1}&=& \sin \varphi \partial_{\theta} + \cot \theta \cos \varphi \partial_{\varphi},\\
\leo_{2}&=& -\cos \varphi \partial_{\theta} + \cot \theta \sin \varphi \partial_{\varphi},\\
\lb{angmom2}
\leo_{3}&=& - \partial_{\varphi},
\eea
which satisfy
\bea
\lb{sum1}
\sum_{i=1}^{3} \left(\leo_i \phi\right)^2&=&r^2|\nabb \phi|^2,\\
\lb{sum2}
\sum_{i=1}^{3} \sum_{j=1}^{3}\left(\leo_i \leo_j\phi\right)^2&=&r^4|\nabb^2 \phi|^2,
\eea
where we define
\bea
\lb{nabb}
|\nabb \phi|^2=\frac{1}{r^2}\left[(\partial_{\theta} \phi)^2+ \frac{1}{\sin^2 {\theta}}(\partial_{\varphi} \phi)^2\right].
\eea

\subsubsection{The redshift, noshift and blueshift region}
\lb{bnrsection}
As we have already mentioned in the introduction, in the interior we can distinguish 
\bea
\mbox{redshift} &\mbox{$\cR=\left\{r_{red}\leq r \leq r_+\right\}$}&,\\ 
\lb{noshiftdef}
\mbox{noshift} &\mbox{$\cN=\left\{r_{blue}\leq r\leq r_{red}\right\}$}&,\\
\mbox{and blueshift} &\mbox{$\cB=\left\{r_{-}\leq r\leq r_{blue}\right\}$}&
\eea
subregions, as shown in Figure \ref{IIbnr}, for values $r_{red}$, $r_{blue}$ to be defined immediately below.
{\begin{figure}[ht]
\centering
\includegraphics[width=0.4\textwidth]{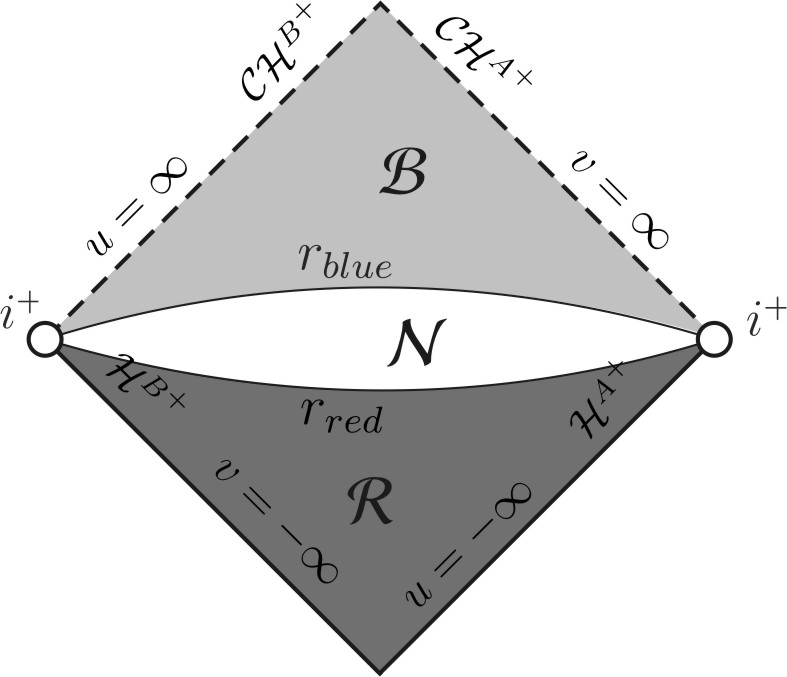}
\caption[]{Region $II$ with distinction into redshift $\cR$, noshift $\cN$ and blueshift $\cB$ regions.}
\label{IIbnr}\end{figure}}

In the redshift region $\cR$ we make use of the fact that the surface gravity $\kappa_+$ of the event horizon is positive. The region is then characterized by the fact that there exists a vector field $N$ such that its associated current $J^N_{\mu}n^{\mu}_{v=const}$ on a $v=const$ hypersurface can be controlled by the related bulk term $K^N$, cf.~ Proposition \ref{mi}. This positivity of the bulk term $K^N$ is only possible sufficiently close to $\cH^+$. In particular we shall define 
\bea
\lb{rred}
r_{red}=r_+-\epsilon,
\eea 
with $\epsilon>0$ and small enough such that Proposition \ref{mi} is applicable. (Furthermore, note that the quantity $M-\frac{e^2}{r}$ is always positive in $\cR$.)

As defined in \eqref{noshiftdef} the $r$ coordinate in the noshift region $\cN$ ranges between $r_{red}$ defined by \eqref{rred} and $r_{blue}$, defined below, strictly bigger than $r_-$. 
In $\cN$ we exploit the fact that $J^{-\partial_r}$ and $K^{-\partial_r}$ are invariant under translations along $\partial_t$. For that reason we can uniformly control the bulk by the current along a constant $r$ hypersurface. This will be explained further in Section \ref{zweite_region}.

The blueshift region $\cB$ is characterized by the fact that the bulk term $K^{S_0}$ associated to the vector field $S_0$ to be defined in \eqref{Nstar} is positive. We define 
\bea
\lb{rblue}
r_{blue}=r_-+\tilde{\epsilon},
\eea
with $\tilde{\epsilon}>0$ for an $\tilde{\epsilon}$ such that $M-\frac{e^2}{r}$ carries a negative sign
and such that (for convenience) 
\bea
\lb{rchoice}
r^{\star}(r_{blue})>0.
\eea
In particular, in view of \eqref{u-neg} and \eqref{v-neg} 
for $\tilde{\epsilon}$ sufficiently small the following lower bound holds in $\cB$
\bea
\lb{lowerboundu}
0<{\beta} &\leq&  -\frac{\partial_u \Omega}{\Omega},\\
\lb{lowerboundv}
0<{\beta} &\leq& -\frac{\partial_v \Omega}{\Omega},
\eea
with $\beta$ a positive constant.

\subsection{Notation}
\lb{nota}
We will describe certain regions derived from the hypersurfaces $r=r_{red}$, $r=r_{blue}$ and in addition the hypersurface $\gamma$ which will be defined in Section \ref{gamma_curve}.
For example given the hypersurface $r=r_{red}$ and the hypersurface $u=\tilde{u}$ we define the $v$ value at which these two hypersurfaces intersect by a function $v_{red}(\tilde{u})$ evaluated for $\tilde{u}$. Let us therefore introduce the following notation:
\bea
\lb{notation_neu}
v_{red}(\tilde{u}) \quad &\mbox{is determined by}& \quad r(v_{red}(\tilde{u}), \tilde{u})=r_{red},\nonumber\\
v_{\gamma}(\tilde{u}) \quad &\mbox{is determined by}& \quad (v_{\gamma}(\tilde{u}), \tilde{u}) \in \gamma,\nonumber\\
v_{blue}(\tilde{u}) \quad &\mbox{is determined by}& \quad r(v_{blue}(\tilde{u}), \tilde{u})=r_{blue},\nonumber\\
\mbox{and similarly we will also use}&&\nonumber\\
u_{red}(\tilde{v}) \quad &\mbox{is determined by}& \quad r(u_{red}(\tilde{v}), \tilde{v})=r_{red},\nonumber\\
u_{\gamma}(\tilde{v}) \quad &\mbox{is determined by}& \quad (u_{\gamma}(\tilde{v}), \tilde{v}) \in \gamma,\nonumber\\
u_{blue}(\tilde{v}) \quad &\mbox{is determined by}& \quad r(u_{blue}(\tilde{v}), \tilde{v})=r_{blue}.
\eea
For a better understanding the reader may also refer to Figure \ref{integralbild} and Figure \ref{integralbild2}.
{\begin{figure}[!ht]
\centering
\includegraphics[width=0.5\textwidth]{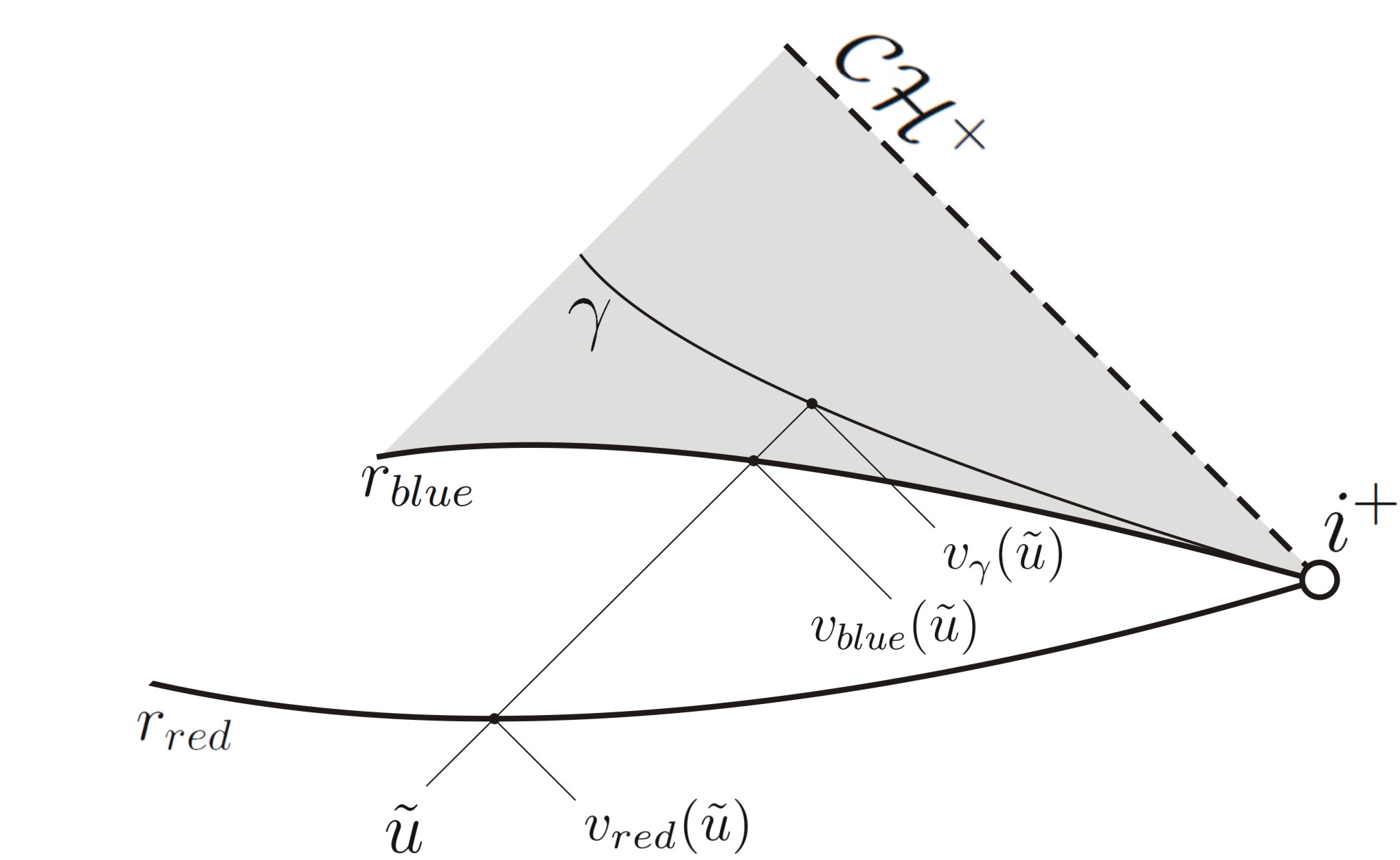}
\caption[]{Sketch of blueshift region $\cB$ with quantities depicted dependent on $\tilde{u}$.}
\label{integralbild}\end{figure}}

{\begin{figure}[!ht]
\centering
\includegraphics[width=0.5\textwidth]{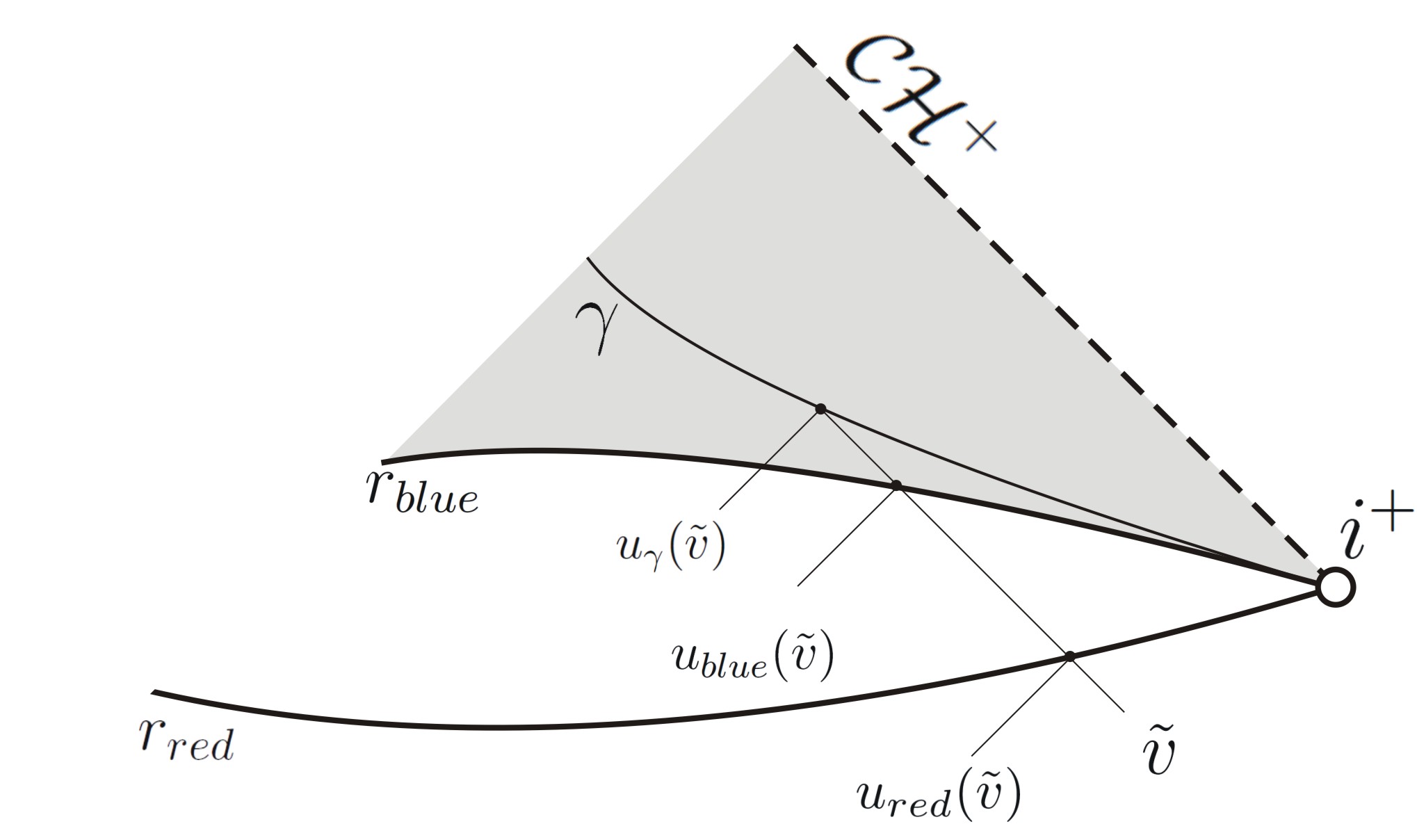}
\caption[]{Sketch of blueshift region $\cB$ with quantities depicted dependent on $\tilde{v}$.}
\label{integralbild2}\end{figure}}

Note that the above functions are well defined since $r=r_{red}$, $r=r_{blue}$ and $\gamma$ are spacelike hypersurfaces terminating at $i^+$.


\section{Horizon estimates and Cauchy stability}
\lb{horizon_estimates}

Our starting point will be previously proven decay bounds for $\phi$ and its derivatives in the black hole {\it exterior} up to and including the event horizon; in particular we can state:

\begin{thm}
\lb{anfang}
Let $\phi$ be a solution of the wave equation \eqref{wave} on a subextremal Reissner-Nordstr\"om background $(\cM,g)$, with mass $M$ and charge $e$ and $M>|e|\neq 0$, arising from smooth compactly supported initial data on an arbitrary Cauchy hypersurface $\Sigma$, cf.~ Figure \ref{RN2}. Then, there exists $\delta>0$ such that
\bea
\lb{thpur}
\int\limits_{\bbS^2}\int\limits^{v+1}_v \left[(\partial_v \phi)^2 (-\infty, v)+|\nabb \phi|^2(-\infty, v)\right]r^2\md v\md \sigma_{\mathbb S^2}&\leq& C_{0}v^{-2-2\delta},\\
\lb{thpur1}
\int\limits_{\bbS^2}\int\limits^{v+1}_v \left[(\partial_v \leo\phi)^2(-\infty, v) +|\nabb \leo\phi|^2(-\infty, v)\right]r^2\md v\md \sigma_{\mathbb S^2}&\leq& C_{1}v^{-2-2\delta},\\
\lb{thpur2}
\int\limits_{\bbS^2}\int\limits^{v+1}_v\left[ (\partial_v \leo^2\phi)^2(-\infty, v) +|\nabb \leo^2\phi|^2(-\infty, v)\right]r^2\md v\md \sigma_{\mathbb S^2}&\leq& C_{2}v^{-2-2\delta},
\eea
on ${\cH_A}^+$, for all $v$ and some positive constants $C_{0}$, $C_{1}$ and $C_{2}$ depending on the initial data.\footnote{The notation $\leo$ and $\leo^2$ is explained in Section \ref{leonotationsec} and simply denotes summation over angular momentum operators $\leo_i$ and $\leo_i\leo_j$.}
\end{thm}
\begin{proof}
The Theorem follows by putting together work of P.~ Blue and A.~ Soffer \cite{blue} on integrated local energy decay, M.~ Dafermos and I.~ Rodnianski \cite{m_price} on the redshift and V.~ Schlue \cite{volker} on improved decay using the method of \cite{m_new} in the exterior region.
The assumption of smoothness and compact support can be weakened. Moreover, we can in fact take $\delta$ arbitrarily close to $\frac12$, but $\delta>0$ is sufficient for our purposes and allows in principle for a larger class of data on $\Sigma$. 
\end{proof}

On the other hand, trivially from Cauchy stability, boundedness of the energy along the second component of the past boundary of the characteristic rectangle $\Xi$, cf.~ Section \ref{firstlook}, which we have picked to be $v=1$, can be derived. More generally we can state the following proposition.
\begin{prop}
\lb{initialdataprop}
Let \mbox{$u_{\diamond}, v_{\diamond} \in (-\infty, \infty)$}.
Under the assumption of Theorem \ref{anfang}, the energy at advanced Eddington-Finkelstein coordinate \mbox{$\left\{v=v_{\diamond}\right\}\cap\left\{{-\infty}\leq u \leq {u_{\diamond}}\right\}$}
is bounded from the initial data 
\bea
\lb{proppur}
\int\limits_{\bbS^2}\int\limits_{-\infty}^{u_{\diamond}}\left[ \Omega^{-2}(\partial_u \phi)^2(u,v_{\diamond})+\frac{\Omega^{2}}{2}|\nabb \phi|^2(u,v_{\diamond})\right]r^2\md u\md \sigma_{\mathbb S^2} &\leq& D_{0}(u_{\diamond},v_{\diamond}),\\
\lb{proppur1}
\int\limits_{\bbS^2}\int\limits_{-\infty}^{u_{\diamond}}\left[ \Omega^{-2}(\partial_u \leo\phi)^2(u,v_{\diamond})+\frac{\Omega^{2}}{2}|\nabb  \leo\phi|^2(u,v_{\diamond})\right]r^2\md u\md \sigma_{\mathbb S^2}&\leq& D_{1}(u_{\diamond},v_{\diamond}),\\
\lb{proppur2}
\int\limits_{\bbS^2}\int\limits_{-\infty}^{u_{\diamond}} \left[\Omega^{-2}(\partial_u \leo^2\phi)^2(u,v_{\diamond})+\frac{\Omega^{2}}{2}|\nabb \leo^2\phi|^2(u,v_{\diamond})\right]r^2\md u\md \sigma_{\mathbb S^2}&\leq& D_{2}(u_{\diamond},v_{\diamond}),
\eea
and further
\bea
\lb{aufeins1}
\sup_{-\infty\leq u \leq {u_{\diamond}}}\int\limits_{\bbS^2} (\phi)^2(u,v_{\diamond})\md \sigma_{\mathbb S^2}&\leq&  D_{0}(u_{\diamond},v_{\diamond}),\\
\sup_{-\infty\leq u \leq {u_{\diamond}}}\int\limits_{\bbS^2} (\leo\phi)^2(u,v_{\diamond})\md \sigma_{\mathbb S^2}&\leq&  D_{1}(u_{\diamond},v_{\diamond}),\\
\lb{aufeins3}
\sup_{-\infty\leq u \leq {u_{\diamond}}}\int\limits_{\bbS^2} (\leo^2\phi)^2(u,v_{\diamond})\md \sigma_{\mathbb S^2}&\leq&  D_{2}(u_{\diamond},v_{\diamond}),
\eea
with $D_{0}(u_{\diamond},v_{\diamond})$, $D_{1}(u_{\diamond},v_{\diamond})$ and $D_{2}(u_{\diamond},v_{\diamond})$ positive constants depending on the initial data on $\Sigma$.
\end{prop}
\begin{proof}
This follows immediately from local energy estimates in a compact spacetime region. Note the $\Omega^{-2}$ and $\Omega^{2}$ weights which arise since $u$ is not regular at $\cH^+_A$.
\end{proof}

\section{Statement of the theorem and outline of the proof in the neighbourhood of $i^+$}
\lb{outline}
The most difficult result of this paper can now be stated in the following theorem.
\begin{thm}
\lb{dashier}
On subextremal Reissner-Nordstr\"om spacetime with $M>|e|\neq 0$, let $\phi$ be as in Theorem \ref{anfang}, then
\ben
|\phi|\leq C
\een
locally in the black hole interior up to
$\cC\cH^+$ in a ``small neighbourhood'' of timelike infinity $i^+$,
that is in \mbox{$(-\infty, u_{\schere}] \times [1, \infty)$} for some $u_{\schere}>-\infty$.
\end{thm}

{\em Remark.} We will see that $C$ depends only on the initial data. 
{\begin{figure}[ht]
\centering
\includegraphics[width=0.5\textwidth]{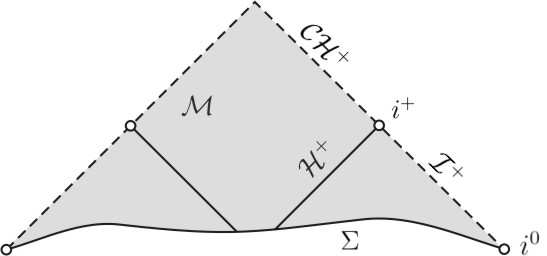}
\caption[]{Maximal development of Cauchy hypersurface $\Sigma$ in Reissner-Nordstr\"om spacetime $(\cM,g)$.}
\label{RN2}\end{figure}}

We will consider a characteristic rectangle $\Xi$ extending from ${\cH_A}^+$ as shown in Figure \ref{RN_character}. 
We pick the characteristic rectangle to be defined by \mbox{$\Xi=\left\{(-\infty\leq u\leq u_{\schere}), (1\leq v <\infty)\right\}$}, where $u_{\schere}$ is sufficiently close to $-\infty$ for reasons that will become clear later on, cf.~ Proposition \ref{kastle}.
As described in Section \ref{horizon_estimates}, from bounds of data on $\Sigma$ 
bounds on the solution on the lower segments follow according to Theorem \ref{anfang} and Proposition \ref{initialdataprop}.
{\begin{figure}[ht]
\centering
\includegraphics[width=0.6\textwidth]{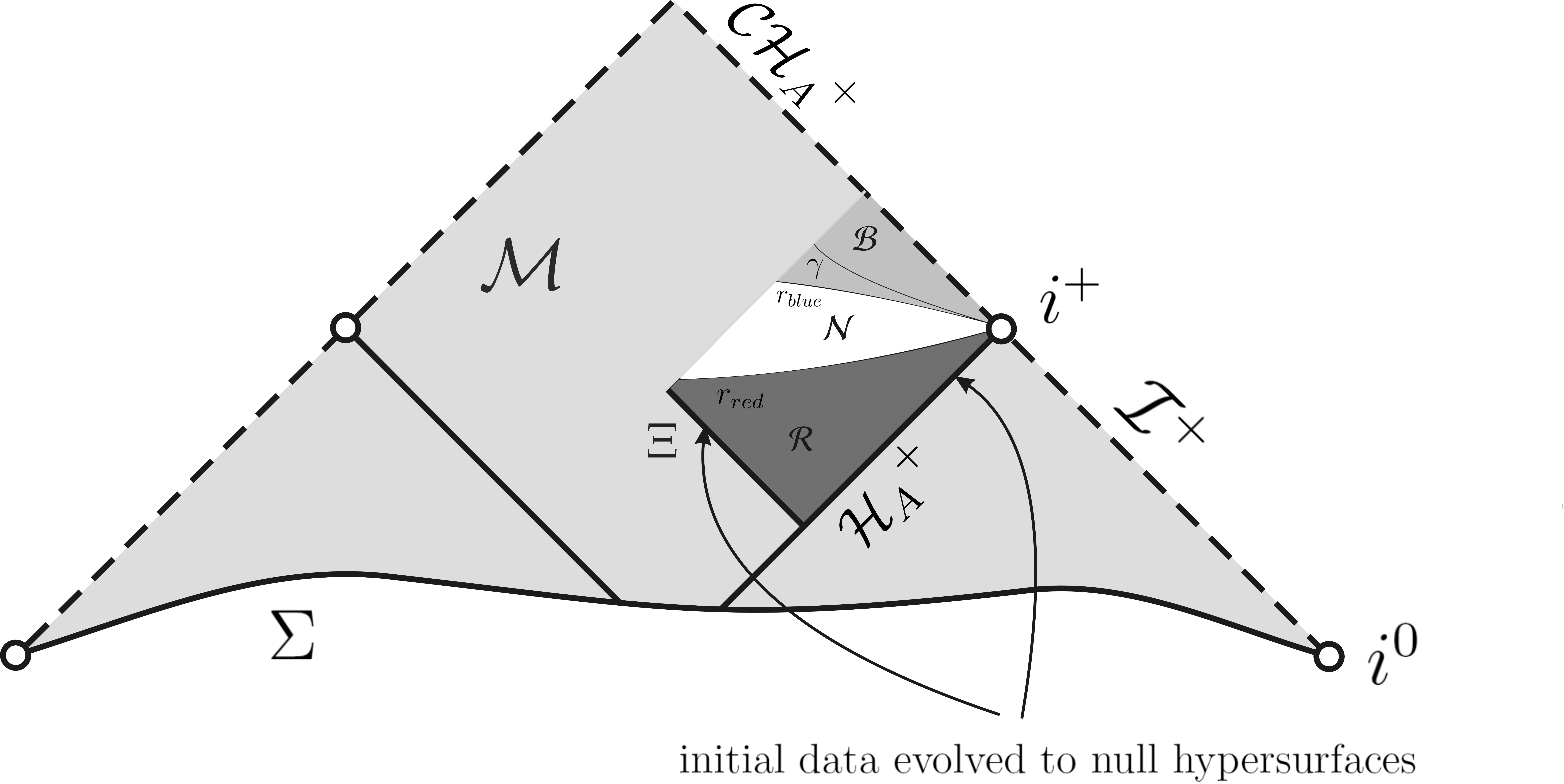}
\caption[]{Characteristic rectangle $\Xi$ in the interior of Reissner-Nordstr\"om spacetime $(\cM,g)$, for $\Xi$ zoomed in see Figure \ref{bnr}.}
\label{RN_character}\end{figure}} 

In order to prove Theorem \ref{dashier} 
we distinguish the redshift $\cR$, the noshift $\cN$ and the blueshift $\cB$ region, with the properties as explained in Section \ref{bnrsection}, 
cf.~ Figure \ref{bnr}.
{\begin{figure}[ht]
\centering
\includegraphics[width=0.4\textwidth]{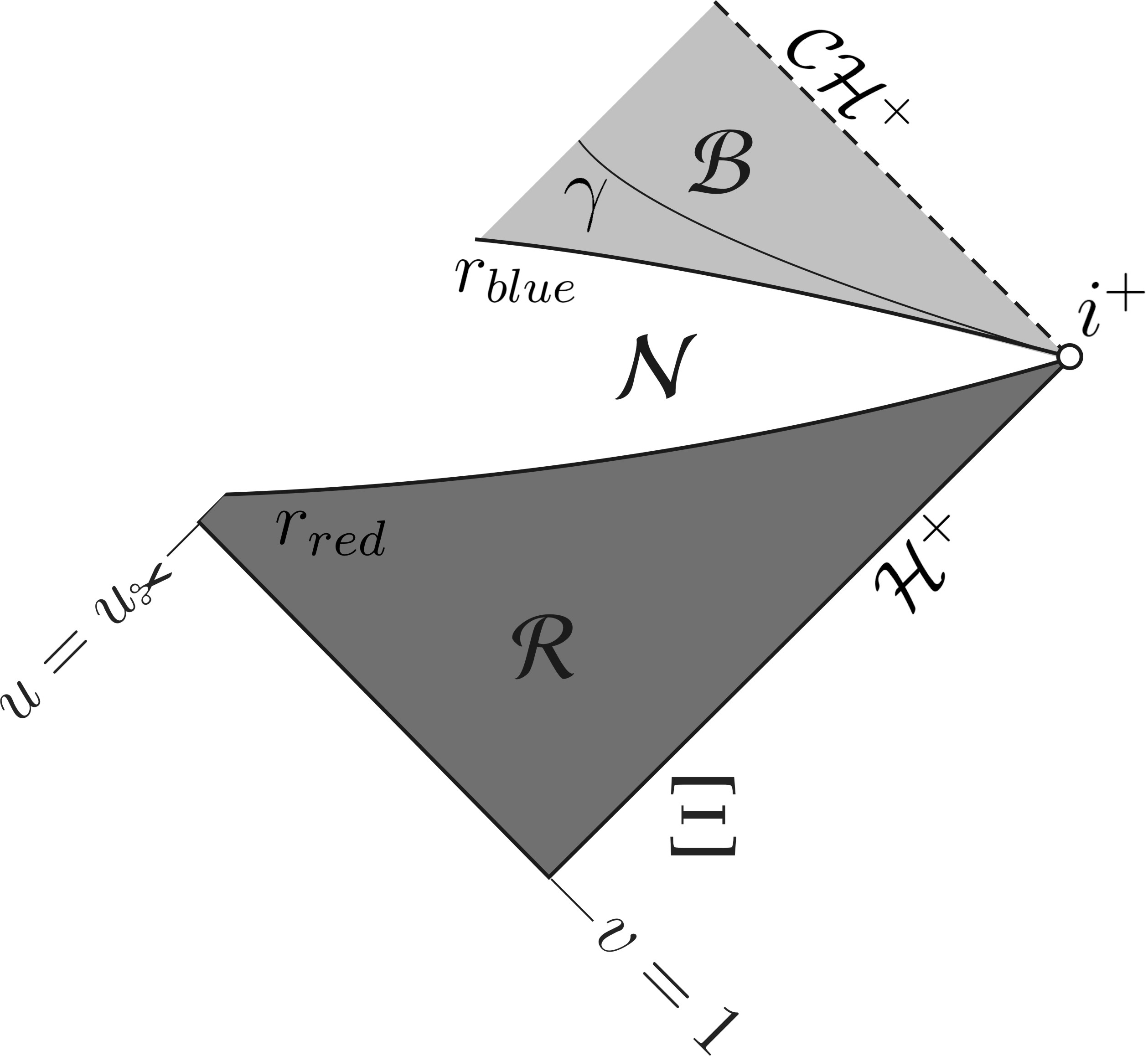}
\caption[]{Characteristic rectangle $\Xi$ with redshift $\cR$, noshift $\cN$ and blueshift $\cB$ regions.}
\label{bnr}\end{figure}}
This distinction is made since different vector fields have to be employed in the different regions\footnote{The reader may wonder why the noshift region $\cN$ is introduced instead of just separating the red- and the blueshift regions along the $r$ hypersurface whose value renders the quantity $M-\frac{e^2}{r}$ equal zero. This was to ensure strict positivity/negativity of the quantity in the redshift/blueshift region.}. 

In the redshift region $\cR$ we will make use of the redshift vector field $N$ of \cite{m_lec} on which we will elaborate more in Section \ref{first_section}.
Proposition \ref{mi} gives the positivity of the bulk $K^N$ which thus bounds the current $J^N_{\mu}n^{\mu}_{v=const}$ from above.
Applying the divergence theorem, decay up to $r=r_{red}$ will be proven.

In the noshift region $\cN$ we can simply appeal to the fact that the future directed timelike vector field $-\partial_r$ is invariant under the flow of the spacelike Killing vector field $\partial_t$.
It is for that reason that the bulk term $K^{-\partial_r}$
can be uniformly controlled by the energy flux $J_{\mu}^{-\partial_r}n^{\mu}_{r=\bar{r}}$ through the $r=\bar{r}$ hypersurface.
Decay up to $r=r_{blue}$ will be proven by making use of this together with the uniform boundedness of the $v$ length of $\cN$.

To understand the blueshift region $\cB$, we will partition it by the hypersurface $\gamma$ admitting logarithmic distance in $v$ from $r=r_{blue}$, cf.~ Section \ref{gamma_curve}. We will then separately consider the region to the past of $\gamma$, \mbox{$J^-(\gamma)\cap \cB$} and the region to the future of $\gamma$, \mbox{$J^+(\gamma)\cap \cB$}.
The region to the future of $\gamma$ is characterized by good decay bounds on $\Omega^2$
(implying for instance that the
spacetime volume is finite, $\operatorname{Vol}(J^+(\gamma))<C$).

In \mbox{$J^-(\gamma)\cap \cB$} we use a vector field 
\bea
\lb{Nstar}
{S_0}=r^q\partial_{r^{\star}}=r^q(\partial_u+\partial_v),
\eea
where $q$ is sufficiently large, cf.~ Section \ref{blueshift1}. We will see that for the right choice of $q$ we can render the associated bulk term $K^{{S_0}}$ positive which is the ``good'' sign when using the divergence theorem.

In order to complete the proof, we consider finally the region \mbox{$J^+(\gamma)\cap \cB$} and propagate the decay further from the hypersurface $\gamma$ up to the Cauchy horizon in a neighbourhood of $i^+$. For this, we introduce a new timelike vector field ${S}$ defined by
\bea
\lb{N1}
{S}=|u|^p\partial_u+v^p\partial_v,
\eea
for an arbitrary $p$ such that
\bea
\lb{waspist}
1<p\leq 1+2\delta,
\eea
where $\delta$ is as in Theorem \ref{anfang}.
We use pointwise estimates on $\Omega^2$ in $J^+(\gamma)$ as a crucial step, cf.~ Section \ref{finiteness}. 

Putting everything together, in view of the geometry and the weights of $S$, we finally obtain for all $v_*\geq 1$
\bea
\lb{fluxscetch}
\int\limits_{\bbS^2} \int\limits_1^{v_*} v^p(\partial_v \phi)^2 r^2\md v \md \sigma_{\mathbb S^2}\leq \mbox{Data},
\eea
for the weighted flux. 
Using the above, the uniform boundedness for $\phi$ stated in Theorem \ref{dashier} then follows from an argument that can be sketched as follows. 

Let us first see how we get an integrated bound on the spheres of symmetry.
By the fundamental theorem of calculus and the Cauchy-Schwarz inequality one obtains
\bea
\int\limits_{\bbS^2} \phi^2(u, v_*, \theta, \varphi)\md \sigma_{\mathbb S^2}
&\leq& C\int\limits_{\bbS^2}\left(\int\limits_1^{v_*} v^p(\partial_v \phi)^2\md v\right)\left(\int\limits_1^{v_*} v^{-p}\md v\right)r^2\md \sigma_{\mathbb S^2}+\mbox{data},\nonumber
\eea
where the first factor of the first term is controlled by \eqref{fluxscetch}.
Therefore, we further get
\bea
\lb{fundcauchy11}
\int\limits_{\bbS^2} \phi^2\md \sigma_{\mathbb S^2}&\stackrel{\eqref{fluxscetch}}{\leq}&\mbox{Data}\int\limits_{\bbS^2}\int\limits_1^{v_*} v^{-p}\md v\md \sigma_{\mathbb S^2}+\mbox{data}\nonumber\\
&\leq&\mbox{Data}+\mbox{data},
\eea 
where we have used \mbox{$\int\limits^{\infty}_{1} v^{-p}\md v< \infty$} which followed from the first inequality of \eqref{waspist}.

Obtaining a pointwise statement from the above will be achieved by commuting \eqref{wave} with symmetries as well as applying Sobolev embedding.
As outlined in Section \ref{en_cur} in Reissner-Nordstr\"om geometry we have \mbox{$\Box_g \leo_i \phi=0$}, where $\leo_i$ with $i=1,2,3$ are the 3 spacelike Killing vector fields resulting from the spherical symmetry. 
Thus one obtains the analogue of \eqref{fundcauchy11} but with $\leo_i \phi$ and \mbox{$\leo_i \leo_j \phi$} in place of $\phi$.
Using Sobolev embedding on $\bbS^2$ thus leads immediately to the desired bounds. See Section \ref{uni_bounded}.
This will close the proof of Theorem \ref{dashier}.\\

\section{Propagating through $\cR$ from $\cH^+$ to $r=r_{red}$}
\lb{first_section}

The estimates in this and the following section are motivated by work of Luk \cite{luk}.
He proves that any polynomial decay estimate that holds along the event horizon of Schwarzschild black holes can be propagated to any constant $r$ hypersurface in the black hole interior.
This followed a previous spherically symmetric argument of \cite{m_interior}. See also Dyatlov \cite{dyatlov}.

As outlined in Section \ref{horizon_estimates}, we will first propagate energy decay from $\cH^+$ up to the $r=r_{red}$ hypersurface. 

The rough idea can be understood with the help of Figure \ref{tilde_r_decay}.  
\begin{figure}[ht]
\centering
\includegraphics[width=0.6\textwidth]{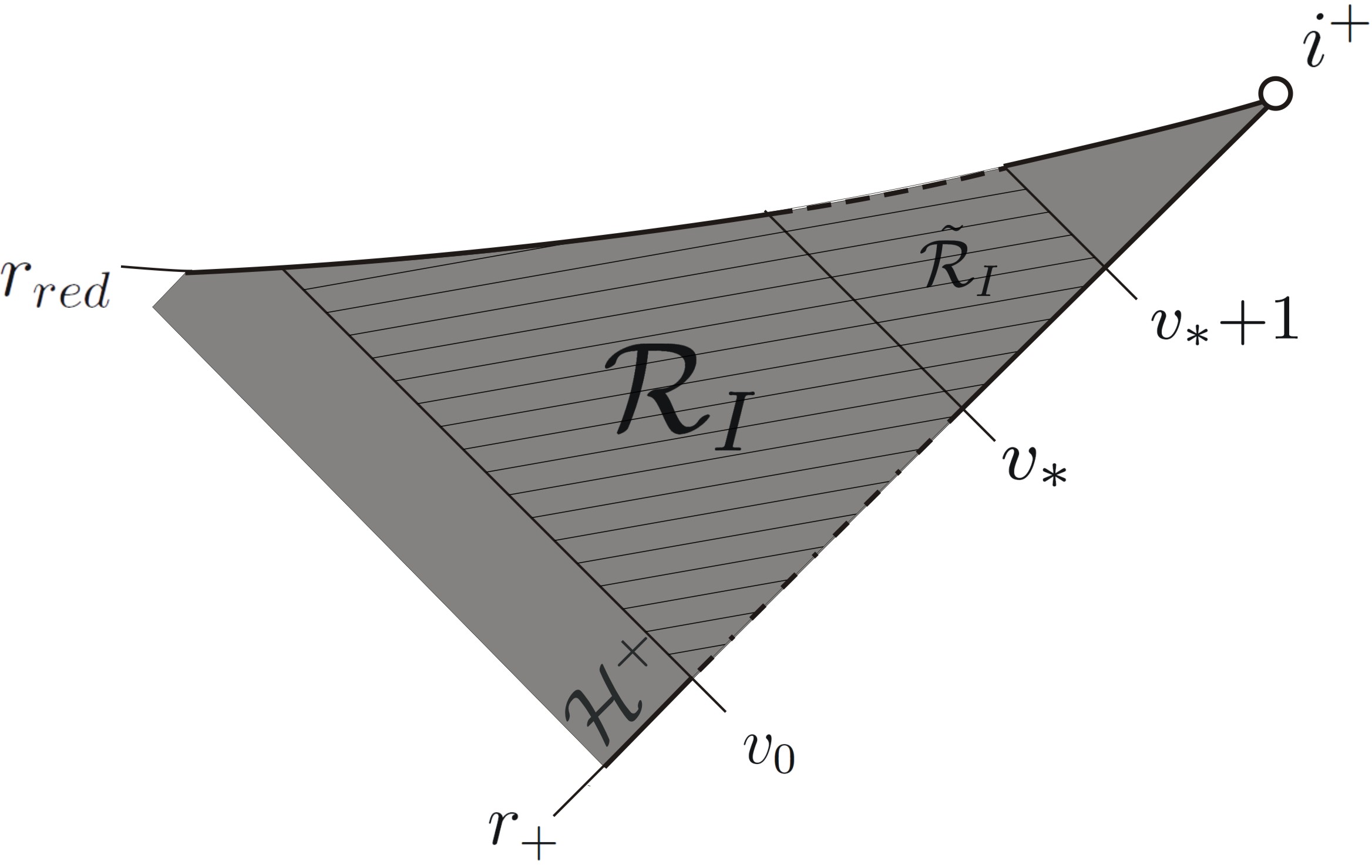}
\caption[Text der im Bilderverzeichnis auftaucht]{Regions $\cR_I$ and $\tilde{\cR_I}$.}
\label{tilde_r_decay}\end{figure}
By Theorem \ref{anfang} we are given energy decay on the event horizon $\cH^+$, see dash-dotted line.
By using the energy identity for the vector field $N$ in region \mbox{$\cR_I=\left\{r_{red}\leq r \leq r_+\right\}\cap \left\{1<v\leq v_*\right\}$}, the coarea formula etc., we obtain decay of the flux through constant $v$ hypersurfaces throughout the entire region. Using this result and considering the energy identity once again in region \mbox{$\tilde{\cR}_I=\left\{r_{red}\leq r \leq r_+\right\}\cap \left\{v_*\leq v\leq v_*+1\right\}$} we eventually obtain decay on the $r=r_{red}$ hypersurface, note the dashed line.

The redshift vector field was already introduced by Dafermos and Rodnianski in \cite{m_red} and elaborated on again in \cite{m_lec}. 
The existence of such a vector field in the neighbourhood of a Killing horizon $\cH^+$ depends only on the positivity of the surface gravity, in this case $\kappa_+$. Thus by \eqref{kappa} the following proposition follows by Theorem 7.1 of \cite{m_lec}.
\begin{prop}
\lb{mi}
(M. Dafermos and I. Rodnianski) For $r_{red}$ sufficiently close to $r_+$ there exists a $\varphi_{\tau}$-invariant\footnote{cf.~ Section \ref{geometry}} smooth future directed timelike vector field $N$ on \mbox{$\left\{r_{red}\leq r \leq r_+\right\}\cap \left\{v\geq 1\right\}$} and a positive constant $b_1$ such that 
\bea
\lb{energy_controll}
b_1J_{\mu}^N(\phi) n^{\mu}_{{v}} \leq K^N(\phi) ,
\eea
for all solutions $\phi$ of $\Box_g \phi=0$.
\end{prop}
The decay bound along $r=r_{red}$ can now be stated in the following proposition.
\begin{prop}
\label{rtildedecay}
Let $\phi$ be as in Theorem \ref{anfang}. 
Then, for all $\tilde{r} \in [r_{red}, r_+)$, with $r_{red}$ as in Proposition \ref{mi} and for all $v_*>1$,
\ben
\lb{decay_*}
\int\limits_ {\left\lbrace  v_* \leq v \leq v_*+1\right\rbrace } J_{\mu}^N(\phi) n^{\mu}_{r=\tilde{r}} \dV_{r=\tilde{r}} \leq C v_*^{-2-2\delta} , 
\een
with $C$ depending on $C_{0}$ of Theorem \ref{anfang} and $D_{0}(u_{\diamond}, 1)$ of Proposition \ref{initialdataprop}, where $u_{\diamond}$ is defined by $r_{red}=r(u_{\diamond},1)$.
\end{prop}

{\em Remark 1.} The decay in Proposition \ref{rtildedecay} matches the decay on $\cH^+$ of Theorem \ref{anfang}.

{\em Remark 2.} 
$n^{\mu}_{r=r_{red}}$ denotes the normal to the $r=r_{red}$ hypersurface oriented according to Lorentzian geometry convention. $\dV$ denotes the volume element over the entire spacetime region and $\dV_{r=r_{red}}$ denotes the volume element on the $r=r_{red}$ hypersurface. Similarly for all other subscripts.\footnote{Refer to Appendix \ref{Jcurrents} for further discussion of the volume elements.}

\begin{proof}
 Applying the divergence theorem, see e.g.~ \cite{m_lec} or \cite{taylor},
in region \mbox{$\cR_I=\left\{r_{red}\leq r \leq r_+\right\}\cap \left\{v_0\leq v\leq v_*\right\}$}, with $v_0\geq 1$, we obtain
\ben
& &\int\limits_ {\cR_I}  K^N(\phi) \dV
+ \int\limits_ {\left\lbrace  v_0 \leq v \leq v_*\right\rbrace }  J_{\mu}^N(\phi) n^{\mu}_{r=r_{red}} \dV_{r=r_{red}}
+ \int\limits_ {\left\lbrace  r_{red}\leq r \leq r_+\right\rbrace }  J_{\mu}^N(\phi) n^{\mu}_{v=v_*} \dV_{v=v_*} \\
&=&
 \int\limits_ {\left\lbrace  r_{red}\leq r \leq r_+\right\rbrace }  J_{\mu}^N(\phi) n^{\mu}_{v=v_0} \dV_{v=v_0}
+\int\limits_ {\left\lbrace  v_0 \leq v \leq v_*\right\rbrace }  J_{\mu}^N(\phi) n^{\mu}_{\cH^+} \dV_{\cH^+}.
\een
We immediately see that the second term on the left hand side is positive since $r=r_{red}$ is a spacelike hypersurface and $N$ is a timelike vector field. Therefore, we write
\bea
\lb{R_I}
& & \int\limits_ {\cR_I}  K^N(\phi) \dV
+ \int\limits_ {\left\lbrace  r_{red}\leq r \leq r_+\right\rbrace }  J_{\mu}^N(\phi) n^{\mu}_{v=v_*} \dV_{v=v_*} \nonumber \\
&\leq&
 \int\limits_ {\left\lbrace  r_{red}\leq r \leq r_+\right\rbrace }  J_{\mu}^N(\phi) n^{\mu}_{v=v_0} \dV_{v=v_0}
+\int\limits_ {\left\lbrace  v_0 \leq v \leq v_*\right\rbrace }  J_{\mu}^N(\phi) n^{\mu}_{\cH^+} \dV_{\cH^+}. 
\eea
By Theorem \ref{anfang} we have
\ben
\int\limits_ {\left\lbrace  v_0 \leq v \leq v_*\right\rbrace }   J_{\mu}^N(\phi) n^{\mu}_{\cH^+} \dV_{\cH^+} \leq C_{0} \mbox{max}\left\lbrace v_*-v_0,1\right\rbrace v_0^{-2-2\delta}.
\een
Using that the energy current associated to the timelike vector field $N$ is controlled by the deformation $K^N$ as shown in \eqref{energy_controll}
and substituting 
\bea
\lb{def_E}
E(\phi;\tilde{v})=\int\limits_ {\left\lbrace  r_{red}\leq r \leq r_+\right\rbrace }  J_{\mu}^N(\phi) n^{\mu}_{v=\tilde{v}} \dV_{v=\tilde{v}}
\eea
into \eqref{R_I} as well as using the coarea formula
\bea
\lb{coaerea2}
\int\limits_ {{\cR_I} }  J_{\mu}^N(\phi) n^{\mu}_{v=\tilde{v}} \dV
\sim \int\limits_{v_0}^{v_*} \int\limits_ {\left\lbrace  r_{red}\leq r \leq r_+\right\rbrace }  J_{\mu}^N(\phi) n^{\mu}_{v=\bar{v}} \dV_{v=\bar{v}} \md {\bar{v}},
\eea 
for the bulk term,\footnote{where $f\sim g$ means that there exist constants $0<b<B$ with $bf<g<Bf$} we obtain for all $v_0\geq1$ and $v_*>v_0$, the relation
\bea
\lb{decay_statement}
E(\phi;v_*) + \tilde{b}_1\int\limits_{v_0}^{v_*} E(\phi;{\bar{v}})\md {\bar{v}} \leq  E(\phi;v_0)+C_{0} \mbox{max}\left\lbrace v_*-v_0,1\right\rbrace v_0^{-2-2\delta}.
\eea
Note by Proposition \ref{initialdataprop}, applied to $u_{\diamond}$ defined through the relation $r_{red}=r(u_{\diamond},1)$, we have
\bea
\lb{estineu}
E(\phi;1) \leq C D_0(u_{\diamond},1),
\eea
since the vector field $N$ is regular at $\cH^+$ and thus $E(\phi;1)$ is comparable to the left hand side of \eqref{proppur}.
In order to obtain estimates from \eqref{decay_statement} we appeal to the following lemma.
\begin{lem}
Let $f:[1, \infty)\rightarrow \bbR^+$, 
\lb{decay_lemma2}
\bea
\lb{starter3}
f(t) + b\int\limits_{\tilde{t}}^{t} f(\bar{t})\md {\bar{t}} \leq  f(\tilde{t}) +C_0(t-\tilde{t}+1)\tilde{t}^{-\tilde{p}},
\eea
for all $\tilde{t}\geq 1$, where $C_0$, $\tilde{p}$ are positive constants.
Then for any $t\geq 1$ we have
\bea
\lb{ftwant}
f(t) \leq \tilde{C}t^{-\tilde{p}},
\eea
where $\tilde{C}$ depends only on $f(1)$, $b$ and $C_0$.
\end{lem}
\begin{proof}
For $t>t_0$, we will show \eqref{ftwant} by a continuity argument. It suffices to show that 
\bea
\lb{ft}
f(\tilde{t})\leq 2\tilde{C} \tilde{t}^{-\tilde{p}}, \quad \mbox{for $\tilde{t}\leq t$},
\eea
leads to 
\bea
\lb{wantedft}
\Rightarrow f(\tilde{t}) \leq \tilde{C}\tilde{t}^{-\tilde{p}}, \quad \mbox{for $\tilde{t}\leq t$},
\eea
for some large enough constant $\tilde{C}$.

We note first that given any $t_0$, from \eqref{starter3} we obtain, \mbox{$\forall  \, 1\leq t\leq t_0$}
\bea
\lb{startert0}
f(t) &\leq& f(1) +C_0\, t ,\nonumber\\
 &\leq& \left[ f(1){t_0}^{\tilde{p}} +C_0\, {t_0}^{+\tilde{p}+1}\right]{t}^{-\tilde{p}}.
\eea
Given $t\geq t_0$, choose $\tilde{t}=t-L$ for an $L$ to be determined later.
Moreover, $t_0$ will have to be chosen large enough so that $\forall \, t \geq t_0$,
\bea
\lb{tverhaeltnis}
(t-L)^{-\tilde{p}}=\tilde{t}^{-\tilde{p}}<2t^{-\tilde{p}}.
\eea

Given a $t$ satisfying \eqref{ft} applying \eqref{starter3} yields
\bea
\lb{starter5}
f(t) + b\int\limits_{\tilde{t}}^{t} f(\bar{t})\md {\bar{t}} &\leq&  \left[2\tilde{C}  +C_0(L+1)\right]\tilde{t}^{-\tilde{p}}\nonumber\\
&\stackrel{\eqref{tverhaeltnis}}{\leq}&\left[4\tilde{C}  +2C_0(L+1)\right]t^{-\tilde{p}}.
\eea

Further, by the pigeonhole principle, there exists $t_{in} \in \left[\tilde{t},t\right]$ such that
\bea
\lb{pigeon1}
f(t_{in})\leq \frac{1}{L} \int\limits_{\tilde{t}}^{t} f(\bar{t})\md {\bar{t}}.  
\eea
Since $f(t)$ is a positive function \eqref{starter5} also leads to
\bea
\lb{starter4}
b\int\limits_{\tilde{t}}^{t} f(\bar{t})\md {\bar{t}} \leq \left[4\tilde{C}  +2C_0(L+1)\right]t^{-\tilde{p}}.
\eea
Thus, \eqref{pigeon1} and \eqref{starter4} yield
\bea
\lb{vier}
f(t_{in})\leq \frac{1}{bL}\left[4\tilde{C}  +2C_0(L+1)\right]t^{-\tilde{p}}.
\eea

Now let $\tilde{t}=t_{in}$ and use \eqref{vier} in \eqref{starter3}, then
\bea
\lb{schluss}
f(t)\leq f(t) + b\int\limits_{t_{in}}^{t} f(\bar{t})\md {\bar{t}} &\leq&  \left(\frac{1}{bL}\left[2\tilde{C}  +C_0(L+1)\right]+C_0(L+1)\right)t_{in}^{-\tilde{p}},\nonumber\\
&\stackrel{\eqref{tverhaeltnis}}{\leq}& \left[\frac{4\tilde{C}}{bL}+\frac{2C_0(L+1)}{bL}+2C_0(L+1)\right]t^{-\tilde{p}}.
\eea

If $1-\frac{4}{bL}>0$
and
\bea
\tilde{C}\geq \left(1-\frac{4}{bL}\right)^{-1}\left[\frac{2C_0(L+1)}{bL}+2C_0(L+1)\right]
\eea
then \eqref{wantedft} follows.

Thus picking first $L$ such that $1-\frac{4}{bL}>0$, and then $t_0$ such that $t_0\geq L+1$ and satisfying \eqref{tverhaeltnis}, and finally choosing $\tilde{C}$ as \mbox{$\tilde{C}=\mbox{max}\left\{\left[ f(1) +C_0{t_0}^{-1-2\delta}\right],\left(1-\frac{4}{bL}\right)^{-1}\left[\frac{2C_0(L+1)}{bL}+2C_0(L+1)\right]\right\}$} \eqref{wantedft} and thus
\eqref{ftwant} follows by continuity.
\end{proof}

By Lemma \ref{decay_lemma2} we obtain from \eqref{decay_statement} together with \eqref{estineu}

\bea
\lb{decay_1}
E(\phi; v_*)=\int\limits_ {\left\lbrace  r_{red}\leq r \leq r_+\right\rbrace }  J_{\mu}^N(\phi) n^{\mu}_{v=v_*} \dV_{v=v_*}
\leq \tilde{\tilde{C}} {v_*}^{-2-2\delta},
\eea
with $\tilde{\tilde{C}}$ depending on $\tilde{b}_1$ and $D_{0}(u_{\diamond}, 1)$.

Finally, in order to close the proof of Proposition \ref{rtildedecay} we perform again the divergence theorem but for region \mbox{$\tilde{\cR}_I=\left\{r_{red}\leq r \leq r_+\right\}\cap \left\{v_*\leq v\leq v_*+1\right\}$}:
\bea
\lb{nochmaldiv}
& &\int\limits_{\tilde{\cR_I}}  K^N(\phi) \dV
+ \int\limits_ {\left\lbrace  v_* \leq v \leq v_*+1\right\rbrace }  J_{\mu}^N(\phi) n^{\mu}_{r=r_{red}} \dV_{r=r_{red}}
+ \int\limits_ {\left\lbrace  r_{red}\leq r \leq r_+\right\rbrace }  J_{\mu}^N(\phi) n^{\mu}_{v=v_*
+1} \dV_{v=v_*+1}\nonumber\\
&=&
 \int\limits_ {\left\lbrace  r_{red}\leq r \leq r_+\right\rbrace }  J_{\mu}^N(\phi) n^{\mu}_{v=v_*} \dV_{v=v_*}
+\int\limits_ {\left\lbrace  v_* \leq v \leq v_*+1\right\rbrace }  J_{\mu}^N(\phi) n^{\mu}_{\cH^+} \dV_{\cH^+}.
\eea
In view of the signs we obtain
\ben
\Rightarrow &&  \int\limits_ {\left\lbrace  v_* \leq v \leq v_*+1\right\rbrace }  J_{\mu}^N(\phi) n^{\mu}_{r=r_{red}} \dV_{r=r_{red}}
\leq
 \int\limits_ {\left\lbrace  r_{red}\leq r \leq r_+\right\rbrace }  J_{\mu}^N(\phi) n^{\mu}_{v=v_*} \dV_{v=v_*}
+\int\limits_ {\left\lbrace  v_* \leq v \leq v_*+1\right\rbrace }  J_{\mu}^N(\phi) n^{\mu}_{\cH^+} \dV_{\cH^+}.
\een
Due to \eqref{decay_1} and Theorem \ref{anfang} we are left with the conclusion of Proposition \ref{rtildedecay}.
\end{proof}
Note that the above also implies the following statement.
\begin{cor}
Let $\phi$ be as in Theorem \ref{anfang} and for $r_{red}$ as in Proposition \ref{mi}. Then, for all $v_*\geq 1$, $v_*+1\leq v_{red}(\tilde{u})$ and for all $\tilde{u}$ such that $r(\tilde{u},v_*+1) \in [r_{red}, r_+)$, we have
\lb{cor5.2}
\bea
\int\limits_ {\left\lbrace  v_* \leq v \leq v_*+1\right\rbrace } J^N_\mu(\phi)n^\mu_{u=\tilde{u}}dVol_{u=\tilde{u}} \le C v_*^{-2-2\delta},
\eea
with $C$ depending on $C_{0}$ of Theorem \ref{anfang} and $D_{0}(u_{\diamond}, 1)$ of Proposition \ref{initialdataprop}, where $u_{\diamond}$ is defined by $r_{red}=r(u_{\diamond},1)$ and $v_{red}(\tilde{u})$ as in \eqref{notation_neu}.
\end{cor}
\begin{proof}
The conclusion of the statement follows by applying again the divergence theorem 
and using the results of the proof of Proposition \ref{rtildedecay}. 
\end{proof}

\section{Propagating through $\cN$ from $r=r_{red}$ to $r=r_{blue}$}
\lb{zweite_region}
Now that we have obtained a decay bound along the $r=r_{red}$ hypersurface in the previous section, we propagate the estimate further inside the black hole through the noshift region $\cN$ up to the $r=r_{blue}$ hypersurface. 
In order to do that we will use the future directed timelike vector field 
\bea
\lb{partial_r}
-\partial_r=\frac{1}{{\Omega^2}}(\partial_u+\partial_v).
\eea

Using \eqref{partial_r} in \eqref{Kplug} of Appendix \ref{Kcurrents} we obtain
\bea
\lb{partialrbulk}
K^{-\partial_r}&=&\frac{4}{\Omega^3}\left[\frac{\partial_u \Omega}{\Omega}(\partial_v\phi)^2 +\frac{\partial_v \Omega}{\Omega}(\partial_u\phi)^2 \right]\nonumber\\
&&- \frac{4}{r\sqrt{\Omega^2}}(\partial_u \phi \partial_v \phi)\nonumber\\
&&-\left(\frac{\partial_u \Omega}{\Omega^2}+\frac{\partial_v \Omega}{\Omega^2}\right)\left(\frac{1}{\Omega}-1\right)|\nabb \phi|^2,
\eea
for the bulk current. It has the property that it can be estimated by
\bea
\lb{controll_bulk_partial_r}
 |K^{-\partial_r}(\phi)| \leq B_1 J_{\mu}^{-\partial_r}(\phi) n^{\mu}_{r={\bar{r}}},
\eea
where $B_1$ is independent of $v_*$.
Validity of the estimate can in fact be seen without computation from the fact that timelike currents, such as $J_{\mu}^{-\partial_r}(\phi) n^{\mu}_{r={\bar{r}}}$ contain all derivatives. 
The uniformity of $B_1$ is given by the fact that $K^{-\partial_r}$ and $J^{-\partial_r}$ are invariant under translations along $\partial_t$, cf.~ Section \ref{rtcoords} for definition of the $t$ coordinate. Therefore, we can just look at the maximal deformation on a compact \mbox{$\left\{t=const \right\} \cap \left\{r_{blue} \leq r \leq r_{red}\right\}$} hypersurface and get an estimate for the deformation everywhere.

\begin{prop}
\label{r_{red}}
Let $\phi$ be as in Theorem \ref{anfang}, $r_{blue}$ as in \eqref{rblue} and $r_{red}$ as in Proposition \ref{mi}. Then, for all $v_*>1$ and $\tilde{r} \in [r_{blue},r_{red})$, we have 
\ben
\int\limits_ {\left\lbrace  v_* \leq v \leq {v_*}+1\right\rbrace } J_{\mu}^{-\partial_r}(\phi) n^{\mu}_{r=\tilde{r}} \dV_{r=\tilde{r}}
\leq
C{v_*}^{-2-2\delta}, 
\een
with $C$ depending on the initial data or more precisely depending on $C_{0}$ of Theorem \ref{anfang} and $D_{0}(u_{\diamond}, 1)$ of Proposition \ref{initialdataprop}, where $u_{\diamond}$ is defined by $r_{red}=r(u_{\diamond},1)$.
\end{prop}

\begin{proof}
Given $v_*$, we define regions $\cR_{II}$ and $\tilde{\cR}_{II}$ as in Figure \ref{r_{red}_decay}, where we use \eqref{notation_neu} and
\bea
\lb{notation_v_*}
v(\tilde{r}, v_*) \quad &\mbox{is determined by}& \quad r(u_{blue}(v_*), v(\tilde{r}, v_*))=\tilde{r}.
\eea
Thus the depicted regions are given by \mbox{$\cR_{II}\cup \tilde{\cR}_{II}= \cD^+(\left\lbrace v_1 \leq v \leq v_*+1\right\rbrace \cap  \left\lbrace r=r_{red}\right\rbrace )\cap \cN $}, where region \mbox{$\cR_{II}$} is given by \mbox{$\cR_{II}= \cD^+(\left\lbrace v_1 \leq v \leq v_*\right\rbrace \cap \left\lbrace r=r_{red}\right\rbrace)  $}.

{\begin{figure}[ht]
\centering
\includegraphics[width=0.6\textwidth]{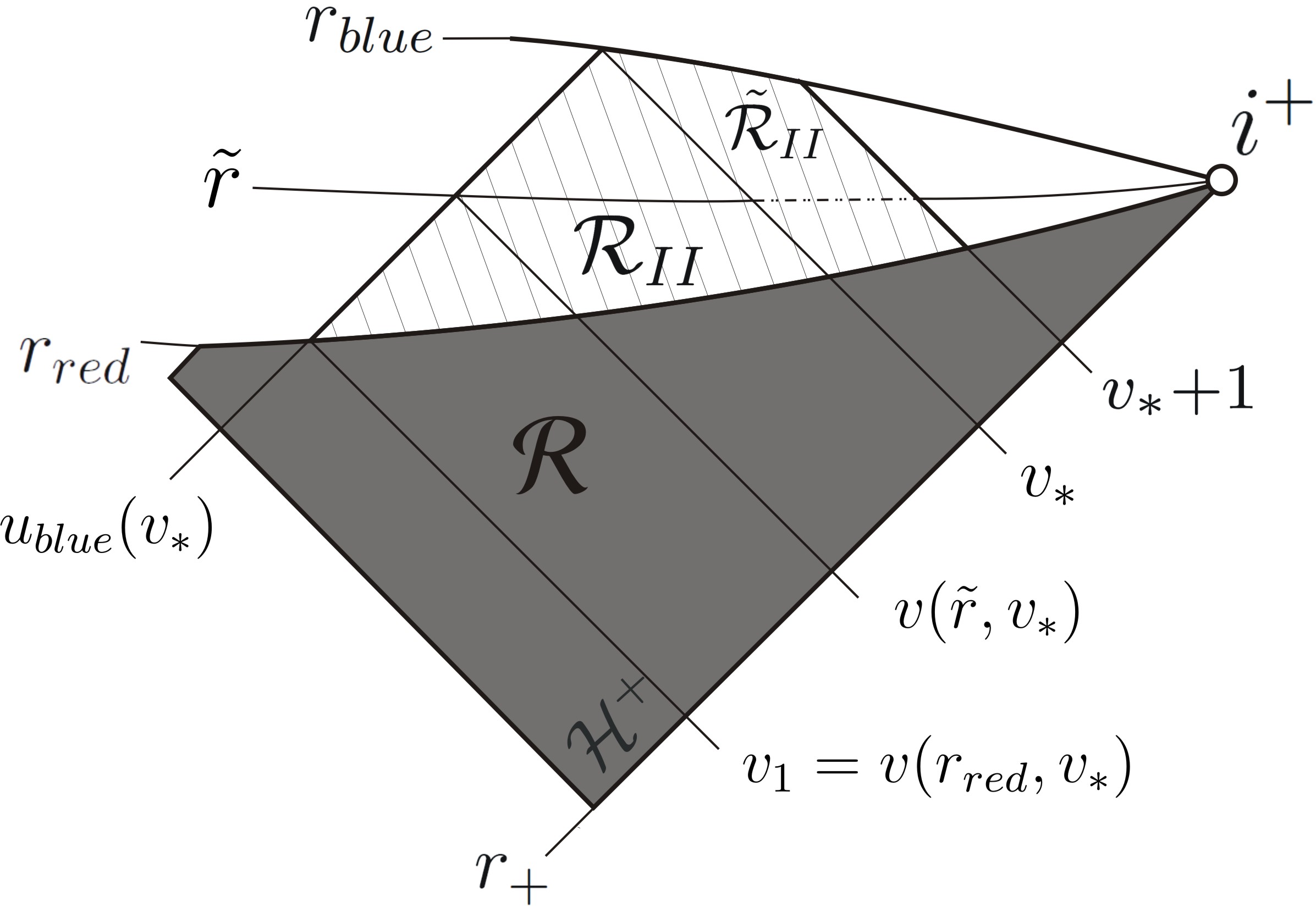}
\caption[Text der im Bilderverzeichnis auftaucht]{Region $\cR_{II}\cup \tilde{\cR}_{II}$ represented as the hatched area.}
\label{r_{red}_decay}\end{figure}}
In the following we will apply the divergence theorem in region $\cR_{II}\cup \tilde{\cR}_{II}$ to obtain decay on an arbitrary $r=\tilde{r}$ hypersurface, dash-dotted line, for $\tilde{r} \in [r_{blue},r_{red})$, from the derived decay on the $r=r_{red}$ hypersurface.

\ben
\int\limits_ {\cR_{II}\cup \tilde{\cR}_{II}}  K^{-\partial_r}(\phi) \dV
&+& \int\limits_ {\left\lbrace  r_{blue}\leq r \leq r_{red}\right\rbrace }  J_{\mu}^{-\partial_r}(\phi) n^{\mu}_{u=u_{blue}(v_*)} \dV_{u=u_{blue}(v_*)}\\
+ \int\limits_ {\left\lbrace  v_* \leq v \leq v_*+1\right\rbrace }  J_{\mu}^{-\partial_r}(\phi) n^{\mu}_{r=r_{blue}} \dV_{r=r_{blue}}
&+& \int\limits_ {\left\lbrace  r_{blue}\leq r \leq r_{red}\right\rbrace }  J_{\mu}^{-\partial_r}(\phi) n^{\mu}_{v=v_*+1} \dV_{v=v_*+1} \\
=
\int\limits_ {\left\lbrace  v_1 \leq v \leq v_*+1\right\rbrace }  J_{\mu}^{-\partial_r}(\phi) n^{\mu}_{r=r_{red}} \dV_{r=r_{red}}.
\een
The second integral of the left hand side represents the current through the $u=u_{blue}(v_*)$ hypersurface, defined by \eqref{notation_neu}.
%
As $u=u_{blue}(v_*)$ is a null hypersurface and $-\partial_r$ is timelike, the positivity of that second term is immediate. Similarly, the fourth term of the left hand side of our equation is positive and we obtain
\ben
\int\limits_ {\left\lbrace  v_* \leq v \leq v_*+1\right\rbrace }  J_{\mu}^{-\partial_r}(\phi) n^{\mu}_{r=r_{blue}} \dV_{r=r_{blue}}
\leq
\int\limits_ {\cR_{II}\cup \tilde{\cR}_{II}}  |K^{-\partial_r}(\phi)| \dV
+\int\limits_ {\left\lbrace  v_1 \leq v \leq v_*+1\right\rbrace }  J_{\mu}^{-\partial_r}(\phi) n^{\mu}_{r=r_{red}} \dV_{r=r_{red}}.
\een

Further, we use that the deformation $K^{-\partial_r}$ is controlled by the energy associated to the timelike vector field $-\partial_r$ as stated in \eqref{controll_bulk_partial_r}.
Thus we obtain
\bea
\lb{oberes}
\int\limits_ {\left\lbrace  v_* \leq v \leq v_*+1\right\rbrace }  J_{\mu}^{-\partial_r}(\phi) n^{\mu}_{r=r_{blue}} \dV_{r=r_{blue}}
&\leq&
B_1 \int\limits_{\cR_{II}\cup{\tilde{\cR}}_{II}}  J_{\mu}^{-\partial_r}(\phi) n^{\mu}_{r=\bar{r}} \dV\nonumber\\
&&+\int\limits_ {\left\lbrace  v_1 \leq v \leq v_*+1\right\rbrace }  J_{\mu}^{-\partial_r}(\phi) n^{\mu}_{r=r_{red}} \dV_{r=r_{red}}.
\eea
By the coarea formula we obtain
\bea
\lb{94}
\int\limits_ {\left\lbrace  v_* \leq v \leq v_*+1\right\rbrace }  J_{\mu}^{-\partial_r}(\phi) n^{\mu}_{r=r_{blue}} \dV_{r=r_{blue}}
&\leq&
\tilde{B}_1 \int\limits_ {r_{blue}}^{r_{red}}\int\limits_{\left\lbrace  v(\bar{r},v_*) \leq v \leq v_*+1\right\rbrace }  J_{\mu}^{-\partial_r}(\phi) n^{\mu}_{r=\bar{r}} \dV_{r=\bar{r}}\md \bar{r}\nonumber\\
&&+\int\limits_ {\left\lbrace  v_1 \leq v \leq v_*+1\right\rbrace }  J_{\mu}^{-\partial_r}(\phi) n^{\mu}_{r=r_{red}} \dV_{r=r_{red}}.
\eea

Now let 
\bea
E(\phi;\tilde{r},\tilde{v})=\int\limits_ {\left\lbrace  \tilde{v} \leq v \leq v_*+1\right\rbrace }  J_{\mu}^{-\partial_r}(\phi) n^{\mu}_{r=\tilde{r}} \dV_{r=\tilde{r}}, 
\eea
with $\tilde{r} \in [r_{blue},r_{red})$. 
Replacing $r_{blue}$ with $\tilde{r}$ in the above, considering the future domain of dependence of $\left\lbrace v_1 \leq v \leq v_*+1\right\rbrace \cap \left\lbrace r=r_{red}\right\rbrace$ up to the $r=\tilde{r}$ hypersurface we obtain similarly to \eqref{94}
\be
\lb{E_energy2}
E(\phi;\tilde{r},v(\tilde{r}, v_*))\leq \tilde{B}_1\int\limits_ {\tilde{r}}^{r_{red}}E(\phi;\bar{r},v(\bar{r}, v_*))\md \bar{r} +E(\phi;r_{red},v_1).
\ee

Using Gr\"onwall's inequality in \eqref{E_energy2} yields
\bea
E(\phi;\tilde{r},v(\tilde{r}, v_*))&\leq& E(\phi;r_{red},v_1)\left[ 1+\tilde{B}_1(r_{red}-\tilde{r})e^{\tilde{B}_1(r_{red}-\tilde{r})}\right]\nonumber \\
\Rightarrow E(\phi;\tilde{r},v(\tilde{r}, v_*))&\leq& \tilde{C}E(\phi;r_{red},v_1).
\eea
Finally, note that 
\bea
\lb{vconstk}
\left[v_*+1\right]-v(r_{red}, v_*) =\left[v_*+1\right]-v_1= k< \infty,
\eea
where $k=2\left[r^{\star}(r_{blue})-r^{\star}(r_{red})\right]+1$.
This can be seen since \eqref{def_l_n} and \eqref{rstar1} yields
\ben
&&\frac{\partial_v r}{\Omega^2}=-\partial_v r^{\star}=-\frac12\\
\Rightarrow&&r^{\star}(u_{blue}(v_*), v_*)-r^{\star}(u_{blue}(v_*), v_1)=r^{\star}(r_{blue})-r^{\star}(r_{red})\stackrel{\eqref{regge}}{=}\frac12 (v_*-v_1).
\een
Further, by using the conclusion of Proposition \ref{rtildedecay} and \eqref{vconstk} we have
\bea
\lb{v_1_decay}
E(\phi; r_{red}, v_1)=\int\limits_ {\left\lbrace  v_1 \leq v \leq v_*+1\right\rbrace } J_{\mu}^{-\partial_r}(\phi) n^{\mu}_{r=r_{red}} \dV_{r=r_{red}}
&=&\int\limits_ {\left\lbrace  v_1 \leq v \leq v_1+k\right\rbrace }  J_{\mu}^{-\partial_r}(\phi) n^{\mu}_{r=r_{red}} \dV_{r=r_{red}}\nonumber\\
&\leq&
 C\mbox{max}\left\lbrace {k,1}\right\rbrace {v_1}^{-2-2\delta}\sim  C{v_*}^{-2-2\delta}.
\eea

We thus infer
\ben
\int\limits_ {\left\lbrace  v_* \leq v \leq v_*+1\right\rbrace }  J_{\mu}^{-\partial_r}(\phi) n^{\mu}_{r=\tilde{r}} \dV_{r=\tilde{r}}
\leq
\tilde{C}\int\limits_ {\left\lbrace  v_1 \leq v \leq v_*+1\right\rbrace }  J_{\mu}^{-\partial_r}(\phi) n^{\mu}_{r=r_{red}} \dV_{r=r_{red}}
\leq
\tilde{C} C{v_*}^{-2-2\delta}.
\een
\end{proof}
The above now also implies the following statement.
\begin{cor}
\lb{cor6.1}
Let $\phi$ be as in Theorem \ref{anfang}, $r_{blue}$ as in \eqref{rblue} and $r_{red}$ as in Proposition \ref{mi}. Then, for all $v_*>1$ and all $\tilde{u}$ such that $r(\tilde{u}, v_*) \in [r_{blue}, r_{red})$
\bea
\int\limits_ {\left\lbrace v_{red}(\tilde{u}) \leq v \leq v_{blue}(\tilde{u})\right\rbrace } J^N_\mu(\phi)n^\mu_{u=\tilde{u}}dVol_{u=\tilde{u}} \le C v_*^{-2-2\delta}, 
\eea
with $C$ depending on $C_{0}$ of Theorem \ref{anfang} and $D_{0}(u_{\diamond}, 1)$ of Proposition \ref{initialdataprop}, where $u_{\diamond}$ is defined by $r_{red}=r(u_{\diamond},1)$ and $v_{red}(\tilde{u})$, $v_{blue}(\tilde{u})$ are as in \eqref{notation_neu}.
\end{cor}
\begin{proof}
The conclusion of the statement follows by considering the divergence theorem for a triangular region \mbox{$J^-(x)\cap \cN$} with \mbox{$x=(\tilde{u}, v_{blue}(\tilde{u}))$}, $x \in J^-(r=r_{blue})$ and using the results of Proposition \ref{r_{red}}. Note that \mbox{$v_*\sim v_{blue}(\tilde{u})\sim v_{red}(\tilde{u})$}.
\end{proof}

By the previous proposition we have successfully propagated the energy estimate further inside the black hole, up to $r=r_{blue}$.
To go even further will be more difficult and we will address this in the next section.

\section{Propagating through $\cB$ from $r=r_{blue}$ to the hypersurface $\gamma$}
\lb{blueshift1}

In the following we want to propagate the estimates from the $r=r_{blue}$ hypersurface further into the blueshift region to a hypersurface $\gamma$ which is located a logarithmic distance in $v$ from the $r=r_{blue}$ hypersurface, cf.~ Figure \ref{region3_neu}. 
{\begin{figure}[ht]
\centering
\includegraphics[width=0.8\textwidth]{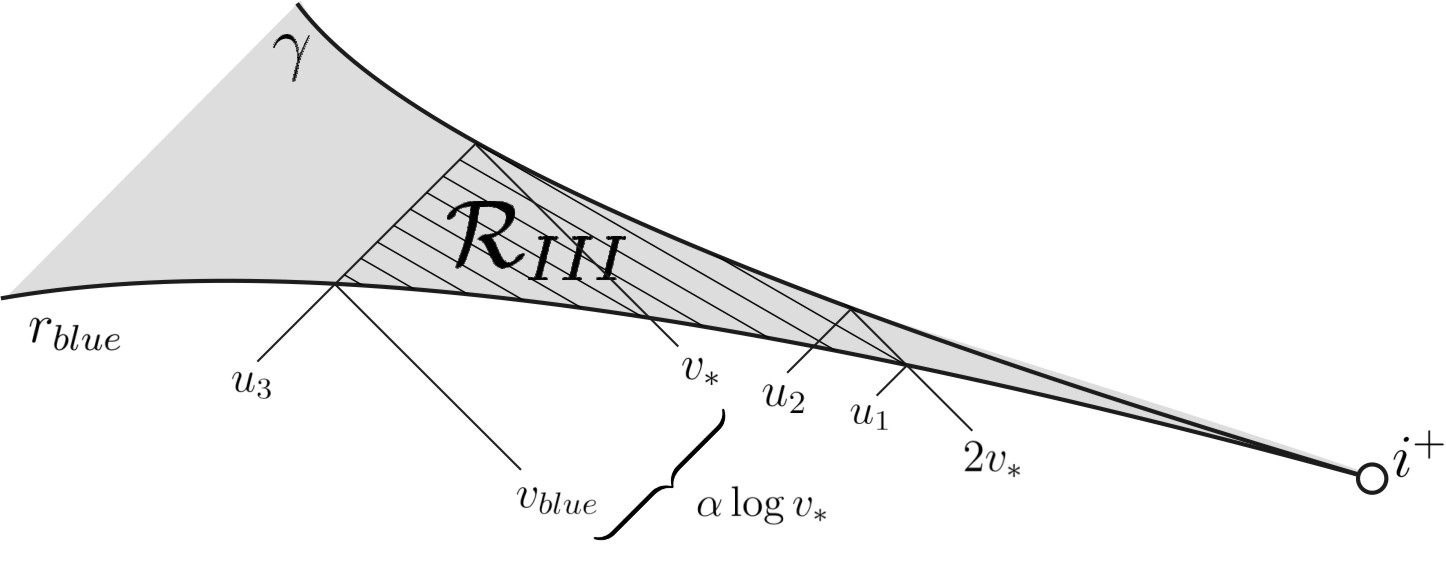}
\caption[Text der im Bilderverzeichnis auftaucht]{Logarithmic distance of hypersurface $r=r_{blue}$ and hypersurface $\gamma$ depicted in a Penrose diagram.}
\label{region3_neu}\end{figure}}
We will define the hypersurface $\gamma$ and its most basic properties in Section \ref{gamma_curve} and propagate the decay bound to $\gamma$ in Section \ref{togamma}.

\subsection{The hypersurface $\gamma$}
\lb{gamma_curve}
The idea of the hypersurface $\gamma$ was already entertained in \cite{m_interior} by Dafermos and basically locates $\gamma$ a logarithmic distance in $v$ from a constant $r$ hypersurface living in the blueshift region. 

Let $\alpha$ be a fixed constant satisfying 
\bea
\lb{alpha}
\alpha>\frac{p+1}{\beta},
\eea 
with $\beta$ as in \eqref{lowerboundu} and \eqref{lowerboundv}.
(The significance of the bound \eqref{alpha} will become clear later.) Let us for convenience also assume that
\bea
\alpha&>&1,
\eea
and
\bea
\alpha(2-\log 2\alpha) &>&2r^{\star}_{blue}+1.
\eea
We define the function $H(u,v)$ by
\bea
\lb{gross_h}
H(u,v)=u+v-\alpha \log v-2r^{\star}(r_{blue})=u+v-\alpha \log v-2r^{\star}_{blue},
\eea
were $r^{\star}(r_{blue})=r^{\star}_{blue}$ is the $r^{\star}$ value evaluated at $r_{blue}$ according to \eqref{rstar2}, and $r^{\star}_{blue}>0$ according to the choice \eqref{rchoice}.
We then define the hypersurface $\gamma$ as the levelset 
\bea
\lb{gammadefine}
\gamma=\left\{H(u,v)=0\right\}\cap \{v> 2\alpha\}.
\eea 
Since
\bea
\lb{ableit}
\frac{\partial H}{\partial u}=1, \qquad \frac{\partial H}{\partial v}=1-\frac{\alpha}{v},
\eea
we see that $\gamma$ is a spacelike hypersurface and terminates at $i^+$, cf.~ Appendix \ref{Jcurrents}. (In the notation \eqref{notation_neu}, \mbox{$u_{\gamma}(v)\rightarrow -\infty$} as \mbox{$v\rightarrow \infty$}.)
Note that by our choices $u<-1$ and $v>|u|$ in $\cD^+(\gamma)$.

Recall that in Section \ref{rtcoords} we have defined the Regge-Wheeler tortoise coordinate $r^{\star}$ depending on $u$, $v$ by \eqref{regge}. Using this for $r^{\star}_{blue}$ we have
\bea
r^{\star}(r_{blue})=\frac{v_{blue}(u)+u}{2},
\eea
with $v_{blue}(u)$ as in \eqref{notation_neu}. Plugging this into \eqref{gammadefine} recalling $v_{\gamma}(u)$ defined in \eqref{notation_neu}, we obtain
the relation
\bea
\lb{gamma}
v_{\gamma}(u)-v_{blue}(u)=\alpha \log v_{\gamma}(u).
\eea
As we shall see in Section \ref{finiteness} the above properties of $\gamma$ will allow us to derive pointwise estimates of $\Omega^2$ in $J^+(\gamma)\cap \cB$. We first turn however to the region $J^-(\gamma)\cap \cB$.

\subsection{Energy estimates from $r=r_{blue}$ to the hypersurface $\gamma$}
\lb{togamma}
Now we are ready to propagate the energy estimates further into the blueshift region $\cB$ up to the hypersurface $\gamma$.
We will in this part of the proof use the vector field 
\ben
{S_0}=r^q\partial_{r^{\star}}=r^q(\partial_u+\partial_v),
\een 
which we have defined in \eqref{Nstar}.
Let us now consider the bulk term and derive positivity properties which are needed later on.
Plugging \eqref{Nstar} in \eqref{Kplug} of Appendix \ref{Kcurrents} yields
\bea
\lb{KNstar}
K^{{S_0}}=&+& qr^{q-1}\left[(\partial_v \phi)^2+(\partial_u \phi)^2\right]\nonumber\\
&-&\left[ \frac{qr^{q-1}}{2}\left[\partial_v r +\partial_u r\right]+{r^q}\left(\frac{\partial_u \Omega}{\Omega}+\frac{ \partial_v \Omega}{\Omega}\right)\right] |\nabb \phi|^2\nonumber\\
&-&4r^{q-1}(\partial_u \phi\partial_v \phi).
\eea
Our aim is to show that $K^{{S_0}}$ is positive. All terms multiplied by the angular derivatives are manifestly positive in $\cB$, cf.~ \eqref{lowerboundu}, \eqref{lowerboundv} together with \eqref{def_l_n} to \eqref{v-neg}. 
Therefore, it is only left to show that the first term on the right hand side dominates the last term.
Since by the Cauchy-Schwarz inequality
\bea
\lb{positivityproperty}
2qr^{q-1}(\partial_u \phi\partial_v \phi)\leq qr^{q-1}\left[(\partial_v \phi)^2+(\partial_u \phi)^2\right],
\eea
$K^{{S_0}}$ is positive in $\cB$ for all $q \geq 2$.

We show now that at the expense of one polynomial power, we can extend the local energy estimate on $r=r_{blue}$ to an energy estimate along $\gamma$ which is valid for a dyadic length. 
\begin{prop}
\lb{to_gamma}
Let $\phi$ be as in Theorem \ref{anfang}.
Then, for all $v_*> 2\alpha$ 
\bea
\lb{sieben}
\int\limits_ {\left\lbrace  v_* \leq v \leq 2v_* \right\rbrace } J_{\mu}^{{S_0}}(\phi) n^{\mu}_{\gamma} \dV_{\gamma}
\leq
C{v_*}^{-1-2\delta}, 
\eea
on the hypersurface $\gamma$, with $C$ depending on $C_{0}$ of Theorem \ref{anfang} and $D_{0}(u_{\diamond}, 1)$ of Proposition \ref{initialdataprop}, where $u_{\diamond}$ is defined by $r_{red}=r(u_{\diamond},1)$.
\end{prop}
{\em Remark.} $n^{\mu}_{\gamma}$ denotes the normal vector on the hypersurface $\gamma$ which is a levelset $\gamma=\left\{H(u,v)=0\right\}$ of the function $H(u,v)$ defined in
\eqref{gross_h}. For calculation of $n^{\mu}_{\gamma}$ and $J_{\mu}^{{S_0}}(\phi) n^{\mu}_{\gamma}$ refer to \eqref{ngamma} and \eqref{jgamma} of Appendix \ref{Jcurrents}.

\begin{proof}
In the following we will again make use of notation \eqref{notation_neu}.

Let $v_*> 2\alpha$, such that $\gamma$ is spacelike for $v>v_*$, cf.~ Section \ref{gamma_curve}.
Define $u_3$ by \mbox{$(u_3, v_*) \in \gamma$}, i.e.~ \mbox{$u_{\gamma}(v_*)=u_3$} and define $v_{blue}$ as the intersection of $u_3$ with $r_{blue}$, i.e.~ $v_{blue}(u_3)=v_{blue}$.
And similarly the hypersurfaces $u=u_1$ and $u=u_2$ as shown in Figure \ref{region3_neu} are given by $u_{blue}(2v_*)=u_1$ and $u_{\gamma}(2v_*)=u_2$.
Having defined the relations between all these quantities we can now carry out the divergence theorem for region ${\cR}_{III}$:

\bea
\lb{div_region3}
&&\int\limits_ {{\cR}_{III}}  K^{{S_0}}(\phi) \dV
+ \int\limits_ {\left\lbrace  v_{blue} \leq v \leq v_{*}\right\rbrace }  J_{\mu}^{{S_0}}(\phi) n^{\mu}_{u=\tilde{u}} \dV_{u=\tilde{u}}
+ \int\limits_ {\left\lbrace  v_* \leq v \leq 2v_*\right\rbrace }  J_{\mu}^{{S_0}}(\phi) n^{\mu}_{\gamma} \dV_{\gamma}\nonumber\\
&+& \int\limits_ {\left\lbrace u_1 \leq u \leq u_2 \right\rbrace } J_{\mu}^{{S_0}}(\phi) n^{\mu}_{v=2v_*} \dV_{v=2v_*}
= \int\limits_ {\left\lbrace  v_{blue} \leq v \leq 2v_*\right\rbrace } J_{\mu}^{{S_0}}(\phi) n^{\mu}_{r=r_{blue}} \dV_{r=r_{blue}}.
\eea

Positivity of the flux along the $u=u_3$ segment and the flux along the $v=2v_*$ segment, as well as positivity of $K^{{S_0}}$ for the choice $q\geq 2$, which was derived in \eqref{KNstar} and \eqref{positivityproperty}, leads to
\bea
\lb{r3esti}
\int\limits_ {\left\lbrace  v_* \leq v \leq 2v_*\right\rbrace }  J_{\mu}^{{S_0}}(\phi) n^{\mu}_{\gamma} \dV_{\gamma}
&\leq& \int\limits_ {\left\lbrace  v_{blue} \leq v \leq 2v_*\right\rbrace } J_{\mu}^{{S_0}}(\phi) n^{\mu}_{r=r_{blue}} \dV_{r=r_{blue}},\nonumber\\
&{\leq}& C \mbox{max}\left\{2v_*-v_{blue},1 \right\} v_{blue}^{-2-2\delta},\nonumber\\
&\stackrel{\eqref{gamma}}{\leq}& C \left(v_*+\alpha\log{v_*}\right) v_{blue}^{-2-2\delta},\nonumber\\
&{\leq}& \tilde{C}C v_{*}^{-1-2\delta},
\eea
where the second step is implied by Proposition \ref{r_{red}} and the last step follows from the inequality $v_*\leq C v_{blue}$ which is implied by \eqref{gamma}.
\end{proof}

We have already mentioned in the introduction that we will use the vector field ${S}$, cf.~ \eqref{N1} in the region $J^+(\gamma)\cap\cB$. 
To control the initial flux term of $S$ we require a weighted energy estimate along the hypersurface $\gamma$. 
\begin{cor}
\lb{cor2}
Let $\phi$ be as in Theorem \ref{anfang}.
Then, for all $v_*>2\alpha$ 
\bea
\lb{Nstargammaesti}
\int\limits_ {\left\lbrace  v_* \leq v <\infty \right\rbrace }v^p J_{\mu}^{{S_0}}(\phi) n^{\mu}_{\gamma} \dV_{\gamma}
&\leq& Cv_*^{-1-2\delta+p}, 
\eea
on the hypersurface $\gamma$, with $C$ depending on $C_{0}$ of Theorem \ref{anfang} and $D_{0}(u_{\diamond}, 1)$ of Proposition \ref{initialdataprop}, where $u_{\diamond}$ is defined by $r_{red}=r(u_{\diamond},1)$ and $p$ as in \eqref{waspist}.
\end{cor}
\begin{proof}
This follows by weighting \eqref{sieben} with $v_*^p$ and summing dyadically.
\end{proof}

Further, we can state the following.
\begin{cor}
\lb{cor7.1}
Let $\phi$ be as in Theorem \ref{anfang}, $r_{blue}$ as in \eqref{rblue} and $\gamma$ as in \eqref{gammadefine}. Then, for all $v_*>2\alpha$  and for all $\tilde{u} \in [u_{blue}(v_*), u_{\gamma}(v_*))$
\bea
\int\limits_ {\left\lbrace  v_{blue}(\tilde{u}) \leq v \leq v_{\gamma}(\tilde{u}) \right\rbrace } v^pJ^{S_0}_\mu(\phi)n^\mu_{u=\tilde{u}}dVol_{u=\tilde{u}} \le C v_*^{-1-2\delta+p},
\eea
with $C$ depending on $C_{0}$ of Theorem \ref{anfang} and $D_{0}(u_{\diamond}, 1)$ of Proposition \ref{initialdataprop}, where $u_{\diamond}$ is defined by $r_{red}=r(u_{\diamond},1)$ and $v_{\gamma}(\tilde{u})$, $v_{blue}(\tilde{u})$ as in \eqref{notation_neu}.
\end{cor}
\begin{proof}
The proof is similar to the proof of Corollary \ref{cor6.1} by considering the divergence theorem for a triangular region \mbox{$J^-(x)\cap \cB$} with \mbox{$x=(\tilde{u}, v_{\gamma}(\tilde{u}))$}, $x \in J^-(\gamma)$ and using the results of the proof of Proposition \ref{to_gamma}.
\end{proof}

\section{Propagating through $\cB$ from the hypersurface $\gamma$ to $\cC\cH^+$ in the neighbourhood of $i^+$}
\lb{innerhorizon}
In order to prove our Theorem \ref{dashier} and close our estimates up to the Cauchy horizon in the neighbourhood of $i^+$ we are interested in considering a region ${\cR_{IV}}$ within the trapped region whose boundaries are made up of the hypersurface $\gamma$, a constant $u$ and a constant $v$ segment, which can reach up to the Cauchy horizon, cf.~ Figure \ref{RN_mit_u_v}.
{\begin{figure}[ht]
\centering
\includegraphics[width=0.5\textwidth]{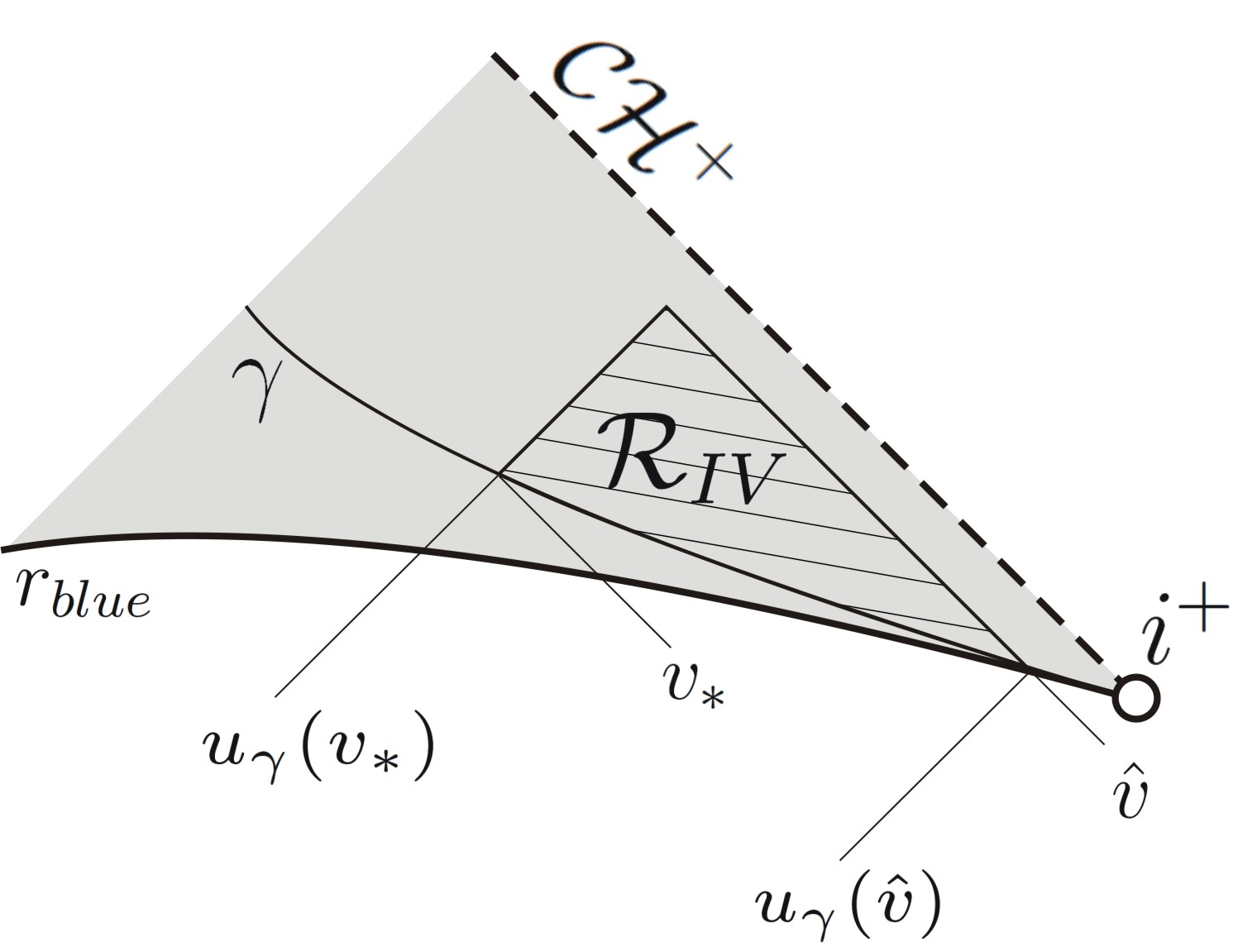}
\caption[]{Blueshift region of Reissner-Nordstr\"om spacetime from hypersurface $\gamma$ onwards.}
\label{RN_mit_u_v}\end{figure}}
Let $v_*>2\alpha$ and let $\hat{v}>v_*$.
We may write \mbox{$\cR_{IV}=J^+(\gamma)\cap J^-(x)$} with 
\mbox{$x=(u_{\gamma}(v_*), \hat{v})$}, $x \in J^+(\gamma)\cap\cB$.
Note that $\cR_{IV}$ lies entirely in the blueshift region, which was characterized by the fact that the quantity $M-\frac{e^2}{r}$ takes the negative sign, cf.~ \eqref{lowerboundu}, \eqref{lowerboundv} and \eqref{def_l_n} to \eqref{v-neg}.

In Section \ref{finiteness} we will derive pointwise estimates for $\Omega^2$ in the future of the hypersurface $\gamma$. 
With this estimate, the bulk term will be bounded in terms of the currents through the null hypersurfaces. Consequently, we will be able to absorb the bulk term and to show that the currents through the null hypersurfaces can be bounded by the current along the hypersurface $\gamma$, cf.~ Section \ref{cauchy}.

\subsection{Pointwise estimates on $\Omega^2$ in $J^+(\gamma)$}
\lb{finiteness}
In the following we will derive pointwise estimates on $\Omega^2$ in $J^+(\gamma)$. We note that these will imply that the
spacetime volume to the future of the hypersurface $\gamma$ is finite, $\operatorname{Vol}(J^+(\gamma))<C$.

We first derive a future decay bound along a constant $u$ hypersurface for the function $\Omega^2(u,v)$ for $(u, v) \in \cB$.
Let \mbox{$x=(u_{fix}, v_{fix})$, $x \in \cB$}, 
then, from \eqref{v-neg} we can immediately see that  
\bea
\lb{lambda_comp}
\left.\log\left({\Omega^2(u_{fix},{v})}\right)\right|^{\bar{v}}_{v_{fix}}&=&\int\limits^{\bar{v}}_{v_{fix}}  \frac{1}{2r^2}\left(M-\frac{e^2}{r}\right) \md {v},\nonumber\\
&\stackrel{\eqref{lowerboundv}}{\leq}&- \beta[\bar{v}-v_{fix}].
\eea
It then immediately follows that 
\bea
\lb{omegaufix}
\Omega^2(\bar{u},\bar{v}){\leq} \Omega^2(\bar{u}, v_{fix}) e^{-\beta[\bar{v}-v_{fix}]}, \quad \mbox{for all $(\bar{u}, v_{fix}) \in \cB$ and $\bar{v}>v_{fix}$}.
\eea
Analogously, we obtain
\bea
\lb{omegavfix}
\Omega^2(\bar{u},v_{fix}){\leq} \Omega^2(u_{fix}, v_{fix}) e^{-\beta[\bar{u}-u_{fix}]}, \quad \mbox{for all $(\bar{u},v_{fix}) \in J^+(x)$},
\eea
and plugging \eqref{omegaufix} into \eqref{omegavfix} it yields
\bea
\lb{omegafix}
\Omega^2(\bar{u},\bar{v}){\leq} \Omega^2(u_{fix}, v_{fix}) e^{-\beta[\bar{u}-u_{fix}+\bar{v}-v_{fix}]}, \quad \mbox{for all $(\bar{u},\bar{v})  \in J^+(x)$}.
\eea

From \eqref{omegaufix} and \eqref{gamma} we obtain a relation for $\Omega^2(u,v)$ on the hypersurface $\gamma$ as follows
\bea
\lb{lambdalog}
\Omega^2(\bar{u}, \bar{v}) \leq \Omega^2(\bar{u}, v_{blue}(\bar{u})) e^{-\beta\alpha \log v_{\gamma}(\bar{u})}=\Omega^2(\bar{u}, v_{blue}(\bar{u})) {v_{\gamma}(\bar{u})}^{-\beta \alpha},\quad \mbox{for $(\bar{u}, \bar{v}) \in \gamma$}.
\eea
For $J^+(\gamma)$, using \eqref{omegaufix} we further get
\bea
\lb{spacevolumedecay}
\Omega^2(\bar{u}, \bar{v}) \leq C{v_{\gamma}(\bar{u})}^{-\beta \alpha} e^{-\beta\left[\bar{v}-v_{\gamma}(\bar{u})\right]}, \quad \mbox{for $(\bar{u}, \bar{v}) \in J^+(\gamma)$},
\eea
where we have used $\Omega^2(\bar{u}, v_{blue}(\bar{u}))\leq C$.
Moreover, we may think of a parameter $\bar{v}$ which determines the associated $u$ value via intersection with $\gamma$, we denote this value by the evaluation the function $u_{\gamma}(\bar{v})$ which was introduced in \eqref{notation_neu}, cf.~ Figure \ref{integralbild2}. 

Moreover, by \eqref{u-neg} we can also state 
\bea
\lb{spacevolumedecay_u}
\Omega^2(\bar{u}, \bar{v}) \leq C {|u_{\gamma}(\bar{v})|}^{-\beta \alpha} e^{\beta\left[u_{\gamma}(\bar{v})-\bar{u}\right]} \quad \mbox{for $(\bar{u}, \bar{v}) \in J^+(\gamma)$}.
\eea
Note that the choice \eqref{alpha} of $\alpha$ implies that $\beta \alpha>1$. From \eqref{spacevolumedecay_u}, the fact that 
$|u_{\gamma}(\bar{v})| \sim \bar{v}$, and the extra exponential factor, finiteness of the spacetime volume to the future of $\gamma$ follows, 
\bea
\lb{finitevol}
\operatorname{Vol}(J^+(\gamma))<C.
\eea
See also \cite{kommemi}.

\subsection{Bounding the bulk term $K^{S}$}
\lb{knsec}
To derive energy estimates in ${\cR_{IV}}$ we use the timelike vector field multiplier
\ben
\lb{{S}}
{S}=|u|^p\partial_u+v^p\partial_v,
\een
which we have given before in \eqref{N1}. 
The weights of ${S}$ are chosen such that they will allow us to derive pointwise estimates from energy estimates; see Section \ref{uni_bounded}.

In order to obtain our desired estimates
first of all we need a bound on the scalar current $K^{S}$, in terms of $J_{\mu}^{S}(\phi) n^{\mu}_{v=\bar{v}}$ and $J_{\mu}^{S}(\phi) n^{\mu}_{u=\bar{u}}$.
In the following we will bound the occurring $(u,v)$-dependent weight functions by functions that depend on either $u$ or $v$, respectively.
Plugging the vector field ${S}$, cf.~ \eqref{N1}, into \eqref{Kplug} of Appendix \ref{Kcurrents} we obtain
\bea
\lb{KN}
K^{S}=&-&\frac{2}{r}\left[ v^p +|u|^p\right] (\partial_u \phi \partial_v \phi)\nonumber\\
&-&\left[ \frac{\partial_u \Omega}{\Omega}|u|^p+\frac{ \partial_v \Omega}{\Omega}v^p+\frac{p}{2}(v^{p-1}+|u|^{p-1})\right] |\nabb \phi|^2.
\eea

Recall \eqref{lowerboundu} and \eqref{lowerboundv}. For large absolute values of $v$ and $u$ the first two terms multiplying the angular derivatives of $\phi$ dominate the last two terms, so in total the term multiplying the angular derivatives is always positive in $\cD^+(\gamma)$. Consequently we will be able to use this property to derive an inequality by using the divergence theorem in the proof of Proposition \ref{kastle}.
Let us therefore define
\bea
\lb{KNtilde}
\tilde{K}^{S}=&-&\frac{2}{r}\left[ v^p +|u|^p\right] (\partial_u \phi \partial_v \phi),
\eea
and state
\bea
\lb{knrelation}
-{K}^{S}\leq |\tilde{K}^{S}| \quad \mbox{for $v>\frac{p}{2\beta}$ and $|u|>\frac{p}{2\beta}$, cf.~ \eqref{lowerboundv}, \eqref{lowerboundu}}.
\eea
(Note that $\tilde{K}^{S}$ coincides with the bulk term for spherically symmetric $\phi$.) We have the following
\begin{lem}
\label{K_spher}
Let $\phi$ be an arbitrary function. Then, for all $v_*>2\alpha$ 
and all $\hat{v}>v_*$,
the integral over region \mbox{$\cR_{IV}=J^+(\gamma)\cap J^-(x)$} with \mbox{$x=(u_{\gamma}(v_*), \hat{v})$}, $x \in \cB$, cf.~ Figure \ref{RN_mit_u_v}, of the current $\tilde{K}^{S}$, defined by \eqref{KNtilde}, can be estimated by
\bea
\lb{Prop6.1}
\int\limits_{\cR_{IV}} |\tilde{K}^{S}| \dV \leq &\delta_1& \sup_{u_{\gamma}(\hat{v})\leq \bar{u}\leq u_{\gamma}(v_*)}\int\limits_{\left\lbrace  v_{\gamma}(\bar{u}) \leq v \leq \hat{v}\right\rbrace }  J_{\mu}^{S}(\phi) n^{\mu}_{u=\bar{u}}\dV_{u=\bar{u}}\nonumber \\
+ &\delta_2& \sup_{v_*\leq \bar{v} \leq \hat{v}}\int\limits_{\left\lbrace  u_{\gamma}(\hat{v}) \leq u \leq u_{\gamma}(\bar{v}) \right\rbrace }
J_{\mu}^{S}(\phi) n^{\mu}_{v=\bar{v}} \dV_{v=\bar{v}},
\eea
where $\delta_1$ and $\delta_2$ are positive constants, with $\delta_1\rightarrow 0$ and $\delta_2\rightarrow 0$ as $v_*\rightarrow \infty$.
\end{lem}
{\em Remark.} In the proof of Proposition \ref{kastle} we will see that the above proposition determines $u_{\schere}$ of Theorem \ref{dashier}, depicted in Figure \ref{bnr}. We have to choose $u_{\schere}=u_{\gamma}(v_*)$, with $v_*$ such that $\delta_1$ is small.

\begin{proof}
Using the Cauchy-Schwarz inequality twice for the remaining part of the bulk term we obtain
\bea
|\tilde{K}^{S}|&\leq&\frac{1}{r}\left[ \left(1 +\frac{|u|^p}{v^p}\right){v^p}(\partial_v \phi)^2+ \left(1 +\frac{v^p}{|u|^p}\right){|u|^p}(\partial_u \phi)^2\right],
\eea
with the related volume element  
\bea
\dV&=&r^2{\frac{\Omega^2}{2}} \md u\md v\md \sigma^2_{\bbS}.
\eea
Note that the currents related to the vector field $S$ with their related volume elements are given by 
\bea
J_{\mu}^{S}(\phi) n^{\mu}_{v=\bar{v}}  &=&{\frac{2}{\Omega^2}}\left[|u|^p(\partial_u \phi)^2+ \frac{\Omega^2}{4}\bar{v}^p|\nabb \phi|^2\right], \quad \dV_{v=\bar{v}}=r^2{\frac{\Omega^2}{2}} \md \sigma_{\mathbb S^2}\md u,\\
J_{\mu}^{S}(\phi) n^{\mu}_{u=\bar{u}}  &=&{\frac{2}{\Omega^2}}\left[v^p(\partial_v \phi)^2+ \frac{\Omega^2}{4}|\bar{u}|^p|\nabb \phi|^2\right],\quad \dV_{u=\bar{u}}=r^2{\frac{\Omega^2}{2}} \md \sigma_{\mathbb S^2}\md v,
\eea
cf.~ Appendix \ref{Jcurrents}.
Taking the integral over the spacetime region yields
\bea
\lb{knboundhier}
\int\limits_ {\cR_{IV}}  |\tilde{K}^{S}(\phi)| \dV &\leq& \int\limits^{u_{\gamma}(v_*)}_{u_{\gamma}(\bar{v})}\int\limits_{\left\lbrace  v_* \leq v \leq \hat{v}\right\rbrace }\frac{ \Omega^2(\bar{u},\bar{v})}{2r}\left(1 +\frac{|{\bar{u}}|^p}{{\bar{v}}^p}\right)
J_{\mu}^{S}(\phi) n^{\mu}_{u=\bar{u}}\dV_{u=\bar{u}} \md \bar{u}\nonumber \\
&&+ \int\limits_{\bar{v}}^{\hat{v}}\int\limits_{\left\lbrace  u_{\gamma}(\hat{v}) \leq u \leq u_{\gamma}(v_*)\right\rbrace } \frac{ \Omega^2(\bar{u},\bar{v})}{2r}
\left(1 +\frac{{\bar{v}}^p}{|{\bar{u}}|^p}\right) J_{\mu}^{S}(\phi) n^{\mu}_{v=\bar{v}}\dV_{v=\bar{v}}\md \bar{v},
\eea
with $u_{\gamma}(v)$ in the integration limits as defined in \eqref{notation_neu}.

Note the following general relation for positive functions $f(\bar{u},\bar{v})$ and $g(\bar{u},\bar{v})$
\bea
\lb{uintegral}
\int\limits^{u_{\gamma}(v_*)}_{u_{\gamma}(\hat{v})}\int\limits_{\bar{v}}^{\hat{v}} f(\bar{u},\bar{v})g(\bar{u},\bar{v})\md \bar{v}\md \bar{u}
&\leq& \int\limits^{u_{\gamma}(v_*)}_{u_{\gamma}(\hat{v})}\int\limits_{\bar{v}}^{\hat{v}} \left[\sup_{v_{\gamma}(\bar{u})\leq \bar{v} \leq \hat{v}}f(\bar{u},\bar{v})\right]g(\bar{u},\bar{v})\md \bar{v}\md \bar{u}\nonumber\\
&\leq& \int\limits^{u_{\gamma}(v_*)}_{u_{\gamma}(\hat{v})}\left[ \sup_{v_{\gamma}(\bar{u})\leq \bar{v} \leq \hat{v}}f(\bar{u},\bar{v})\right]\int\limits_{\bar{v}}^{\hat{v}}g(\bar{u},\bar{v})\md \bar{v}\md \bar{u}\nonumber\\
&\leq& \int\limits^{u_{\gamma}(v_*)}_{u_{\gamma}(\hat{v})} \sup_{v_{\gamma}(\bar{u})\leq \bar{v} \leq \hat{v}}f(\bar{u},\bar{v})\md \bar{u}\sup_{u_{\gamma}(\hat{v})\leq \bar{u}\leq u_{\gamma}(v_*)}\int\limits_{\bar{v}}^{\hat{v}}g(\bar{u},\bar{v})\md \bar{v}\nonumber\\
&\leq& \int\limits^{u_{\gamma}(v_*)}_{u_{\gamma}(\hat{v})} \sup_{v_{\gamma}(\bar{u})\leq \bar{v} \leq \hat{v}}f(\bar{u},\bar{v})\md \bar{u}\sup_{u_{\gamma}(\hat{v})\leq \bar{u}\leq u_{\gamma}(v_*)}\int\limits_{v_*}^{\hat{v}}g(\bar{u},\bar{v})\md \bar{v}.
\eea
Similarly, it immediately follows that
\bea
\lb{vintegral}
\int\limits^{\hat{v}}_{v_*}\int\limits^{u_{\gamma}(v_*)}_{u_{\gamma}(\bar{v})} f(\bar{u},\bar{v})g(\bar{u},\bar{v})\md \bar{u}\md \bar{v}
&\leq&\int\limits^{\hat{v}}_{v_*}\sup_{u_{\gamma}(\bar{v}) \leq \bar{u}\leq u_{\gamma}(v_*)}f(\bar{u},\bar{v})\md \bar{v}\sup_{v_*\leq \bar{v} \leq \hat{v}}\int\limits^{u_{\gamma}(v_*)}_{u_{\gamma}(\bar{v})}g(\bar{u},\bar{v})\md \bar{u}.
\eea

Using \eqref{uintegral} and \eqref{vintegral} in \eqref{knboundhier} we obtain
\bea
\lb{131}
\int\limits_ {\cR_{IV}}  |\tilde{K}^{S}(\phi)| \dV&\leq& \int\limits^{u_{\gamma}(v_*)}_{u_{\gamma}(\hat{v})} \sup_{v_{\gamma}(\bar{u})\leq \bar{v} \leq \hat{v}}
\left[\frac{ \Omega^2(\bar{u},\bar{v})}{2r}\left(1+ \frac{|{\bar{u}}|^p}{{\bar{v}}^p}\right)\right]\md \bar{u} \sup_{u_{\gamma}(\hat{v})\leq \bar{u}\leq u_{\gamma}(v_*)}\int\limits_{\left\lbrace  v_* \leq v \leq \hat{v}\right\rbrace }J_{\mu}^{S}(\phi) n^{\mu}_{u=\bar{u}}\dV_{u=\bar{u}}
\nonumber\\
&&+  \int\limits^{\hat{v}}_{v_*}\sup_{u_{\gamma}(\bar{v}) \leq \bar{u}\leq u_{\gamma}(v_*)}\left[\frac{ \Omega^2(\bar{u},\bar{v})}{2r}
\left( 1 + \frac{{\bar{v}}^p}{|{\bar{u}}|^p}\right) \right]\md \bar{v}\sup_{v_*\leq \bar{v} \leq \hat{v}}\int\limits_{\left\lbrace  u_{\gamma}(\bar{v}) \leq u \leq u_{\gamma}(v_*)\right\rbrace } J_{\mu}^{S}(\phi) n^{\mu}_{v=\bar{v}}\dV_{v=\bar{v}}.\nonumber\\
&&
\eea
It remains to show finiteness and smallness of \mbox{$\int\limits^{u_{\gamma}(v_*)}_{u_{\gamma}(\hat{v})} \sup_{v_{\gamma}(\bar{u})\leq \bar{v} \leq \hat{v}}
\left[\frac{ \Omega^2(\bar{u},\bar{v})}{2r}\left(1+ \frac{|{\bar{u}}|^p}{{\bar{v}}^p}\right)\right]\md \bar{u} $} and \mbox{$\int\limits^{\hat{v}}_{v_{\gamma}(\bar{u})}\sup_{u_{\gamma}(\bar{v}) \leq \bar{u}\leq u_{\gamma}(v_*)}\left[\frac{ \Omega^2(\bar{u},\bar{v})}{2r}
\left( 1 + \frac{{\bar{v}}^p}{|{\bar{u}}|^p}\right) \right]\md \bar{v}$}.
Earlier we obtained the relation \eqref{spacevolumedecay_u} for $\Omega^2$ in region $\cR_{IV}$. Therefore, we can write
\bea
 \int\limits_{u_{\gamma}(\hat{v})}^{u_{\gamma}(v_*)}
\sup_{v_{\gamma}(\bar{u})\leq \bar{v} \leq \hat{v}}\left[\frac{ \Omega^2(\bar{u},\bar{v})}{2r}
\left( 1 + \frac{|{\bar{u}}|^p}{{\bar{v}}^p}\right) \right]\md \bar{u}
&\leq& C \int\limits_{u_{\gamma}(\hat{v})}^{u_{\gamma}(v_*)}\sup_{v_{\gamma}(\bar{u})\leq \bar{v} \leq \hat{v}}\left[{|u_{\gamma}(\bar{v})|}^{-\beta \alpha} e^{\beta\left[u_{\gamma}(\bar{v})-\bar{u}\right]}\left( 1 + \frac{|{\bar{u}}|^p}{{\bar{v}}^p}\right)\right] \md \bar{u}\nonumber\\
&\leq& \tilde{C} \int\limits_{u_{\gamma}(\hat{v})}^{u_{\gamma}(v_*)}{|\bar{u}|}^{-\beta \alpha} \left( 1 + \frac{{|\bar{u}|}^p}{v_*^p}\right) \md \bar{u}\nonumber\\
\lb{integral}
&\leq& \tilde{\tilde{C}} \int\limits_{u_{\gamma}(\hat{v})}^{u_{\gamma}(v_*)}{|\bar{u}|}^{-\beta \alpha+p}\md \bar{u}\nonumber\\
&\leq& \frac{\tilde{\tilde{C}}}{|-\beta\alpha+p+1|}\left[{|\bar{u}|}^{-\beta\alpha+p+1}\right]_{u_{\gamma}(\hat{v})}^{u_{\gamma}(v_*)}\nonumber\\
&\leq& \delta_1, 
\eea
where $\delta_1\rightarrow 0$ for $|u_{\gamma}(v_*)|\rightarrow -\infty$ and thus for $v_*\rightarrow \infty$. Note that we have here used \eqref{alpha}.

For finiteness of the second term in \eqref{131} we follow the same strategy and use \eqref{spacevolumedecay} for the second term to obtain
\bea
 \int\limits^{\hat{v}}_{v_*}\sup_{u_{\gamma}(\bar{v}) \leq \bar{u}\leq u_{\gamma}(v_*)}\left[\frac{ \Omega^2(\bar{u},\bar{v})}{2r}
\left( 1 + \frac{{\bar{v}}^p}{|{\bar{u}}|^p}\right) \right]\md \bar{v}
&\leq& C \int\limits^{\hat{v}}_{v_*}\sup_{u_{\gamma}(\bar{v}) \leq \bar{u}\leq u_{\gamma}(v_*)}
\left[\bar{v}^{-\beta\alpha}\left( 1 + \frac{{\bar{v}}^p}{|{\bar{u}}|^p}\right)\right]
\md \bar{v}\nonumber\\
&\leq& \tilde{C}\int\limits^{\hat{v}}_{v_*}{\bar{v}}^{-\beta\alpha}\left( 1 + \frac{{\bar{v}}^p}{|u_{\gamma}(v_*)|^p}\right)\md \bar{v}\nonumber\\
&\leq& \frac{\tilde{\tilde{C}}}{|-\beta\alpha+p+1|}\left[{|\bar{v}|}^{-\beta\alpha+p+1}\right]^{\hat{v}}_{v_*}\nonumber\\
&\leq& \delta_2, 
\eea
where $\delta_2 \rightarrow 0$ for $v_* \rightarrow \infty$.
Therefore, we obtain the statement of Lemma \ref{K_spher}.
\end{proof}

\subsection{Energy estimates from $\gamma$ up to $\cC\cH^+$ in the neighbourhood of $i^+$}
\lb{cauchy}
Now we come to the actual proof of weighted energy boundedness up to the Cauchy horizon.

\begin{prop}
\label{kastle}
Let $\phi$ be as in Theorem \ref{anfang} and $p$ as in \eqref{waspist}. Then, for $u_{\schere}$ sufficiently close to $-\infty$, for all $v_*\geq v_{\gamma}(u_{\schere})$ and $\hat{v}>v_*$ 
\bea
\lb{propgamma}
\int\limits_ {\left\lbrace u_{\gamma}(\hat{v})\leq u \leq u_{\gamma}(v_*)\right\rbrace }  J_{\mu}^{S}(\phi) n^{\mu}_{v=\hat{v}} \dV_{v=\hat{v}}
&+&\int\limits_ {\left\lbrace  v_* \leq v \leq \hat{v}\right\rbrace }  J_{\mu}^{S}(\phi) n^{\mu}_{u=u_{\gamma}(v_*)} \dV_{u=u_{\gamma}(v_*)}\leq C {v_*}^{-1-2\delta+p},  
\eea
where $C$ is a positive constant depending on $C_{0}$ of Theorem \ref{anfang} and $D_{0}(u_{\diamond}, 1)$ of Proposition \ref{initialdataprop}, where $u_{\diamond}$ is defined by $r_{red}=r(u_{\diamond},1)$.
\end{prop}
{\em Remark.} Refer to \eqref{notation_neu} for the definition of $u_{\gamma}(v)$ and see Figure \ref{integralbild2} for further clarification.

\begin{proof}
In Section \ref{blueshift1}, Corollary \ref{cor2}, we have obtained the global estimate \eqref{Nstargammaesti} for the weighted ${S_0}$ current which follows from Proposition \ref{to_gamma}. Recall that in $\cD^+(\gamma)$ we have $|u|^p\leq v^p$, cf.~ Section \ref{gamma_curve}, which immediately leads to
\bea
\lb{vergleichN_Nstar}
\int\limits_ {\left\lbrace  v_* \leq v \leq \hat{v} \right\rbrace } J_{\mu}^{{S}}(\phi) n^{\mu}_{\gamma} \dV_{\gamma}\leq \tilde{C}\int\limits_ {\left\lbrace  v_* \leq v \leq \hat{v} \right\rbrace }v^p J_{\mu}^{{S_0}}(\phi) n^{\mu}_{\gamma} \dV_{\gamma},
\eea
cf.~ Appendix \ref{Jcurrents} for explicit expressions of $J_{\mu}^{{S_0}}(\phi) n^{\mu}_{\gamma}$ and $J_{\mu}^{{S}}(\phi) n^{\mu}_{\gamma} $.
From \eqref{vergleichN_Nstar} we see that 
\bea
\lb{hypo}
\int\limits_ {\left\lbrace  v_* \leq v \leq \hat{v} \right\rbrace } J_{\mu}^{S}(\phi) n^{\mu}_{\gamma} \dV_{\gamma} \leq C {v_*}^{-1-2\delta+p} \quad \mbox{for all $v_*>\alpha$ and $p$ as in \eqref{waspist},}
\eea
is implied by Corollary \ref{cor2}.

Let $v_*>2\alpha$ and $\hat{v}>v_*$.
In order to obtain \eqref{propgamma} we consider a region \mbox{$\cR_{IV}=J^+(\gamma)\cap J^-(x)$} with \mbox{$x=(u_{\gamma}(v_*), \hat{v})$}, $x \in \cB$, as shown in
Figure \ref{RN_mit_u_v}. Applying the divergence theorem we obtain
\bea
&&\int\limits_ {\cR_{IV}}  K^{S}(\phi) \dV
+\int\limits_ {\left\lbrace  u_{\gamma}(\hat{v})\leq u \leq u_{\gamma}(v_*)\right\rbrace }  J_{\mu}^{S}(\phi) n^{\mu}_{v=\hat{v}} \dV_{v=\hat{v}}
+\int\limits_ {\left\lbrace  v_* \leq v \leq \hat{v}\right\rbrace }  J_{\mu}^{S}(\phi) n^{\mu}_{u=u_{\gamma}(v_*)} \dV_{u=u_{\gamma}(v_*)}\nonumber\\
&=&
 \int\limits_{\left\lbrace  v_* \leq v \leq \hat{v}\right\rbrace }   J_{\mu}^{S}(\phi) n^{\mu}_{\gamma} \dV_{\gamma}.
\eea
In Section \ref{knsec} we found that the angular part of $K^{S}(\phi)$ is positive in $\cR_{IV}$ and we called the remaining part $\tilde{K}^{S}(\phi) $ given in \eqref{KNtilde}. Using \eqref{knrelation} we can therefore write
\bea
\lb{cons_law}
&&\int\limits_ {\left\lbrace  u_{\gamma}(\hat{v})\leq u \leq u_{\gamma}(v_*)\right\rbrace }  J_{\mu}^{S}(\phi) n^{\mu}_{v=\hat{v}} \dV_{v=\hat{v}}
+\int\limits_ {\left\lbrace  v_* \leq v \leq \hat{v}\right\rbrace }  J_{\mu}^{S}(\phi) n^{\mu}_{u=u_{\gamma}(v_*)} \dV_{u=u_{\gamma}(v_*)}\nonumber\\
&\leq& \int\limits_ {\cR_{IV}}  |\tilde{K}^{S}(\phi)| \dV
+ \int\limits_{\left\lbrace  v_* \leq v \leq \hat{v}\right\rbrace }   J_{\mu}^{S}(\phi) n^{\mu}_{\gamma} \dV_{\gamma}.
\eea
Using Lemma \ref{K_spher} we obtain
\bea
\lb{start}
&&\int\limits_ {\left\lbrace  u_{\gamma}(\hat{v})\leq u \leq u_{\gamma}(v_*)\right\rbrace }  J_{\mu}^{S}(\phi) n^{\mu}_{v=\hat{v}} \dV_{v=\hat{v}}
+\int\limits_ {\left\lbrace  v_* \leq v \leq \hat{v}\right\rbrace }  J_{\mu}^{S}(\phi) n^{\mu}_{u=u_{\gamma}(v_*)} \dV_{u=u_{\gamma}(v_*)}\nonumber\\
&\leq&\delta_1 \sup_{ u_{\gamma}(\hat{v})\leq \bar{u} \leq u_{\gamma}(v_*)}\int\limits_ {\left\lbrace  v_{\gamma}(\bar{u}) \leq v \leq \hat{v} \right\rbrace }  J_{\mu}^{S}(\phi) n^{\mu}_{u=\bar{u}}\dV_{u=\bar{u}}+  \delta_2 \sup_{v_*\leq \bar{v} \leq \hat{v}}\int\limits_ {\left\lbrace  u_{\gamma}(\hat{v}) \leq u \leq u_{\gamma}(\bar{v}) \right\rbrace } J_{\mu}^{S}(\phi) n^{\mu}_{v=\bar{v}} \dV_{v=\bar{v}}\nonumber\\
&&+ \int\limits_{\left\lbrace  v_* \leq v \leq \hat{v}\right\rbrace }   J_{\mu}^{S}(\phi) n^{\mu}_{\gamma} \dV_{\gamma}.
\eea
Repeating estimate \eqref{start} with $\bar{u}$, $\bar{v}$ in place of $u_{\gamma}(v_*)$, $\hat{v}$ and taking the supremum we have
\bea
\lb{140}
&&\sup_{u_{\gamma}(\hat{v})\leq \bar{u}\leq u_{\gamma}(v_*)}\int\limits_ {\left\lbrace  v_{\gamma}(\bar{u}) \leq v \leq \hat{v}\right\rbrace }  J_{\mu}^{S}(\phi) n^{\mu}_{u=\bar{u}}\dV_{u=\bar{u}}
+ \sup_{v_*\leq \bar{v} \leq \hat{v}}\int\limits_ {\left\lbrace  u_{\gamma}(\hat{v})\leq u \leq u_{\gamma}(\bar{v})\right\rbrace } J_{\mu}^{S}(\phi) n^{\mu}_{v=\bar{v}} \dV_{v=\bar{v}}\nonumber\\
&\leq&\delta_1 \sup_{ u_{\gamma}(\hat{v})\leq \bar{u} \leq u_{\gamma}(v_*)}\int\limits_ {\left\lbrace  v_{\gamma}(\bar{u}) \leq v \leq \hat{v}\right\rbrace }   J_{\mu}^{S}(\phi) n^{\mu}_{u=\bar{u}}\dV_{u=\bar{u}} +  \delta_2 \sup_{v_*\leq \bar{v} \leq \hat{v}}\int\limits_ {\left\lbrace  u_{\gamma}(\hat{v}) \leq u \leq u_{\gamma}(\bar{v}) \right\rbrace }J_{\mu}^{S}(\phi) n^{\mu}_{v=\bar{v}} \dV_{v=\bar{v}}\nonumber\\
&&+ \int\limits_{\left\lbrace  v_* \leq v \leq \hat{v}\right\rbrace }   J_{\mu}^{S}(\phi) n^{\mu}_{\gamma} \dV_{\gamma}.
\eea
Recalling $\delta_1\rightarrow 0$, $\delta_2\rightarrow 0$ as $v_*\rightarrow \infty$, choose $u_{\schere}$ sufficiently close to $-\infty$, such that for $v_*>v_{\gamma}(u_{\schere})$, say 
\bea
 \delta_1, \delta_2 \leq \frac12
\eea
holds.
The conclusion of Proposition \ref{kastle} then follows by absorbing the first two terms of the right hand side of \eqref{140} by the two terms on the left and estimating the third from \eqref{hypo}.
\end{proof}

\subsection{Energy estimates globally in the rectangle $\Xi$ up to $\cC\cH^+$ in the neighbourhood of $i^+$}
\lb{allxi}
In the previous Sections \ref{first_section} to \ref{cauchy} we have proven energy estimates for each region with specific properties separately. Putting all results together we can state the following proposition. 
\begin{prop}
\label{gesamtesti}
Let $\phi$ be as in Theorem \ref{anfang} and $p$ as in \eqref{waspist}. Then, for $u_{\schere}$ sufficiently close to $-\infty$, for all $v_*>1$, $\hat{v}>v_*$ and $\tilde{u} \in (-\infty, u_{\schere})$.
\bea
\lb{propgamma2}
\int\limits_ {\left\lbrace  v_* \leq v \leq \hat{v}\right\rbrace }  J_{\mu}^{S}(\phi) n^{\mu}_{u=\tilde{u}} \dV_{u=\tilde{u}}\leq {C} {v_*}^{-1-2\delta+p},
\eea
where ${C}$ is a positive constant depending on $C_0$ of Theorem \ref{anfang} and $D_{0}(u_{\diamond},1)$ of Proposition \ref{initialdataprop}, where $u_{\diamond}$ is defined by $r_{red}=r(u_{\diamond},1)$.
\end{prop}
\begin{proof}
First of all we partition the integral of the statement into a sum of integrals of the different regions. That is to say
\ben
\int\limits_ {\left\lbrace  v_* \leq v \leq \hat{v}\right\rbrace }  J_{\mu}^{S}(\phi) n^{\mu}_{u=\tilde{u}} \dV_{u=\tilde{u}}&=&\int\limits_ {{\left\lbrace  v_* \leq v \leq \hat{v}\right\rbrace } \cap \cR} J_{\mu}^{S}(\phi) n^{\mu}_{u=\tilde{u}} \dV_{u=\tilde{u}}
+\int\limits_ {{\left\lbrace  v_* \leq v \leq \hat{v}\right\rbrace } \cap \cN} J_{\mu}^{S}(\phi) n^{\mu}_{u=\tilde{u}} \dV_{u=\tilde{u}}\\
&&+\int\limits_ {{\left\lbrace  v_* \leq v \leq \hat{v}\right\rbrace } \cap J^-(\gamma)\cap\cB} J_{\mu}^{S}(\phi) n^{\mu}_{u=\tilde{u}} \dV_{u=\tilde{u}}
+\int\limits_ {{\left\lbrace  v_* \leq v \leq \hat{v}\right\rbrace } \cap J^+(\gamma)\cap\cB} J_{\mu}^{S}(\phi) n^{\mu}_{u=\tilde{u}} \dV_{u=\tilde{u}}.
\een
For the integral in $\cR$ and the integral in $\cN$ we use Corollaries \ref{cor5.2} and \ref{cor6.1}. (Note that the former has to be summed resulting in the loss of one polynomial power.)
Further, for the integral in region $J^-(\gamma)\cap\cB$ we apply Corollary \ref{cor7.1} and for the integral in region $J^+(\gamma)\cap\cB$ we use Proposition \ref{kastle}. 
Putting all this together we arrive at the conclusion of Proposition \ref{gesamtesti}.
\end{proof}

In particular, we have
\begin{cor}
\label{endeacht}
Let $\phi$ be as in Theorem \ref{anfang} and $p$ as in \eqref{waspist}. Then, for $u_{\schere}$ sufficiently close to $-\infty$, for all \mbox{$v_{fix}\geq1$}, and $\tilde{u} \in (-\infty, u_{\schere})$,
\bea
\lb{propgamma2}
&&\int\limits_{\bbS^2}\int\limits^{\infty}_{v_{fix}}\left[ v^p (\partial_v \phi)^2(\tilde{u}, v) + |\nabb \phi|^2(\tilde{u}, v) \right]r^2\md v\md \sigma_{\mathbb S^2}\leq C,
\eea
where $C$ is a positive constant dependent on $C_{0}$ of Theorem \ref{anfang} and $D_{0}(u_{\diamond},1)$ of Proposition \ref{initialdataprop}, where $u_{\diamond}$ is defined by $r_{red}=r(u_{\diamond},1)$.
\end{cor}
\begin{proof}
The conclusion of the proposition follows immediately examining the weights in Proposition \ref{gesamtesti}.
\end{proof}

\section{Pointwise estimates from higher order energies}
\lb{nineten}
\subsection{The $\leo$ notation and Sobolev inequality on spheres}
\lb{leonotationsec}
Recall that we had stated the expressions for the generators of spherical symmetry $\leo_i$, $i=1,2,3$, in Section \ref{angular}. They were explicitly given by \eqref{angmom1} to \eqref{angmom2}. Further, having expressions \eqref{sum1} and \eqref{sum2} in mind we introduce the following notation
\bea
\lb{leonotation}
 \sum_{k=0}^{2} \left(\leo^k \phi\right)^2=|\phi|^2+\sum_{i=1}^{3} \left(\leo_i \phi\right)^2+\sum_{i=1}^{3} \sum_{j=1}^{3}\left(\leo_i \leo_j\phi\right)^2,
\eea
where $k$ has to be understood as the order of an exponent and not as an index.
By Sobolev embedding on the standard spheres we have in this notation 
\bea
\lb{sobo_embed}
\sup_{\theta,\varphi \in \bbS^2}|\phi(u,v,\theta,\varphi)|^2\leq \tilde{C} \sum_{k=0}^{2} \int\limits_{\bbS^2} \left(\leo^k \phi\right)^2(u,v,\theta,\varphi)\md \sigma_{\mathbb S^2}, 
\eea
which means that we can derive a pointwise estimate from an estimate of the integrals on the spheres, see e.g.~ \cite{m_red}. 
More generally, in the following we will also use the notation 
\bea
J_\mu^X(\leo\phi) = \sum_{i=1}^{3} J_\mu^X(\leo_i\phi),
\eea
for any $J$-current related to an arbitrary vector field $X$, and similarly for other quadratic expressions, e.g.~ \eqref{energy3_hier}, \eqref{energy3_hier2}. 

\subsection{Higher order energy estimates in the neighbourhood of $i^+$}
\lb{higherorder_neu}
We will need the following extension of Corollary \ref{endeacht} for higher order energies.
\begin{thm}
\lb{energythm2}
On subextremal Reissner-Nordstr\"om spacetime $(\cM,g)$, with mass $M$ and charge $e$ and $M>|e|\neq 0$, let $\phi$ be 
a solution of the wave equation $\Box_g \phi=0$ arising from
sufficiently regular Cauchy data on a two-ended asymptotically flat Cauchy surface $\Sigma$. Then, for $v_{fix} \geq 1$ and $u_{fix} > -\infty$
\bea
\lb{energy2}
&&\int\limits_{\bbS^2}\int\limits^{\infty}_{v_{fix}}\left[v^p (\partial_v \phi)^2(u_{fix}, v, \theta, \varphi) + |\nabb \phi|^2(u_{fix}, v, \theta, \varphi) \right]r^2\md v\md \sigma_{\mathbb S^2}\leq E_{0},\\
\lb{energy3_hier}
&&\int\limits_{\bbS^2}\int\limits^{\infty}_{v_{fix}} \left[v^p (\partial_v \leo\phi)^2(u_{fix}, v, \theta, \varphi) + |\nabb \leo\phi|^2(u_{fix}, v, \theta, \varphi) \right]r^2\md v\md \sigma_{\mathbb S^2}\leq E_{1},\\
\lb{energy3_hier2}
&&\int\limits_{\bbS^2}\int\limits^{\infty}_{v_{fix}}\left[ v^p (\partial_v \leo^2\phi)^2 (u_{fix}, v, \theta, \varphi)+ |\nabb \leo^2\phi|^2(u_{fix}, v, \theta, \varphi) \right]r^2\md v\md \sigma_{\mathbb S^2}\leq E_{2},
\eea
where $p$ is as in \eqref{waspist}. 
\end{thm}
\begin{proof}
Statement \eqref{energy2} was already derived in Corollary \ref{endeacht}.
Recall that $\leo_i\phi$, $\leo_i\leo_j\phi$ also satisfy the massless scalar wave equation, cf.~ Section \ref{en_cur}. Summing over all angular momentum operators, keeping in mind notation \eqref{leonotation}, etc., we therefore obtain \eqref{energy3_hier} and \eqref{energy3_hier2}.
\end{proof}

\subsection{Pointwise boundedness in the neighbourhood of $i^+$}
\lb{uni_bounded}
We turn the discussion to the derivation of pointwise boundedness from energy estimates. In particular we prove Theorem \ref{dashier} from Theorem \ref{energythm2}.

By the fundamental theorem of calculus it follows for all $v_*>1$, $\hat{v} >v_*$ and $u \in (-\infty, u_{\schere})$ that
\bea
\phi(u,\hat{v}, \theta, \varphi)&=& \int\limits_{v_*}^{\hat{v}} \left(\partial_v \phi\right)(u, v, \theta, \varphi) \md v +\phi(u, v_*, \theta, \varphi),\nonumber\\
&\leq& \int\limits_{v_*}^{\hat{v}} (\partial_v \phi)(u, v, \theta, \varphi)v^{\frac{p}{2}}v^{-\frac{p}{2}}\md v+\phi(u, v_*, \theta, \varphi),\nonumber\\
&\leq& \left(\int\limits_{v_*}^{\hat{v}} v^p(\partial_v \phi)^2(u, v, \theta, \varphi)\md v\right)^{\frac12}\left(\int\limits_{v_*}^{\hat{v}} v^{-{p}}\md v\right)^{\frac12}+\phi(u, v_*, \theta, \varphi),\nonumber\\
\eea
where we have used the Cauchy-Schwarz inequality in the last step.
Squaring the entire expression, using Cauchy-Schwarz again and integrating over ${\mathbb S}^2$ we obtain the expression that we had sketched in Section \ref{outline} already
\bea
\int\limits_{\bbS^2} \phi^2(u,\hat{v})\md \sigma_{\mathbb S^2}
&\leq& \tilde{C}\left[\int\limits_{\bbS^2}\left(\int\limits_{v_*}^{\hat{v}} v^p(\partial_v \phi)^2(u,v)\md v\int\limits_{v_*}^{\hat{v}} v^{-{p}}\md v\right)r^2\md \sigma_{\mathbb S^2}+\int\limits_{\bbS^2} \phi^2(u,v_*)\md \sigma_{\mathbb S^2}\right],
\eea
with $p$ as in \eqref{waspist} and the first term on the right hand side controlled by the flux for which we derived boundedness in Section \ref{innerhorizon}.
Therefore, by using Theorem \ref{energythm2} we obtain
\bea
\lb{fundcauchy}
\int\limits_{\bbS^2} \phi^2(u,\hat{v})\md \sigma_{\mathbb S^2}&\leq&\tilde{C}\left[ E_{0}\int\limits_{\bbS^2}\int\limits_{v_*}^{\hat{v}} v^{-p}\md v\md \sigma_{\mathbb S^2}+\int\limits_{\bbS^2} \phi^2(u,v_*)\md \sigma_{\mathbb S^2}\right]\nonumber\\
&\leq&\tilde{C}\left[\tilde{\tilde{C}} E_{0}+\int\limits_{\bbS^2} \phi^2(u,v_*)\md \sigma_{\mathbb S^2}\right].
\eea 
It is here that we have used the requirement $p>1$ of \eqref{waspist}.
Applying all our estimates to $\leo_i \phi$, $\leo_i\leo_j \phi$ and summing, we obtain in the notation of Section \ref{leonotationsec} the following:
\bea
\lb{fundcauchy1}
\int\limits_{\bbS^2} (\leo \phi) ^2(u,\hat{v})\md \sigma_{\mathbb S^2} 
&\leq&\tilde{C}\left[\tilde{\tilde{C}} E_{1}+\int\limits_{\bbS^2} (\leo \phi) ^2(u,v_*)\md \sigma_{\mathbb S^2}\right],\\
\lb{fundcauchy2}
\int\limits_{\bbS^2} ({\leo}^2 \phi)^2(u,\hat{v})\md \sigma_{\mathbb S^2}
&\leq&\tilde{C}\left[\tilde{\tilde{C}}  E_{2}+\int\limits_{\bbS^2} ({\leo}^2 \phi) ^2(u,v_*)\md \sigma_{\mathbb S^2}\right].
\eea 
Let us now use \eqref{aufeins1} to \eqref{aufeins3} of Proposition \ref{initialdataprop} to estimate the right hand sides of \eqref{fundcauchy} to \eqref{fundcauchy2} with $v_*=1$.
Adding all equations up we derive pointwise boundedness according to \eqref{sobo_embed}
\bea
\lb{supr}
\sup_{\bbS^2}|\phi(u,\hat{v},\theta,\varphi)|^2&\leq& \tilde{C} \left[\int\limits_{\bbS^2} ( \phi) ^2(u,\hat{v})\md \sigma_{\mathbb S^2} +\int\limits_{\bbS^2} (\leo \phi) ^2(u,\hat{v})\md \sigma_{\mathbb S^2} +\int\limits_{\bbS^2} (\leo^2 \phi) ^2(u,\hat{v})\md \sigma_{\mathbb S^2} \right],\\
\lb{supr2}
&\leq&\tilde{C}\left[\tilde{\tilde{C}}\left( E_{0}+  E_{1}+ E_{2}\right)+D_{0}(u_{\diamond}, 1)+D_{1}(u_{\diamond}, 1)+D_{2}(u_{\diamond}, 1)\right],\\
\lb{supr3}
&\leq&C,
\eea
with $C$ depending on the initial data on $\Sigma$. 
We therefore arrive at the statement given in Theorem \ref{dashier}. 
\begin{trivlist}
\item[\hskip \labelsep ]\qed\end{trivlist}

\section{{\it Left} neighbourhood of $i^+$}
\lb{leftinterior}

We now turn to establish boundedness in the neighbourhood of the {\em left} timelike infinity $i^+$. 
For this we simply repeat the entire proof carried out in the characteristic rectangle $\Xi$ in Section \ref{first_section} to Section \ref{innerhorizon} but this time for the region $\tilde{\Xi}$ at the {\it left end} of $\cQ|_{II}$, cf.~ \eqref{QII}, as shown in the Penrose diagram \ref{RN_character_linkes}.
{\begin{figure}[ht]
\centering
\includegraphics[width=0.6\textwidth]{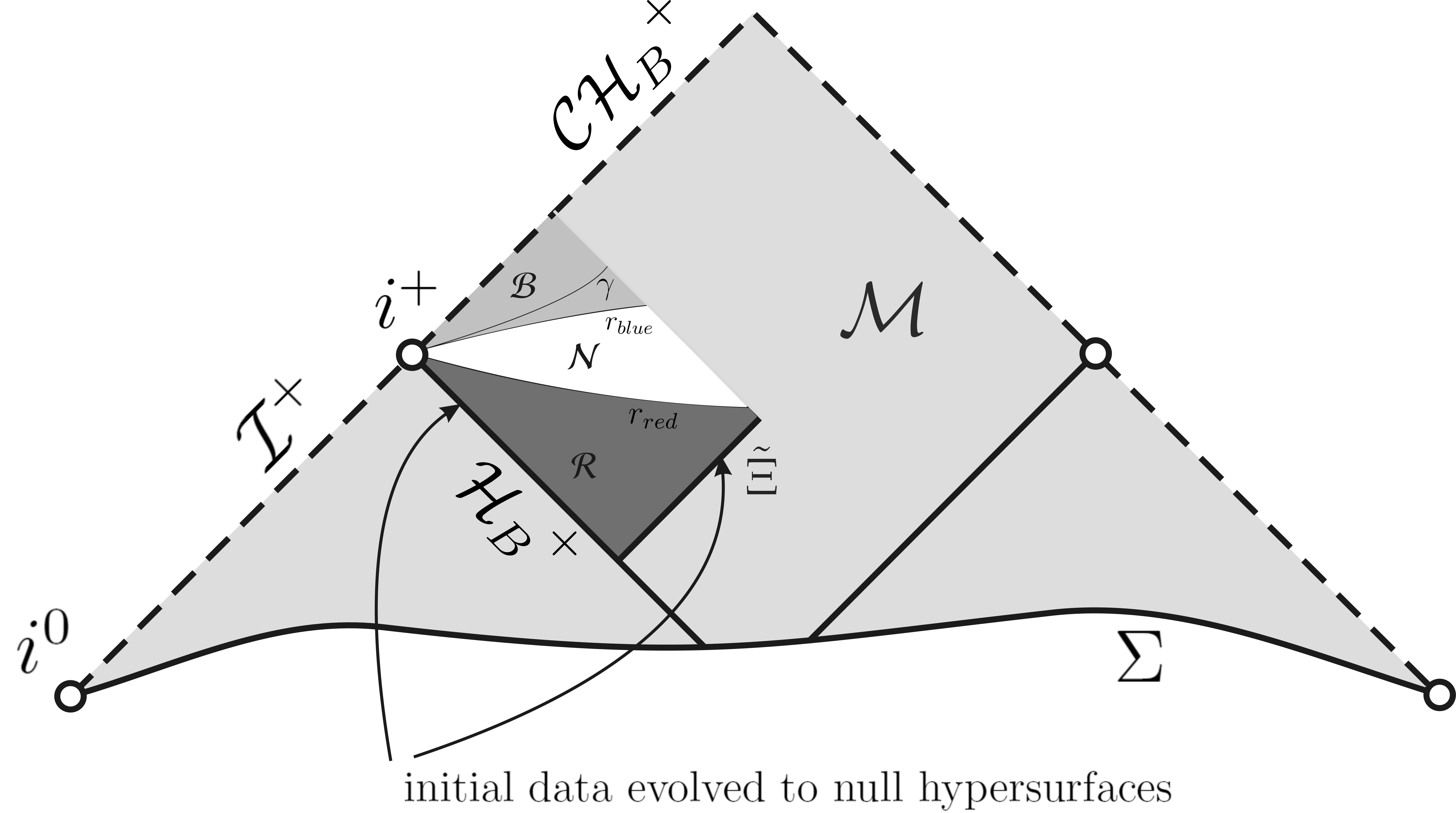}
\caption[]{Penrose diagram with characteristic rectangle $\tilde{\Xi}$ depicted on the {\it left} side.}
\label{RN_character_linkes}\end{figure}}

Since everything is completely analogous to the derivation for the {\it right} side, we will merely state the main Theorems and Propositions and remind the reader that $u$ and $v$ are interchanged now. Recall that $u\rightarrow \infty$ on $\cC\cH_B^+$ and $v\rightarrow -\infty$ on $\cH_B^+$ for the {\it left end}.
The rectangle under consideration is now $\tilde{\Xi}=\left\{(1\leq u< \infty), (-\infty\leq v\leq v_{\schere}) \right\}$
\begin{thm}
\lb{anfangl}
Let $\phi$ be a solution of the wave equation \eqref{wave} on a subextremal Reissner-Nordstr\"om background $(\cM,g)$, with mass $M$ and charge $e$ and $M>|e|\neq 0$, arising from smooth compactly supported initial data on an arbitrary Cauchy hypersurface $\Sigma$, cf.~ Figure \ref{RN_character_linkes}. Then, there exists $\delta>0$ such that
\bea
\lb{thpurl}
\int\limits_{\bbS^2}\int\limits^{u+1}_u \left[(\partial_u \phi)^2(u, -\infty)+|\nabb \phi|^2(u, -\infty)\right]r^2 \md u\md \sigma_{\mathbb S^2}&\leq& {C_{0}}u^{-2-2\delta},\\
\lb{thpur1l}
\int\limits_{\bbS^2}\int\limits^{u+1}_u \left[(\partial_u \leo\phi)^2(u, -\infty) +|\nabb \leo\phi|^2(u, -\infty)\right]r^2\md u\md \sigma_{\mathbb S^2}&\leq& {C_{1}}u^{-2-2\delta},\\
\lb{thpur2l}
\int\limits_{\bbS^2}\int\limits^{u+1}_u \left[(\partial_u \leo^2\phi)^2(u, -\infty) +|\nabb \leo^2\phi|^2(u, -\infty)\right]r^2\md u\md \sigma_{\mathbb S^2}&\leq& {C_{2}}u^{-2-2\delta},
\eea
on ${\cH_B}^+$, for all $u$ and some positive constants ${C_{0}}$, ${C_{1}}$ and ${C_{2}}$ depending on the initial data.
\end{thm}
\begin{prop}
\lb{initialdatapropl}
Let \mbox{$u_{\diamond}, v_{\diamond} \in (-\infty, \infty)$}.
Under the assumption of Theorem \ref{anfangl}, the energy at retarded Eddington-Finkelstein coordinate \mbox{$\left\{u=u_{\diamond}\right\}\cap\left\{{-\infty}\leq v \leq {v_{\diamond}}\right\}$} is bounded from the initial data 
\bea
\lb{proppurl}
\int\limits_{\bbS^2}\int\limits_{-\infty}^{v_{\diamond}}\left[ \Omega^{-2}(\partial_v \phi)^2(u_{\diamond},v)+\frac{\Omega^{2}}{2}|\nabb \phi|^2(u_{\diamond},v)\right]r^2\md v\md \sigma_{\mathbb S^2}&\leq& {D_{0}(u_{\diamond}, v_{\diamond})},\\
\lb{proppur1l}
\int\limits_{\bbS^2}\int\limits_{-\infty}^{v_{\diamond}}\left[\Omega^{-2}(\partial_v \leo\phi)^2(u_{\diamond},v)+\frac{\Omega^{2}}{2}|\nabb  \leo\phi|^2(u_{\diamond},v)\right]r^2\md v\md \sigma_{\mathbb S^2}&\leq& {D_{1}(u_{\diamond}, v_{\diamond})},\\
\lb{proppur2l}
\int\limits_{\bbS^2}\int\limits_{-\infty}^{v_{\diamond}}\left[\Omega^{-2}(\partial_v \leo^2\phi)^2(u_{\diamond},v)+\frac{\Omega^{2}}{2}|\nabb \leo^2\phi|^2(u_{\diamond},v)\right]r^2\md v\md \sigma_{\mathbb S^2}&\leq& {D_{2}(u_{\diamond}, v_{\diamond})},
\eea
and further
\bea
\lb{aufeins1l}
\sup_{-\infty\leq v \leq {v_{\diamond}}}\int\limits_{\bbS^2} (\phi)^2(u_{\diamond},v)\md \sigma_{\mathbb S^2}&\leq&  {D_{0}(u_{\diamond}, v_{\diamond})},\\
\sup_{-\infty\leq v \leq {v_{\diamond}}}\int\limits_{\bbS^2} (\leo\phi)^2(u_{\diamond},v)\md \sigma_{\mathbb S^2}&\leq&  {D_{1}(u_{\diamond}, v_{\diamond})},\\
\lb{aufeins3l}
\sup_{-\infty\leq v \leq {v_{\diamond}}}\int\limits_{\bbS^2} (\leo^2\phi)^2(u_{\diamond},v)\md \sigma_{\mathbb S^2}&\leq&  {D_{2}(u_{\diamond}, v_{\diamond})},
\eea
with ${D_{0}(u_{\diamond}, v_{\diamond})}$, $D_{1}(u_{\diamond}, v_{\diamond})$ and $D_{2}(u_{\diamond}, v_{\diamond})$ positive constants depending on the initial data.
\end{prop}
Note the $\Omega^{-2}$ weights which arise since $v$ is not regular at $\cH^+_B$.
Analogous to the Proposition \ref{gesamtesti} obtained for the {\it right} side we can state the following for the {\it left}.
\begin{prop}
\label{gesamtestil}
Let $\phi$ be as in Theorem \ref{anfangl} and $p$ as in \eqref{waspist}. Then, for $v_{\schere}$ sufficiently close to $-\infty$, for $u_*>1$, $\hat{u}>u_*$ and $\tilde{v} \in (-\infty, v_{\schere})$. 
\bea
\lb{propgamma2l}
\int\limits_ {\left\lbrace  u_* \leq u \leq \hat{u}\right\rbrace }  J_{\mu}^{S}(\phi) n^{\mu}_{v=\tilde{v}} \dV_{v=\tilde{v}}\leq {C} {u_*}^{-1-2\delta+p},
\eea
where ${C}$ is a positive constant depending on $C_0$ of Theorem \ref{anfangl} and $D_{0}(u_{\diamond},1)$ of Proposition \ref{initialdatapropl}, where $v_{\diamond}$ is defined by $r_{red}=r(1, v_{\diamond})$.
\end{prop}
\begin{proof}
The proof is analogous to the proof of Proposition \ref{gesamtesti} with $u$ and $v$ interchanged.
\end{proof}
Having obtained Proposition \eqref{gesamtestil} we can derive higher order estimates analogous to Section \ref{higherorder_neu}. The pointwise estimate is then obtained via the same strategy as in Section \ref{uni_bounded} but integrated in $u$ and not in $v$, and can be stated as follows.
\begin{thm}
\lb{dashierl}
Let $\phi$ be as in Theorem \ref{anfangl}, then
\ben
|\phi|\leq C
\een
locally in the black hole interior up to
$\cC\cH^+$ in a ``small neighbourhood'' of {\em left}  timelike infinity $i^+$,
that is in \mbox{$[1, \infty)\times (-\infty, v_{\schere}] $} for some $v_{\schere}>-\infty$.
\end{thm}

\section{Energy along the future boundaries of $\cR_{V}$} 
\lb{region5_proof}
Let $u_{\diamond}>u_{\schere}$ and $v_*\geq v_{\gamma}(u_{\schere})$. Define \mbox{$\cR_{V}=\left\{u_{\schere}\leq u \leq u_{\diamond}\right\}\cap\left\{ v_*\leq v\leq \hat{v}\right\}$}, cf.~ Figure \ref{region5}, and note that \mbox{$\cR_{V} \subset \cB$}. 
{\begin{figure}[ht]
\centering
\includegraphics[width=0.4\textwidth]{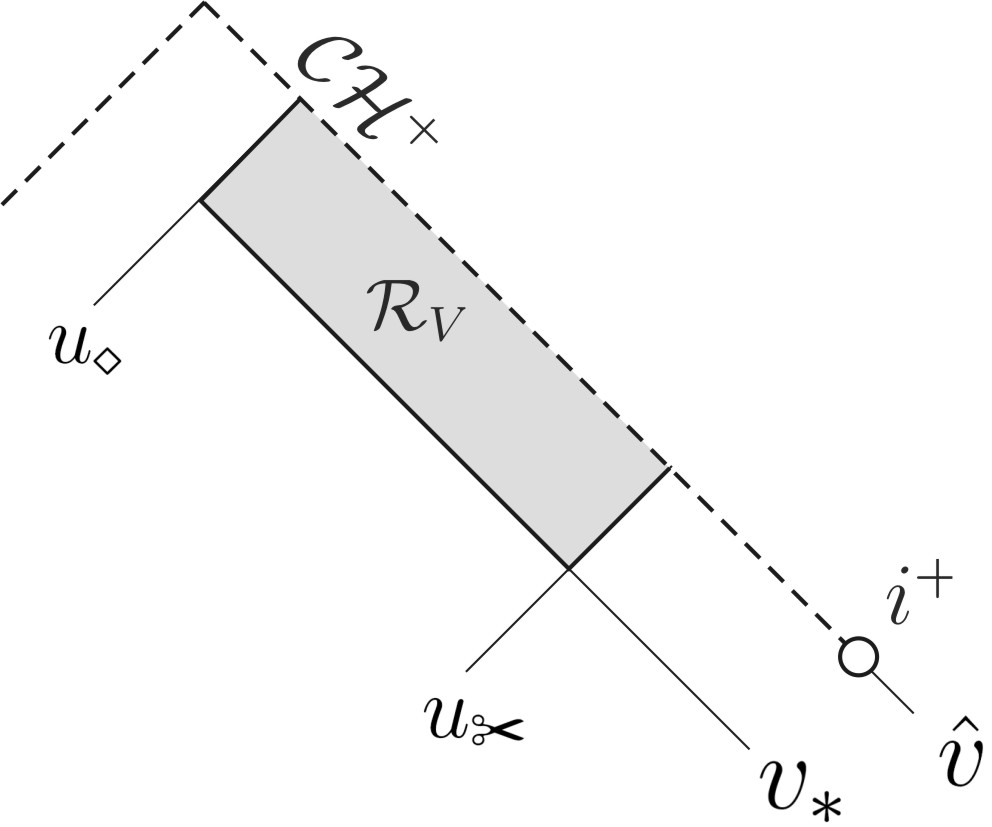}
\caption[]{Penrose diagram depicting region $\cR_{V}$.}
\label{region5}\end{figure}}
We will apply the vector field
\bea
\lb{V_feld}
W=v^p \partial_v+\partial_u
\eea
as a multiplier.
The bulk can be calculated as
\bea
\lb{kv}
K^{W}=&-&\frac{2}{r}[v^p+1](\partial_v\phi\partial_u\phi)\nonumber\\
&-&\left[\frac12 p v^{p-1}+\frac{\partial_v \Omega}{\Omega}v^p+\frac{\partial_u \Omega}{\Omega}\right]|\nabb \phi|^2.
\eea
Let us define
\bea
\lb{tildekv}
\tilde{K}^{W}=&-&\frac{2}{r}[v^p+1](\partial_v\phi\partial_u\phi),
\eea
and
\bea
\lb{Ksphaerisch}
K_{\nabb}^{W}=&-&\left[\frac12 p v^{p-1}+\frac{\partial_v \Omega}{\Omega}v^p+\frac{\partial_u \Omega}{\Omega}\right]|\nabb \phi|^2,
\eea
with $K_{\nabb}^{W}$ positive since the second term in \eqref{Ksphaerisch} dominates over the first for $v_*>2\alpha$ and $\frac{\partial_u \Omega}{\Omega}$, $\frac{\partial_v \Omega}{\Omega}$ are negative in the blueshift region.
We have therefore 
\bea
\lb{167}
-K^{W}\leq|\tilde{K}^{W}|
\eea
in $\cR_V$.
We aim for estimating it via the currents along $v=constant$ and $u=constant$ hypersurfaces.

\begin{lem}
\label{K_spherV}
Let $\phi$ be an arbitrary function. Then, 
for all \mbox{$v_*\geq v_{\gamma}(u_{\schere})$}, 
$\hat{v} >v_*$, 
for $u_{\diamond}\geq u_2>u_1\geq u_{\schere}$ and $\epsilon\geq u_2-u_1>0$
\bea
\lb{Propdelta}
\int\limits_{\cR_{V_1}} |\tilde{K}^{W}| \dV \leq &\delta_1& \sup_{u_{1}\leq \bar{u}\leq u_{2}}\int\limits_ {\left\lbrace  v_*\leq v \leq \hat{v}\right\rbrace } J_{\mu}^{W}(\phi) n^{\mu}_{u=\bar{u}}\dV_{u=\bar{u}}\nonumber \\
+ &\delta_2& \sup_{v_*\leq \bar{v} \leq \hat{v}}\int\limits_ {\left\lbrace  u_1\leq u \leq u_2\right\rbrace } J_{\mu}^{W}(\phi) n^{\mu}_{v=\bar{v}} \dV_{v=\bar{v}},
\eea
where $\cR_{V_1}=\left\{u_1\leq {u} \leq u_2\right\}\cap \cR_{V}$ and $\delta_1$, $\delta_2$ 
are positive constants, depending only on $v_*$ and $\epsilon$ such that $\delta_1\rightarrow 0$ for $\epsilon\rightarrow 0$ and $\delta_2\rightarrow 0$ as ${v}_*\rightarrow \infty$.
\end{lem}
\begin{proof}
Using the Cauchy-Schwarz inequality for equation \eqref{tildekv} we
obtain
\bea
|\tilde{K}^{W}|&\leq&\frac{1}{r}\left[ \left(1 +{v^{-p}}\right){v^p}(\partial_v \phi)^2+ \left(1 +{v^p}\right)(\partial_u \phi)^2\right],
\eea
with the related volume element  
\bea
\dV&=&r^2{\frac{\Omega^2}{2}} \md u\md v\md \sigma^2_{\bbS}.
\eea
Note that the currents related to the vector field $W$ with their related volume elements are given by 
\bea
J_{\mu}^{W}(\phi) n^{\mu}_{v=\bar{v}}  &=&{\frac{2}{\Omega^2}}\left[(\partial_u \phi)^2+ \frac{\Omega^2}{4}\bar{v}^p|\nabb \phi|^2\right], \quad \dV_{v=\bar{v}}=r^2{\frac{\Omega^2}{2}} \md \sigma_{\mathbb S^2}\md u,\\
J_{\mu}^{W}(\phi) n^{\mu}_{u=\bar{u}}  &=&{\frac{2}{\Omega^2}}\left[v^p(\partial_v \phi)^2+ \frac{\Omega^2}{4}|\nabb \phi|^2\right],\quad \dV_{u=\bar{u}}=r^2{\frac{\Omega^2}{2}} \md \sigma_{\mathbb S^2}\md v,
\eea
cf.~ Appendix \ref{Jcurrents}.
Taking the integral over the spacetime region therefore yields
\bea
\lb{knbound}
\int\limits_ {\cR_{V_1}}  |\tilde{K}^{W}(\phi)| \dV &\leq& \int\limits^{u_{2}}_{u_{1}}\int\limits_ {\left\lbrace  v_*\leq v \leq \hat{v}\right\rbrace }\frac{ \Omega^2(\bar{u},\bar{v})}{2r}\left(1 +{\bar{v}}^{-p}\right)
J_{\mu}^{W}(\phi) n^{\mu}_{u=\bar{u}}\dV_{u=\bar{u}} \md \bar{u}\nonumber \\
&&+ \int\limits^{\hat{v}}_{v_*} \int\limits_ {\left\lbrace  u_1\leq u \leq u_2\right\rbrace }\frac{ \Omega^2(\bar{u},\bar{v})}{2r}
\left(1 +{\bar{v}}^p\right) J_{\mu}^{W}(\phi) n^{\mu}_{v=\bar{v}}\dV_{v=\bar{v}}\md \bar{v},\nonumber \\
&\leq& \int\limits^{u_{2}}_{u_1} \sup_{v_*\leq \bar{v} \leq \hat{v}}
\left[\frac{ \Omega^2(\bar{u},\bar{v})}{2r}\left(1+ {{\bar{v}}^{-p}}\right)\right]\md \bar{u} \sup_{u_{1}\leq \bar{u}\leq u_{2}}\int\limits_ {\left\lbrace  v_*\leq v \leq \hat{v}\right\rbrace }J_{\mu}^{S}(\phi) n^{\mu}_{u=\bar{u}}\dV_{u=\bar{u}}
\nonumber\\
&&+  \int\limits^{\hat{v}}_{v_*}\sup_{u_{1}\leq \bar{u}\leq u_{2}}\left[\frac{ \Omega^2(\bar{u},\bar{v})}{2r}
\left( 1 + {\bar{v}}^{p}\right) \right]\md \bar{v}\sup_{v_*\leq \bar{v} \leq \hat{v}}\int\limits_ {\left\lbrace  u_1\leq u \leq u_2\right\rbrace }J_{\mu}^{S}(\phi) n^{\mu}_{v=\bar{v}}\dV_{v=\bar{v}}.\nonumber\\
&&
\eea
It remains to show finiteness and smallness of \mbox{$\int\limits^{u_{2}}_{u_{1}} \sup_{v_*\leq \bar{v} \leq \hat{v}}
\left[\frac{ \Omega^2(\bar{u},\bar{v})}{2r}\left(1+ {{\bar{v}}^{-p}}\right)\right]\md \bar{u} $} and \mbox{$\int\limits^{\hat{v}}_{v_*}\sup_{u_{1}\leq \bar{u}\leq u_{2}}\left[\frac{ \Omega^2(\bar{u},\bar{v})}{2r}
\left( 1 + {\bar{v}}^{p}\right) \right]\md \bar{v}$}.
Recall the properties of the hypersurface $\gamma$ shown in Section \ref{gamma_curve}. Since $v_*> v_{\gamma}(u_{\schere})$, \eqref{spacevolumedecay_u} implies that
\bea
\lb{zukunftu}
\Omega^2(\bar{u}, \bar{v})\leq C \Omega^2(u_{\schere}, v_*), \quad \mbox{for any $(\bar{u}, \bar{v}) \in J^+(x)$, with \mbox{$x=(u_{\schere}, v_*)$}, $x \in \cB$,}
\eea
so that we obtain
\bea
\int\limits^{u_{2}}_{u_1} \sup_{v_*\leq \bar{v} \leq \hat{v}}
\left[\frac{ \Omega^2(\bar{u},\bar{v})}{2r}\left(1+ {{\bar{v}}^{-p}}\right)\right]\md \bar{u} 
&\stackrel{\eqref{zukunftu}}{\leq}& 
C \int\limits^{u_{2}}_{u_1} \sup_{v_*\leq \bar{v} \leq \hat{v}}
\Omega^2(u_{\schere}, v_*)\left(1+ {{\bar{v}}^{-p}}\right)\md \bar{u},\nonumber\\
&\leq&  \tilde{C} \int\limits^{u_{2}}_{u_1}
|u_{\schere}|^{-\beta\alpha}\left(1+ {{v_*}^{-p}}\right)\md \bar{u},\nonumber\\
&\leq& \tilde{\tilde{C}}\left|{u_{2}}-{u_1}\right| \nonumber\\
&\leq&\delta_1,
\eea
and moreover $\delta_1 \rightarrow 0$ for $\epsilon \rightarrow 0$.

Further, in Section \ref{finiteness} we derived that similarly
\bea
\lb{zukunft}
\Omega^2(\bar{u}, \bar{v})\leq C \bar{v}^{-\beta\alpha}, \quad \mbox{for any $(\bar{u}, \bar{v}) \in J^+(x)$, with \mbox{$x=(u_{\schere}, v_*)$}, $x \in \cB$,}
\eea
where $v_*> v_{\gamma}(u_{\schere})$.
\bea
\int\limits^{\hat{v}}_{v_*}\sup_{u_{1}\leq \bar{u}\leq u_{2}}\left[\frac{ \Omega^2(\bar{u},\bar{v})}{2r}
\left( 1 + {\bar{v}}^{p}\right) \right]\md \bar{v}&\stackrel{\eqref{zukunft}}{\leq}& 
C\int\limits^{\hat{v}}_{v_*}\bar{v}^{-\beta\alpha}
\left( 1 + {\bar{v}}^{p}\right) \md \bar{v}\nonumber\\
&\leq& \frac{\tilde{C}}{|-\beta\alpha+p+1|}\left[\bar{v}^{-\beta\alpha+p+1}\right]^{\hat{v}}_{v_*}\nonumber\\
&\leq&\delta_2,
\eea
where $\delta_2 \rightarrow 0$ for $v_*\rightarrow \infty$. Thus the conclusion of Lemma \ref{K_spherV} is obtained.
\end{proof}
From the above we obtain
\begin{prop}
\lb{sequenceregion}
Let $\phi$ be as in Theorem \ref{anfang} and $p$ as in \eqref{waspist}. For all \mbox{$v_*>v_{\gamma}(u_{\schere})$} sufficiently large, 
\mbox{$\hat{v} \in (v_*, \infty)$}, 
for \mbox{$u_{\diamond}\geq u_2>u_1\geq u_{\schere}$} and $\epsilon\geq u_2-u_1>0$. Then for $\epsilon$ sufficiently small, the
following is true. If 
\bea
\lb{istart2}
\int\limits_ {\left\lbrace  v_* \leq v \leq \hat{v}\right\rbrace } J_{\mu}^{W}(\phi) n^{\mu}_{u=u_{1}} \dV_{u=u_{1}}&\leq& {{\tilde{C}}_1},
\eea
then
\bea
\lb{idiamond}
\int\limits_ {\left\lbrace  v_* \leq v \leq \hat{v}\right\rbrace } J_{\mu}^{W}(\phi) n^{\mu}_{u_{2}} \dV_{u_{\diamond}}
+\int\limits_ {\left\lbrace  u_{1} \leq u \leq u_{2}\right\rbrace } J_{\mu}^{W}(\phi) n^{\mu}_{v=\hat{v}} \dV_{v=\hat{v}}
&\leq& {{\tilde{C}}_{2}}(\tilde{C}_1, u_\diamond, v_*), 
\eea
where ${{\tilde{C}}_{2}}$ depends on $\tilde{C}_1$, ${C_{0}}$ of Theorem \ref{anfang} and ${D_{0}(u_{\diamond}, v_*)}$ of Proposition \ref{initialdataprop}. 
\end{prop}
{\em Remark.} Note already that the hypothesis \eqref{istart2} is implied by the conclusion of Proposition \ref{gesamtesti} for $u_1=u_{\schere}$. 

\begin{proof}
By the divergence theorem and \eqref{167} we can state
\bea
\lb{divv}
&&\int\limits_ {\left\lbrace  v_* \leq v \leq \hat{v}\right\rbrace }  J_{\mu}^{W}(\phi) n^{\mu}_{u=u_{2}} \dV_{u=u_{2}}
+ \int\limits_ {\left\lbrace  u_{1}\leq u \leq u_{2}\right\rbrace }  J_{\mu}^{W}(\phi) n^{\mu}_{v=\hat{v}} \dV_{v=\hat{v}}
\nonumber\\
&\leq&\int\limits_ {\cR_{V_1}} |\tilde{K}^{W}| \dV+\int\limits_ {\left\lbrace  u_{1}\leq u \leq u_{2}\right\rbrace }  J_{\mu}^{W}(\phi) n^{\mu}_{v=v_*} \dV_{v=v_*}
+\int\limits_ {\left\lbrace  v_* \leq v \leq \hat{v}\right\rbrace }  J_{\mu}^{W}(\phi) n^{\mu}_{u=u_{1}} \dV_{u=u_{1}}.
\eea
We can replace the hypersurfaces $u=u_2$ and $v=\hat{v}$ with  $u=\bar{u}$ and $v=\bar{v}$ hypersurfaces and therefore obtain
\bea
&&\sup_{u_{1}\leq \bar{u}\leq u_{2}}\int\limits_ {\left\lbrace  v_* \leq v \leq \hat{v}\right\rbrace }  J_{\mu}^{W}(\phi) n^{\mu}_{u=\bar{u}}\dV_{u=\bar{u}}
+\sup_{v_*\leq \bar{v} \leq \hat{v}}\int\limits_ {\left\lbrace  u_{1}\leq u \leq u_{2}\right\rbrace }J_{\mu}^{W}(\phi) n^{\mu}_{v=\bar{v}} \dV_{v=\bar{v}}\nonumber\\
&\leq&\int\limits_ {\cR_{V_i}} |\tilde{K}^{W}| \dV+\int\limits_ {\left\lbrace  u_{\text{\ding{34}}}\leq u \leq u_{\diamond}\right\rbrace }  J_{\mu}^{W}(\phi) n^{\mu}_{v=v_*} \dV_{v=v_*}
+\int\limits_ {\left\lbrace  v_* \leq v \leq \hat{v}\right\rbrace }  J_{\mu}^{W}(\phi) n^{\mu}_{u=u_{1}} \dV_{u=u_{1}},\nonumber\\
&\stackrel{\mbox{Lem.} \ref{K_spherV}}{\leq}&\delta_1\sup_{u_{1}\leq \bar{u}\leq u_{2}}\int\limits_ {\left\lbrace  v_* \leq v \leq \hat{v}\right\rbrace }  J_{\mu}^{W}(\phi) n^{\mu}_{u=\bar{u}}\dV_{u=\bar{u}}
+\delta_2\sup_{v_*\leq \bar{v} \leq \hat{v}}\int\limits_ {\left\lbrace  u_{1}\leq u \leq u_{2}\right\rbrace }J_{\mu}^{W}(\phi) n^{\mu}_{v=\bar{v}} \dV_{v=\bar{v}}\nonumber\\
&&+\int\limits_  {\left\lbrace  u_{\text{\ding{34}}}\leq u \leq u_{\diamond}\right\rbrace }J_{\mu}^{W}(\phi) n^{\mu}_{v=v_*} \dV_{v=v_*}
+\int\limits_ {\left\lbrace  v_* \leq v \leq \hat{v}\right\rbrace }  J_{\mu}^{W}(\phi) n^{\mu}_{u=u_{1}} \dV_{u=u_{1}}.
\eea
Thus, we have
\bea
\lb{deltamax}
&&\Rightarrow \sup_{u_{1}\leq \bar{u}\leq u_{2}}\int\limits_ {\left\lbrace  v_* \leq v \leq \hat{v}\right\rbrace }  J_{\mu}^{W}(\phi) n^{\mu}_{u=\bar{u}}\dV_{u=\bar{u}}
+\sup_{v_*\leq \bar{v} \leq \hat{v}}\int\limits_ {\left\lbrace  u_{1}\leq u \leq u_{2}\right\rbrace }J_{\mu}^{W}(\phi) n^{\mu}_{v=\bar{v}} \dV_{v=\bar{v}}\nonumber\\
&\leq& \frac{1}{1-\mbox{max}\left\{\delta_1,\delta_2\right\}}\left[\int\limits_ {\left\lbrace  u_{\text{\ding{34}}}\leq u \leq u_{\diamond}\right\rbrace } J_{\mu}^{W}(\phi) n^{\mu}_{v=v_*} \dV_{v=v_*}
+\int\limits_ {\left\lbrace  v_* \leq v \leq \hat{v}\right\rbrace }  J_{\mu}^{W}(\phi) n^{\mu}_{u=u_{1}} \dV_{u=u_{1}}\right].\nonumber\\
&\leq&\tilde{C}{D_{0}(u_{\diamond}, v_*)}+{{\tilde{C}}_1}, 
\eea
where the last step 
follows by statement \eqref{proppur} of Proposition \ref{initialdataprop} and by \eqref{istart2}, and where we have chosen $\epsilon$ sufficiently small and $v_*$ sufficiently close to $\infty$, such that
$\delta_1$ and $\delta_2$ satisfy say
\bea
\lb{deltacondition}
\delta_1, \delta_2\leq\frac12.
\eea
The conclusion of Proposition \ref{sequenceregion} is obtained. 
\end{proof}

We are now ready to make a statement for the entire region $\cR_V$.
\begin{prop}
\lb{rechtes}
Let $\phi$ be as in Theorem \ref{anfang} and $p$ as in \eqref{waspist}.
Then,
for all \mbox{$v_*> v_{\gamma}(u_{\schere})$} sufficiently large, 
$\hat{v}>v_*$, 
and \mbox{$u_{\diamond}>\hat{u}>u_{\schere}$},
\bea
\lb{diamond}
\int\limits_ {\left\lbrace  u_{\schere} \leq u \leq u_{\diamond}\right\rbrace } J_{\mu}^{W}(\phi) n^{\mu}_{v=\hat{v}} \dV_{v=\hat{v}}
+\int\limits_ {\left\lbrace  v_* \leq v \leq \hat{v}\right\rbrace } J_{\mu}^{W}(\phi) n^{\mu}_{u=\hat{u}} \dV_{u=\hat{u}}&\leq& {C(u_{\diamond}, v_*) }, 
\eea
where $C$ depends on ${C_{0}}$ of Theorem \ref{anfang} and ${D_{0}(u_{\diamond}, v_*)}$ of Proposition \ref{initialdataprop}. 
\end{prop}
\begin{proof}
Let $\epsilon$ be as in Proposition \ref{sequenceregion}.
We choose a sequence \mbox{$u_{i+1}-u_i\leq\epsilon$} and $i=\left\{1,2,..,n\right\}$ such that $u_1=u_{\schere}$ and $u_n=\hat{u}$.
Denote \mbox{$\cR_{V_i}=\left\{u_{i}\leq u \leq u_{i+1}\right\}\cap\left\{ v_*\leq v\leq \hat{v}\right\}$}, cf.~ Figure \ref{region5_sequence}.
Iterating the conclusion of Proposition \ref{sequenceregion} from $u_1$ up to $u_n$ then completes the proof.
Note that $n$ depends only on the smallness condition on $\epsilon$ from Proposition \ref{sequenceregion},
since $n\lesssim  \frac{u_{\diamond}-u_{\schere}}{\epsilon}$.
{\begin{figure}[ht]
\centering
\includegraphics[width=0.4\textwidth]{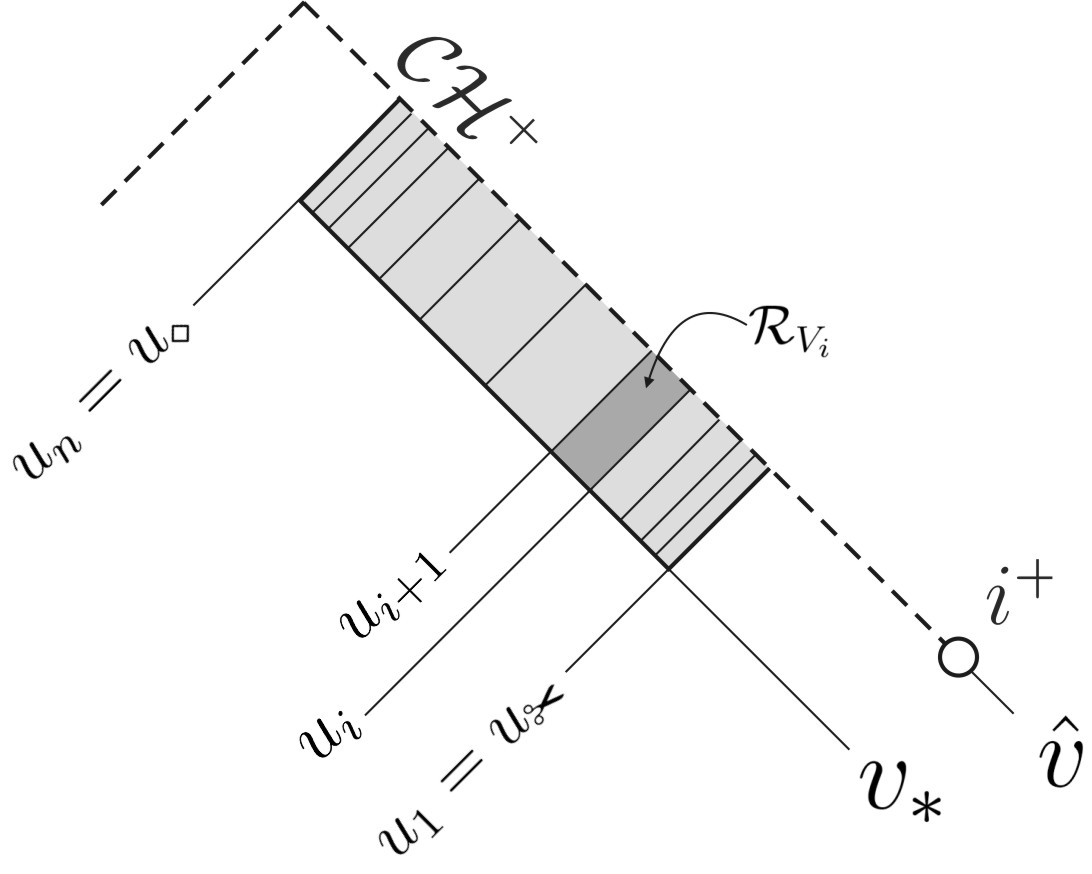}
\caption[]{Penrose diagram depicting regions $\cR_{V_i}$.}
\label{region5_sequence}\end{figure}}
\end{proof}

\section{Energy along the future boundaries of $\tilde{\cR}_{V}$}
\lb{region_tilde5_proof}
Again we also need the estimates on the {\it left} side, therefore, we repeat the derivation of Section \ref{region5_proof} for region
\mbox{$\tilde{\cR}_{V}=\left\{u_*\leq u\leq \hat{u}\right\}\cap\left\{ v_{\schere}\leq v \leq v_{\diamond} \right\}$}, according to Figure \ref{region5_links}, which is located in the blueshift region \mbox{$\tilde{\cR}_{V}\subset \cB$}.
{\begin{figure}[ht]
\centering
\includegraphics[width=0.4\textwidth]{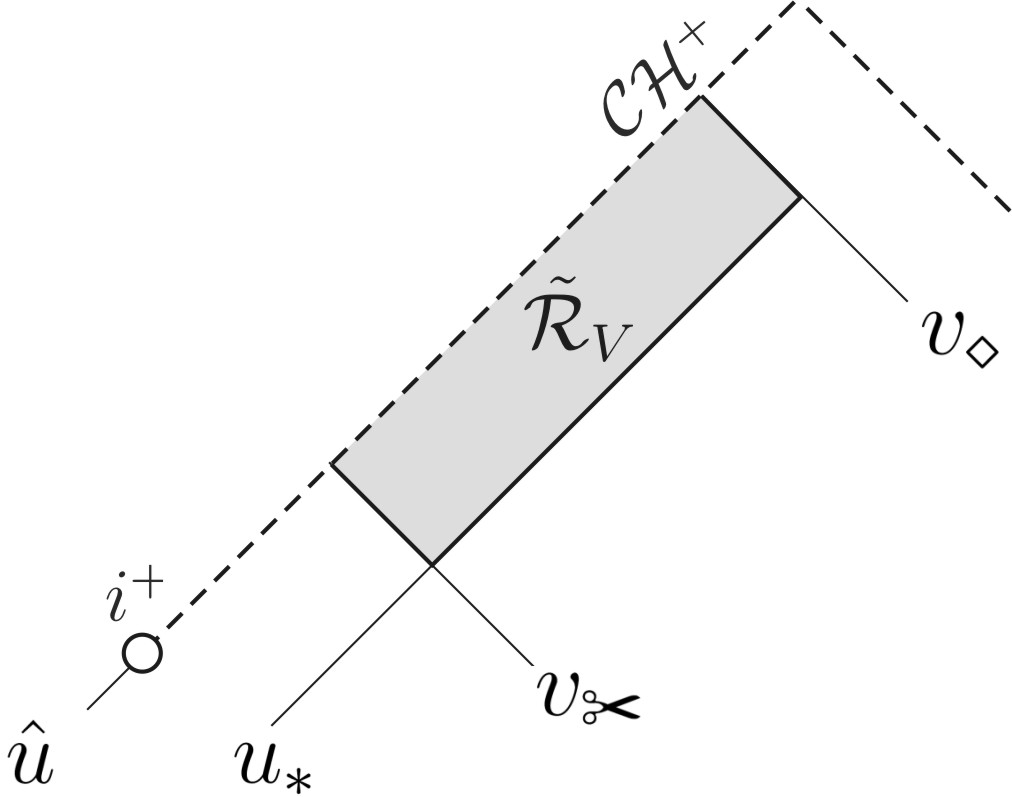}
\caption[]{Penrose diagram depicting region $\tilde{\cR}_{V}$.}
\label{region5_links}\end{figure}}
Note that for region $\tilde{\cR}_{V}$ not only do we have to interchange $u$ and $v$, we also have to use the vector field
\bea
\lb{U_feld}
Z=\partial_v+u^p\partial_u
\eea
instead of $W$, cf.~ \eqref{V_feld}. Therefore, we can immediately state the following proposition about the energy along the horizon away from $v_{\schere}$.
\begin{prop}
\lb{linkes}
Let $\phi$ be as in Theorem \ref{anfangl} and $p$ as in \eqref{waspist}.
Then,
for $v_{\schere}$ sufficiently close to $-\infty$, 
for all \mbox{$u_*> u_{\gamma}(v_{\schere})$} sufficiently large, 
$\hat{u} >u_*$, 
and for $v_{\diamond}>v_{\schere}$ 
\bea
\lb{diamondl}
\int\limits_ {\left\lbrace  v_{\schere} \leq v \leq v_{\diamond}\right\rbrace }  J_{\mu}^{Z}(\phi) n^{\mu}_{u=\hat{u}} \dV_{u=\hat{u}}
+\int\limits_ {\left\lbrace  u_* \leq u \leq \hat{u}\right\rbrace } J_{\mu}^{Z}(\phi) n^{\mu}_{v=v_{\diamond}} \dV_{v=v_{\diamond}}&\leq& {C}(u_*, v_{\diamond}), 
\eea
where $C$ depends on ${C_{0}}$ of Theorem \ref{anfangl} and ${D_{0}(u_*, v_{\diamond})}$ of Proposition \ref{initialdatapropl}.
\end{prop}
\begin{proof}
The proof is analogous to the proof of Proposition \ref{rechtes}.
\end{proof}

\section{Propagating the energy estimate up to the bifurcation sphere}
\lb{bifurcate}
In this section we will use both results from the right and left side on $\cC\cH^+$.
Fix $u_{\diamond}=v_{\diamond}$, such that moreover Proposition \ref{rechtes} holds with $v_{\diamond}=v_*$,
and such that Proposition \ref{linkes} holds with $u_{\diamond}=u_*$.
We will consider a region $\cR_{VI}=\left\{u_{\diamond}\leq u \leq \hat{u}, v_{\diamond}\leq v \leq \hat{v}\right\}$, with $\hat{u} \in (u_{\diamond}, \infty)$ and $\hat{v} \in (v_{\diamond}, \infty)$, cf.~ Figure \ref{diamant}.
{\begin{figure}[ht]
\centering
\includegraphics[width=0.6\textwidth]{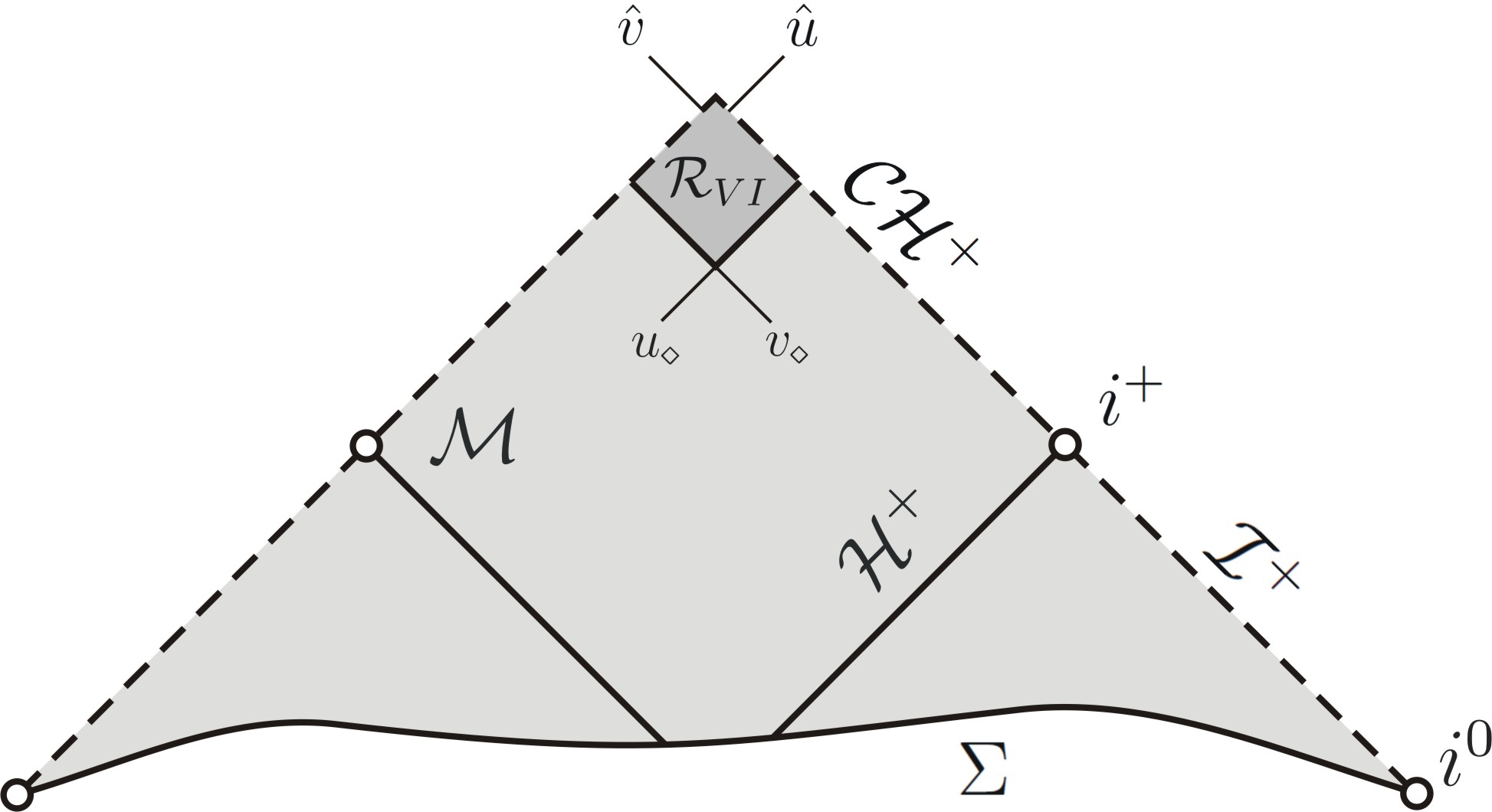}
\caption[]{Penrose diagram depicting region $\cR_{VI}$.}
\label{diamant}\end{figure}}
Recall, that in Section \ref{knsec} we have defined the weighted vector field\footnote{Since $u$ is always positive in the remaining region under consideration $\cR_{VI}$, we have omitted the absolute value in the $u$-weight.}
\bea
\lb{{S2}}
{S}=v^p\partial_v+u^p\partial_u,
\eea
which we are going to use again to obtain an energy estimate up to the bifurcate two-sphere.
Recall $K^{S}$ given in \eqref{KN},
where the terms multiplying the angular derivatives are positive since $\cR_{VI}$ is located in the blueshift region.
We further defined $\tilde{K}^{S}$ in \eqref{KNtilde}
and stated \eqref{knrelation} which will be useful to state the following proposition.
\begin{lem}
\label{K_spher_ende}
Let $\phi$ be an arbitrary function. Then, for all $(u_{\diamond},v_{\diamond})\in J^+(\gamma)\cap\cB$ and all $\hat{u}>u_{\diamond}$, all $\hat{v}>v_{\diamond}$,
the integral over \mbox{$\cR_{VI}$}, cf.~ Figure \ref{diamant} of the current $\tilde{K}^{S}$, defined by \eqref{KNtilde}, can be estimated by
\bea
\lb{deltaende}
\int\limits_{\cR_{VI}} |\tilde{K}^{S}| \dV \leq &\delta_1& \sup_{u_{\diamond}\leq \bar{u}\leq \hat{u}}\int\limits_ {\left\lbrace  v_{\diamond} \leq v \leq \hat{v}\right\rbrace }  J_{\mu}^{S}(\phi) n^{\mu}_{u=\bar{u}}\dV_{u=\bar{u}}\nonumber \\
+ &\delta_2& \sup_{v_{\diamond}\leq \bar{v} \leq \hat{v}}\int\limits_ {\left\lbrace  u_{\diamond} \leq u \leq \hat{u}\right\rbrace }J_{\mu}^{S}(\phi) n^{\mu}_{v=\bar{v}} \dV_{v=\bar{v}},
\eea
where $\delta_1$ and $\delta_2$ are positive constants, with $\delta_1\rightarrow 0$ as $u_{\diamond}\rightarrow \infty$ and $\delta_2\rightarrow 0$ as $v_{\diamond}\rightarrow \infty$.
\end{lem}
\begin{proof}
The proof is similar to the proof of Lemma \ref{K_spher} of Section \ref{knsec} and Lemma \ref{K_spherV} of Section \ref{region5_proof}.
We still need to show finiteness and smallness of \mbox{$\int\limits_{u_{\diamond}}^{\hat{u}} \sup_{v_{\diamond}\leq \bar{v} \leq \hat{v}}
\left[\frac{ \Omega^2(\bar{u},\bar{v})}{2r}\left(1+ \frac{{\bar{u}}^p}{{\bar{v}}^p}\right)\right]\md \bar{u} $} and \mbox{$\int\limits^{\hat{v}}_{v_{\diamond}}\sup_{u_{\diamond} \leq \bar{u}\leq \hat{u}}\left[\frac{ \Omega^2(\bar{u},\bar{v})}{2r}
\left( 1 + \frac{{\bar{v}}^p}{{\bar{u}}^p}\right) \right]\md \bar{v}$}.
In Section \ref{finiteness} we derived \eqref{omegafix} which we will use now 
for all $\bar{u},\bar{v} \in J^+(u_{\diamond}, v_{\diamond})$
Therefore, we can write
\bea
\int\limits_{u_{\diamond}}^{\hat{u}} \sup_{v_{\diamond}\leq \bar{v} \leq \hat{v}}\left[\frac{ \Omega^2(\bar{u},\bar{v})}{2r}
\left( 1 + \frac{{\bar{u}}^p}{{\bar{v}}^p}\right) \right]\md \bar{u}
&\leq& C\int\limits_{u_{\diamond}}^{\hat{u}} \sup_{v_{\diamond}\leq \bar{v} \leq \hat{v}}\left[\Omega^2(u_{\diamond}, v_{\diamond})e^{-\beta\left[\bar{v}-v_{\diamond}+\bar{u}-u_{\diamond}\right]}\left( 1 + \frac{{\bar{u}}^p}{{\bar{v}}^p}\right)\right] \md \bar{u},\nonumber\\
&\leq& \tilde{C}\int\limits_{u_{\diamond}}^{\hat{u}} \Omega^2(u_{\diamond}, v_{\diamond})e^{-\beta\left[\bar{u}-u_{\diamond}\right]}
\left( 1 + \frac{{\bar{u}}^p}{v_{\diamond}^p}\right) \md \bar{u},\nonumber\\
&\leq& \delta_1, 
\eea
where $\delta_1\rightarrow 0$ as $u_{\diamond}= v_{\diamond}\rightarrow \infty$ (since $\Omega^2(u_{\diamond}, v_{\diamond})\rightarrow 0$, cf.~ \eqref{omeganull}). 
Similarly, for finiteness of the second term we obtain
\bea
 \int\limits^{\hat{v}}_{v_{\diamond}}\sup_{u_{\diamond} \leq \bar{u}\leq \hat{u}}\left[\frac{ \Omega^2(\bar{u},\bar{v})}{2r}\left( 1 + \frac{{\bar{v}}^p}{{\bar{u}}^p}\right) \right]\md \bar{v}
&\leq& C\int\limits_{v_{\diamond}}^{\hat{v}} \sup_{u_{\diamond}\leq \bar{u} \leq \hat{u}}\left[\Omega^2(u_{\diamond}, v_{\diamond})e^{-\beta\left[\bar{v}-v_{\diamond}+\bar{u}-u_{\diamond}\right]}\left( 1 + \frac{{\bar{v}}^p}{{\bar{u}}^p}\right)\right] \md \bar{v},\nonumber\\
&\leq&  \tilde{\tilde{C}}\int\limits_{v_{\diamond}}^{\hat{v}} \Omega^2(u_{\diamond}, v_{\diamond})e^{-\beta\left[\bar{v}-v_{\diamond}\right]}
\left( 1 + \frac{{\bar{v}}^p}{u_{\diamond}^p}\right) \md \bar{v},\nonumber\\
&\leq& \delta_2, 
\eea
where $\delta_2\rightarrow 0$ as $u_{\diamond}= v_{\diamond}\rightarrow \infty$. 
Thus we obtain the statement of Lemma \ref{K_spher} by fixing $u_{\diamond}=v_{\diamond}$ sufficiently large.
\end{proof}

\begin{prop}
\lb{bifu}
Let $\phi$ be as in Theorem \ref{anfang} and Theorem \ref{anfangl}.
Then, for $u_{\diamond}=v_{\diamond}$ sufficiently close to $\infty$ and $\hat{u} >u_{\diamond}$, $\hat{v} >v_{\diamond}$
\bea
\int\limits_ {\left\lbrace  v_{\diamond} \leq v \leq \hat{v}\right\rbrace } J_{\mu}^{S}(\phi) n^{\mu}_{u=\tilde{u}} \dV_{u=\tilde{u}}+
\int\limits_ {\left\lbrace  u_{\diamond} \leq u \leq \hat{u}\right\rbrace } J_{\mu}^{S}(\phi) n^{\mu}_{v=\tilde{v}} \dV_{v=\tilde{v}}&\leq&  C(u_{\diamond}, v_{\diamond}),
\eea
where $C$ depends on ${C_{0}}$ of Theorems \ref{anfang}, \ref{anfangl} and ${D_{0}(u_{\diamond}, v_{\diamond})}$ of Propositions \ref{initialdataprop}, \ref{initialdatapropl}. 
\end{prop}
\begin{proof}
The proof follows from applying the divergence theorem for the current $J_{\mu}^{S}(\phi)$ in the region $\cR_{VI}$. The past boundary terms are estimated by Proposition \ref{rechtes} and Proposition \ref{linkes}. Note that the weights of $J_{\mu}^{S}(\phi)$ are comparable to the weights of $J_{\mu}^{W}(\phi)$ for fixed $u_{\diamond}$, and similarly the weights of $J_{\mu}^{S}(\phi)$ are comparable to the weights of $J_{\mu}^{Z}(\phi)$ for fixed $v_{\diamond}$. The bulk term is absorbed by Lemma \ref{K_spher_ende}.
\end{proof}

Now that we have shown boundedness for different subregions of the interior we can state the following proposition for the entire interior region \mbox{$\mathring{\cM}|_{II}$}, cf.~ Section \ref{rtcoords}
\begin{cor}
\lb{letzteprop}
Let $\phi$ be as in Theorem \ref{anfang} and Theorem \ref{anfangl}.
Then 
\bea
\int\limits_{\bbS^2}\int\limits^{\infty}_{v_{fix}}\left[ (|v|+1)^p (\partial_v \phi)^2(u_{fix}, v, \theta, \varphi) +|\nabb \phi|^2(u_{fix}, v, \theta, \varphi) \right]r^2\md v\md \sigma_{\mathbb S^2}&\leq C, \quad \mbox{for  $v_{fix} \geq v_{\schere}$, $u_{fix} > -\infty$ }, \\
\int\limits_{\bbS^2}\int\limits^{\infty}_{u_{fix}} \left[(|u|+1)^p (\partial_u \phi)^2 (u, v_{fix}, \theta, \varphi)+|\nabb \phi|^2(u, v_{fix}, \theta, \varphi) \right]r^2\md u\md \sigma_{\mathbb S^2} &\leq C,\quad \mbox{for  $u_{fix} \geq u_{\schere}$, $v_{fix} > -\infty$}, 
\eea
where $p$ is as in \eqref{waspist} and 
$C$ depends on ${C_{0}}$ of Theorems \ref{anfang}, \ref{anfangl} and ${D_{0}(u_{\diamond}, v_{\diamond})}$ of Propositions \ref{initialdataprop}, \ref{initialdatapropl}, where $u_{\diamond}=v_{\diamond}$ is as in \ref{bifu}. 
\end{cor}
\begin{proof}
This follows by examining the weights in Propositions \ref{rechtes}, \ref{linkes} and \ref{bifu} together with Theorem \ref{energythm2} and its analog for the {\it left} side. 
\end{proof}

\section{Global higher order energy estimates and pointwise boundedness}
\lb{global}
To obtain pointwise bounds in analogy to Section \ref{nineten} we first have to extend Corollary \ref{letzteprop} to a higher order statement.
\begin{thm}
\lb{energythm3}
On subextremal Reissner-Nordstr\"om spacetime $(\cM,g)$, with mass $M$ and charge $e$ and $M>|e|\neq 0$, let $\phi$ be 
a solution of the wave equation $\Box_g \phi=0$ arising from
sufficiently regular Cauchy data on a two-ended asymptotically flat Cauchy surface $\Sigma$. Then, for $v_{fix}\geq v_{\schere}$, $u_{fix}>-\infty$
\bea
\lb{array1}
&&\int\limits_{\bbS^2}\int\limits^{\infty}_{v_{fix}}\left[ (|v|+1)^p (\partial_v \phi)^2(u_{fix}, v, \theta, \varphi) +|\nabb \phi|^2(u_{fix}, v, \theta, \varphi) \right]r^2\md v\md \sigma_{\mathbb S^2}\leq {E_{0}},\\
&&\int\limits_{\bbS^2}\int\limits^{\infty}_{v_{fix}} \left[(|v|+1)^p (\partial_v \leo\phi)^2(u_{fix}, v, \theta, \varphi) +|\nabb \leo\phi|^2(u_{fix}, v, \theta, \varphi) \right]r^2\md v\md \sigma_{\mathbb S^2}\leq {E_{1}},\\
\lb{array1c}
&&\int\limits_{\bbS^2}\int\limits^{\infty}_{v_{fix}}\left[ (|v|+1)^p (\partial_v \leo^2\phi)^2 (u_{fix}, v, \theta, \varphi)+|\nabb \leo^2\phi|^2(u_{fix}, v, \theta, \varphi) \right]r^2\md v\md \sigma_{\mathbb S^2}\leq {E_{2}};
	\eea
and for $u_{fix}\geq u_{\schere}$, $v_{fix}>-\infty$
\bea
\lb{array2}
&&\int\limits_{\bbS^2}\int\limits^{\infty}_{u_{fix}} \left[(|u|+1)^p (\partial_u \phi)^2 (u, v_{fix}, \theta, \varphi)+|\nabb \phi|^2(u, v_{fix}, \theta, \varphi) \right]r^2\md u\md \sigma_{\mathbb S^2} \leq {E_{0}},\\
&&\int\limits_{\bbS^2}\int\limits^{\infty}_{u_{fix}}\left[ (|u|+1)^p (\partial_u \leo\phi)^2 (u, v_{fix}, \theta, \varphi)+|\nabb \leo \phi|^2(u, v_{fix}, \theta, \varphi) \right]r^2\md u\md \sigma_{\mathbb S^2}\leq{E_{1}},\\
\lb{array2c}
&&\int\limits_{\bbS^2}\int\limits^{\infty}_{u_{fix}}\left[ (|u|+1)^p (\partial_u \leo^2\phi)^2(u, v_{fix}, \theta, \varphi) +|\nabb \leo^2\phi|^2(u, v_{fix}, \theta, \varphi) \right]r^2\md u\md \sigma_{\mathbb S^2}\leq {E_{2}}, 
\eea
where $p$ is as in \eqref{waspist}.
\end{thm}
\begin{proof}
This follows immediately from Corollary \ref{letzteprop} by commutation.
\end{proof}
Having proven Theorem \ref{energythm3}, the pointwise boundedness of $|\phi|$ in all of $\mathring{\cM}|_{II}$ follows analogously to Section \ref{uni_bounded}.
We estimate
\bea
\lb{uglobalfundcauchy}
\int\limits_{\bbS^2} \phi^2(\hat{u}, v)\md \sigma_{\mathbb S^2}
&\leq& \tilde{C}\left[\int\limits_{\bbS^2}\left(\int\limits_{u_*}^{\hat{u}} (|u|+1)^p(\partial_u \phi)^2(u,v)\md u\int\limits_{u_*}^{\hat{u}} (|u|+1)^{-{p}}\md v\right)r^2\md \sigma_{\mathbb S^2}+\int\limits_{\bbS^2} \phi^2(u_*,v)\md \sigma_{\mathbb S^2}\right],\nonumber \\
\int\limits_{\bbS^2} \phi^2(\hat{u},{v})\md \sigma_{\mathbb S^2}
&\leq&\tilde{C}\left[\tilde{\tilde{C}} E_{0}+\int\limits_{\bbS^2} \phi^2(u_*,v)\md \sigma_{\mathbb S^2}\right],
\eea 
where $u_*\geq u_{\schere}$, $\hat{u} \in (u_*, \infty)$ and $v \in (1, \infty)$. 
Commuting by angular momentum operators $\leo_i$ and summing over them we obtain
\bea
\lb{uglobalfundcauchy1}
\int\limits_{\bbS^2} (\leo \phi) ^2(\hat{u},{v})\md \sigma_{\mathbb S^2} 
&\leq&\tilde{C}\left[\tilde{\tilde{C}} E_{1}+\int\limits_{\bbS^2} (\leo \phi) ^2(u_*,v)\md \sigma_{\mathbb S^2}\right],\\
\lb{uglobalfundcauchy2}
\int\limits_{\bbS^2} ({\leo}^2 \phi)^2(\hat{u},{v})\md \sigma_{\mathbb S^2}
&\leq&\tilde{C}\left[\tilde{\tilde{C}}  E_{2}+\int\limits_{\bbS^2} ({\leo}^2 \phi) ^2(u_*,v)\md \sigma_{\mathbb S^2}\right].
\eea 
By using the result \eqref{supr2} in \eqref{uglobalfundcauchy} we derive pointwise boundedness according to \eqref{sobo_embed} 
\bea
\lb{usupr}
\sup_{\bbS^2}|\phi(\hat{u},{v},\theta,\varphi)|^2&\leq& \tilde{C} \left[\int\limits_{\bbS^2} ( \phi) ^2(\hat{u},{v})\md \sigma_{\mathbb S^2} +\int\limits_{\bbS^2} (\leo \phi) ^2(\hat{u},{v})\md \sigma_{\mathbb S^2} +\int\limits_{\bbS^2} (\leo^2 \phi) ^2(\hat{u},{v})\md \sigma_{\mathbb S^2} \right],\nonumber\\
&\leq&\tilde{C}\left[\tilde{\tilde{C}}\left( E_{0}+  E_{1}+ E_{2}\right)+C )\right],\nonumber\\
&\leq&C,
\eea
with $C$ depending on the initial data. 

Inequalities \eqref{usupr} and \eqref{supr3} gives the desired \eqref{maineq} for all $v\geq1$.
Interchanging the roles of $u$ and $v$,
likewise \eqref{maineq} follows for all $u\geq 1$. 
The remaining subset of the interior has compact closure in spacetime for which \eqref{maineq} thus follows by Cauchy stability. We have thus shown \eqref{maineq} globally in the interior.

As we will see in the next section, the continuity statement of Theorem \ref{main}
follows easily by revisiting the Sobolev estimates.

\section{Continuity statement of Theorem \ref{main}}
\lb{continuity}
In the previous section we have shown pointwise boundedness, $|\phi(u,v,\varphi,\theta)|\leq C$. In the following we prove 
that $\phi$ extends continuously to $\cC\cH^+$, that is to say
$\phi$ extends to \mbox{$\left\{\infty\right\}\times(-\infty, \infty] \cup (-\infty, \infty]\times\left\{\infty\right\}$} so that $\phi$ is continuous as a function on \mbox{$(-\infty, \infty] \times (-\infty, \infty] \times \bbS^2$}.
Showing the extension closes the proof of Theorem \ref{main}. 

In order to first show continuous extendibility of $\phi$ to \mbox{$(-\infty, \infty)\times\left\{\infty\right\}$}, it suffices to show:
Given $-\infty <u<\infty$ and  $\varphi, \theta \in \bbS^2$, \mbox{$\forall \epsilon>0 \quad \exists \delta, v_*$}, such that
\bea
\lb{entire_proof}
|\phi(u,v,\varphi, \theta)- \phi(\tilde{u},\tilde{v},\tilde{\varphi}, \tilde{\theta})|<4\epsilon,\quad \mbox{for all} \left\{\begin{aligned}
              &v>\tilde{v},\quad \tilde{v}\geq v_*& \\
							&u-\tilde{u}<\delta&\\
							&\varphi-\tilde{\varphi}<\delta&\\
							&\theta-\tilde{\theta}<\delta.&
      \end{aligned}    \right.
\eea

By the triangle inequality we obtain
\bea
\lb{triangle}
|\phi(u,v,\varphi, \theta)- \phi(\tilde{u},\tilde{v},\tilde{\varphi}, \tilde{\theta})| &\leq& |\phi(u,v,\varphi, \theta)- \phi(\tilde{u},{v},\varphi, \theta)|+|\phi(\tilde{u},{v},\varphi, \theta)- \phi(\tilde{u},\tilde{v},\varphi, \theta)|\nonumber\\
&+&|\phi(\tilde{u},\tilde{v},{\varphi}, \theta)- \phi(\tilde{u},\tilde{v},\tilde{\varphi}, \theta)|+|\phi(\tilde{u},\tilde{v},\tilde{\varphi}, \theta)- \phi(\tilde{u},\tilde{v},\tilde{\varphi}, \tilde{\theta})|.\nonumber\\
\eea 
We will show that each term can be bounded by $\epsilon$.

Considering first the $u$ direction by the fundamental theorem of calculus we have
\bea
\lb{fund_part}
\phi(u,v,\varphi, \theta)- \phi(\tilde{u},{v},\varphi, \theta)&=& \int\limits_{\tilde{u}}^{u} \partial_{\bar{u}}\phi(u,v,\varphi, \theta)\md \bar{u}.
\eea
Applying Cauchy Schwarz, we obtain for fixed $v, \varphi, \theta$
\bea
|\phi(u,v,\varphi, \theta)- \phi(\tilde{u}, {v},\varphi, \theta)|^2 &\leq& \left(\int\limits_{\tilde{u}}^{u}|\partial_{\bar{u}}\phi(\bar{u}, {v},\varphi, \theta)|\md \bar{u}\right)^2\nonumber\\
 &\leq&
\left(\int\limits_{\tilde{u}}^{u} (|{\bar{u}}|+1)^p\left(\partial_{\bar{u}}\phi(\bar{u}, {v},\varphi, \theta)\right)^2\md \bar{u}\right)\left(\int\limits_{\tilde{u}}^{u} (|{\bar{u}}|+1)^p\md \bar{u}\right)\nonumber\\
&\leq&
\left(\int\limits_{\tilde{u}}^{u}  (|{\bar{u}}|+1)^p\left(\partial_{\bar{u}}\phi(\bar{u}, {v},\varphi, \theta)\right)^2\md \bar{u}\right)
\nonumber\\
&&\times\left(\frac{1}{-p+1} (|{u}|+1)^{-p+1}-\frac{1}{-p+1} (|{\tilde{u}}|+1)^{-p+1}\right)\nonumber\\
\lb{sobo}
&\leq&
 \tilde{C}\left(\int\limits_{\tilde{u}}^{u} \sum_{k=0}^{2} \int\limits_{\bbS^2} (|{\bar{u}}|+1)^p\left(\leo^k \partial_{\bar{u}}\phi\right)^2\md \sigma_{\mathbb S^2}\md \bar{u}\right)\nonumber\\
&&\times\frac{1}{p-1}\left((|{\tilde{u}}|+1)^{-p+1}-{(|{u}|+1)}^{-p+1}\right)\\
\lb{fund_partu}
&\leq& {\epsilon}.
\eea
In the above $p$ is as in \eqref{waspist} and \eqref{sobo} follows from \eqref{sobo_embed} applied to $\partial_{\bar{u}}\phi$,
\bea
\lb{sobo_embedup}
\sup_{\theta,\varphi \in \bbS^2}|(|{\bar{u}}|+1)^p\partial_{\bar{u}}\phi(\bar{u},v,\theta,\varphi)|^2\leq \tilde{C} \sum_{k=0}^{2} \int\limits_{\bbS^2} (|{\bar{u}}|+1)^p\left(\leo^k \partial_{\bar{u}}\phi\right)^2(\bar{u},v,\theta,\varphi)\md \sigma_{\mathbb S^2}. 
\eea
Further, the last step, \eqref{fund_partu}, then follows from \eqref{array2}-\eqref{array2c} for a suitable chosen $\delta$ in \eqref{entire_proof}. 

For the second term in \eqref{triangle}, again by the fundamental theorem of calculus and the Cauchy Schwarz inequality we obtain
\bea
|\phi(\tilde{u},v,\varphi, \theta)- \phi(\tilde{u},\tilde{v},\varphi, \theta)|^2 &\leq& \left(\int\limits_{\tilde{v}}^{v}|\partial_{\bar{v}}\phi(\tilde{u},\bar{v},\varphi, \theta)|\md \bar{v}\right)^2\nonumber\\
 &\leq&
\left(\int\limits_{\tilde{v}}^{v} {\bar{v}}^{p}\left(\partial_{\bar{v}}\phi(u,\bar{v},\varphi, \theta)\right)^2\md \bar{v}\right)\left(\int\limits_{\tilde{v}}^{v} {\bar{v}}^{-p}\md \bar{v}\right)\nonumber\\
&\leq&
\lb{dritter}
 \tilde{C}\left(\int\limits_{\tilde{v}}^{v} \sum_{k=0}^{2} \int\limits_{\bbS^2} \bar{v}^p\left(\leo^k \partial_{\bar{v}}\phi\right)^2\md \sigma_{\mathbb S^2}\md \bar{v}\right)\frac{1}{p-1}\left(\tilde{v}^{-p+1}-{v}^{-p+1}\right)\\
\lb{dasda}
&\leq&
 \tilde{C}\left(\int\limits_{\tilde{v}}^{v} \sum_{k=0}^{2} \int\limits_{\bbS^2} \bar{v}^p\left(\leo^k \partial_{\bar{v}}\phi\right)^2\md \sigma_{\mathbb S^2}\md \bar{v}\right)\frac{v_*^{-p+1}}{p-1}\\
\lb{fund_partv}
&\leq& {\epsilon},
\eea 
where in the third step, \eqref{dritter}, we have used \eqref{sobo_embed} applied to $\partial_{\bar{v}}\phi$. Equation \eqref{dasda} follows since $v>\tilde{v}$, $\tilde{v}\geq v_*$ and the last step follows by using \eqref{array1}-\eqref{array1c} and for $v_*$ large enough.

In the $\varphi$ direction for fixed $\tilde{u},\tilde{v}$ and $\theta$ we can state
\bea
|\phi(\tilde{u},\tilde{v},\varphi, \theta)- \phi(\tilde{u},\tilde{v},\tilde{\varphi}, \theta)|^2 &\leq& \left(\int\limits_{\tilde{\varphi}}^{\varphi}|\partial_{\bar{\varphi}}\phi(\tilde{u},\tilde{v},\bar{\varphi}, \theta)|\md \bar{{\varphi}}\right)^2\nonumber\\
&\leq&
C\left(\int\limits_{\tilde{\varphi}}^{\varphi} \int\limits_{0}^{\pi}\left[ \left|\partial_{\bar{\varphi}}\phi(\tilde{u},\tilde{v},\bar{\varphi}, \theta)\right|+\left|\partial_{\theta}\partial_{\bar{\varphi}}\phi(\tilde{u},\tilde{v},\bar{\varphi}, \theta)\right|\right]\md \sigma_{\mathbb S^2}\right)^2\nonumber\\
&\leq&
\tilde{C}\left(\int\limits_{\bbS^2}\left[ \left|\leo_3\phi(\tilde{u},\tilde{v},\bar{\varphi}, \theta)\right|^2+\left|\sum_i a_i \leo_i \leo_3\phi(\tilde{u},\tilde{v},\bar{\varphi}, \theta)\right|^2\right]\md \sigma_{\mathbb S^2}\right)\left(\int\limits_{\tilde{\varphi}}^{\varphi}\int\limits_{0}^{\pi}\md \sigma_{\mathbb S^2}\right)\nonumber\\
\lb{fund_partphi}
&\leq& {\epsilon},
\eea
where in the second step we have used one dimensional Sobolev embedding.
In the third step we have used \eqref{angmom2} for $\partial_{\bar{\varphi}}$ and \eqref{angmom1} to \eqref{angmom2} for $\partial_{\theta}$. Further, we applied the Cauchy Schwarz inequality twice. The last step then follows by using \eqref{uglobalfundcauchy}-\eqref{uglobalfundcauchy2} for $\tilde{u}\geq u_{\schere}$ and \eqref{fundcauchy}-\eqref{fundcauchy2} for $\tilde{u}\leq u_{\schere}$ and since the second integral term is arbitrarily small by suitable choice of $\delta$.

Similarly, in $\theta$ direction for fixed $\tilde{u},\tilde{v}$ and $\tilde{\varphi}$ we obtain
\bea
\lb{fund_parttheta}
|\phi(\tilde{u},\tilde{v},\tilde{\varphi}, \theta)- \phi(\tilde{u},\tilde{v},\tilde{\varphi}, \tilde{\theta})|^2 &\leq& \left(\int\limits_{\tilde{\theta}}^{\theta}|\partial_{\bar{\theta}}\phi(\tilde{u},\tilde{v},\tilde{\varphi}, \bar{\theta})|\md \bar{{\theta}}\right)^2\nonumber\\
&\leq&
C\left(\int\limits_{\tilde{\theta}}^{\theta} \int\limits_{0}^{2\pi}\left[ \left|\partial_{\bar{\theta}}\phi(\tilde{u},\tilde{v},\tilde{\varphi}, \bar{\theta})\right|+\left|\partial_{\varphi}\partial_{\bar{\theta}}\phi(\tilde{u},\tilde{v},\tilde{\varphi}, \bar{\theta})\right|\right]\md \sigma_{\mathbb S^2}\right)^2\nonumber\\
&\leq&
\tilde{C}\left(\int\limits_{\bbS^2}\left[ \left|\sum_i a_i \leo_i\phi(\tilde{u},\tilde{v},\tilde{\varphi}, \bar{\theta})\right|^2+\left| \leo_3\sum_i a_i \leo_i\phi(\tilde{u},\tilde{v},\tilde{\varphi}, \bar{\theta})\right|^2\right]\md \sigma_{\mathbb S^2}\right)\left(\int\limits_{\tilde{\varphi}}^{\varphi}\int\limits_{0}^{2\pi}\md \sigma_{\mathbb S^2}\right)\nonumber\\&\leq& {\epsilon}.
\eea
The second step follows by one dimensional Sobolev embedding and the third from \eqref{angmom1}-\eqref{angmom2} and using the Cauchy Schwarz inequality twice.
In the last step we used \eqref{uglobalfundcauchy}-\eqref{uglobalfundcauchy2} for $\tilde{u}\geq u_{\schere}$ and \eqref{fundcauchy}-\eqref{fundcauchy2} for $\tilde{u}\leq u_{\schere}$  and
a suitable choice of $\delta$.

Using the above results \eqref{fund_partu}, \eqref{fund_partv}, \eqref{fund_partphi} and \eqref{fund_parttheta} in \eqref{triangle} yields the desired result \eqref{entire_proof}.

To show continuous extendibility of $\phi$ to \mbox{$\left\{\infty\right\}\times(-\infty, \infty)$}, it suffices to show:
Given $-\infty <v<\infty$ and  $\varphi, \theta \in \bbS^2$, \mbox{$\forall \epsilon>0 \quad \exists \delta, u_*$}, such that
\bea
\lb{entire_proofu}
|\phi(u,v,\varphi, \theta)- \phi(\tilde{u},\tilde{v},\tilde{\varphi}, \tilde{\theta})|<4\epsilon,\quad \mbox{for} \left\{\begin{aligned}
              &u>\tilde{u},\quad \tilde{u}\geq u_*&\\
               &v-\tilde{v}<\delta&\\
							&\varphi-\tilde{\varphi}<\delta&\\
							&\theta-\tilde{\theta}<\delta.&
      \end{aligned}    \right.
\eea

This can be proven by substituting $v$ with $u$ and $\tilde{v}$ with $\tilde{u}$ and repeating all above steps. 

Similarly, to show continuous extendibility to \mbox{$\left\{\infty\right\}\times\left\{\infty\right\}$}, it suffices to show:
Given $\varphi, \theta \in \bbS^2$, \mbox{$\forall \epsilon>0 \quad \exists \delta, u_*, v_*$}, such that
\bea
\lb{entire_proofuv}
|\phi(u,v,\varphi, \theta)- \phi(\tilde{u},\tilde{v},\tilde{\varphi}, \tilde{\theta})|<4\epsilon,\quad \mbox{for} \left\{\begin{aligned}
              &u>\tilde{u},\quad \tilde{u}\geq u_*&\\
              &v>\tilde{v},\quad \tilde{v}\geq v_*&\\
							&\varphi-\tilde{\varphi}<\delta&\\
							&\theta-\tilde{\theta}<\delta.&
      \end{aligned}    \right.
\eea
This follows as in \eqref{dasda} and completes the proof of Theorem \ref{main}.

\begin{trivlist}
\item[\hskip \labelsep ]\qed\end{trivlist}

\section{Outlook}
\lb{outlook}
In the following section we give a brief overview of further related open problems.

\subsection{Scalar waves on subextremal Kerr interior backgound}
All constructions of this paper have direct generalizations to the subextremal Kerr case.\footnote{This is in contrast to the exterior region for which the
analysis of the Kerr case is significantly harder than for spherically symmetric black hole spacetimes in view of the difficulties of superradiance and the more complicated trapping. See the references below.} 
The following theorem will appear in a subsequent companion paper, cf.~ \cite{ich}.
\begin{thm}
\lb{main2} 
\cite{ich}
On a fixed subextremal Kerr background $0\neq|a|<M$, with angular momentum per unit mass $a$ and mass $M$, let $\phi$ be a solution of the wave equation arising from sufficiently regular initial data on a Cauchy hypersurface $\Sigma$.
Then
\ben
|\phi|\leq C
\een
globally in the black hole interior up to 
$\cC\cH^+$, to which $\phi$ in fact extends continuously. 
\end{thm}
The above Theorem of course depends on the fact that the analog of Theorem \ref{anfang} has been proven in the full subextremal range $|a|<M$ on Kerr backgrounds by Dafermos, Rodnianski and Shlapentokh-Rothman, cf.~ \cite{m_scalar, m_yakov}; see also \cite{anderson, tataru, tataru2, luk2, m_I_kerr, m_lec, m_bound} for the $|a|\ll M$ case and \cite{bernard, yakov} for mode stability.

\subsection{Mass inflation}
In view of our stability result, what remains of the ``blueshift instability''?

The result of Theorem \ref{main} is still compatible with the expectation that the {\it transverse} derivatives (with respect to regular coordinates on $\cC\cH^+$) of $\phi$ will blow up along $\cC\cH^+$, cf.~ the work of Simpson and Penrose \cite{simpson}.
In fact given a lower bound, say 
\bea
\lb{lowerb}
|\partial_{v} \phi|\geq c{v}^{-4},
\eea
for some constant $c>0$ and all sufficiently large $v$
on $\cH^+$, Dafermos has shown for the {\it spherically symmetric} Einstein-Maxwell-scalar field model that transverse derivatives blow up, cf.~ \cite{m_interior}. 
This blow up result of \cite{m_bh} would also apply to our setting here {\it if the above lower bound \eqref{lowerb} is assumed on the spherical mean of $\phi$ on $\cH^+$}.
See also \cite{mc_i}.
Such a lower bound however has not been proven yet for solutions arising from generic data on $\Sigma$; see also Sbierski \cite{jan}.
One might therefore aim for proving the following conjecture.
\begin{con}
Let $\phi$ be as in Theorem \ref{main} or \ref{main2}.
For \underline{generic} data on $\Sigma$,
transverse derivatives of $\phi$ blow up on $\cC\cH^+$, in fact $\phi$ is not $H^1_{\rm loc}$.
\end{con}

\subsection{Extremal black holes}

For a complete geometric understanding of black hole interiors one must also consider extremal black holes.
Aretakis proved stability and instability properties for the evolution of a massless scalar field on a fixed extremal Reissner-Nordstr\"om exterior background. For data on a spacelike hypersurface $\Sigma$ intersecting the event horizon and extending to infinity he has proven decay of $\phi$ up to and including the horizon, cf.~ \cite{stef1}. 
In subsequent work \cite{stef2} Aretakis showed that first transverse derivatives of $\phi$ generically {\it do not decay} along the event horizon for late times. He further proved that higher derivatives {\it blow up} along the event horizon. 

The analysis of the evolution of the scalar wave in the region beyond the event horizon for the extreme case remains to be shown. Motivated by heuristics and numerics of Murata, 
Reall and Tanahashi \cite{reall3}, which suggest stronger stability results in the interior than in the subextremal case, we conjecture
\begin{con}
\lb{outlookext}
(See \cite{reall3}.)
In extremal Reissner-Nordstr\"om spacetime $(\cM,g)$, with mass $M$ and charge $e$ and $M=|e|\neq 0$, let $\phi$ be 
a solution of the wave equation $\Box_g \phi=0$ arising from
sufficiently regular Cauchy data on $\Sigma$. Then, for $V_-$ regular at $\cC\cH^+$ 
\ben
|\phi|\leq C, \qquad |\partial_{V_-} \phi| \leq C, 
\een
globally in the maximal domain of dependence $\cD^+(\Sigma)$ of the hypersurface $\Sigma$ as shown in Figure \ref{exRN}.
\end{con}
{\begin{figure}[!ht]
\centering
\includegraphics[width=0.4\textwidth]{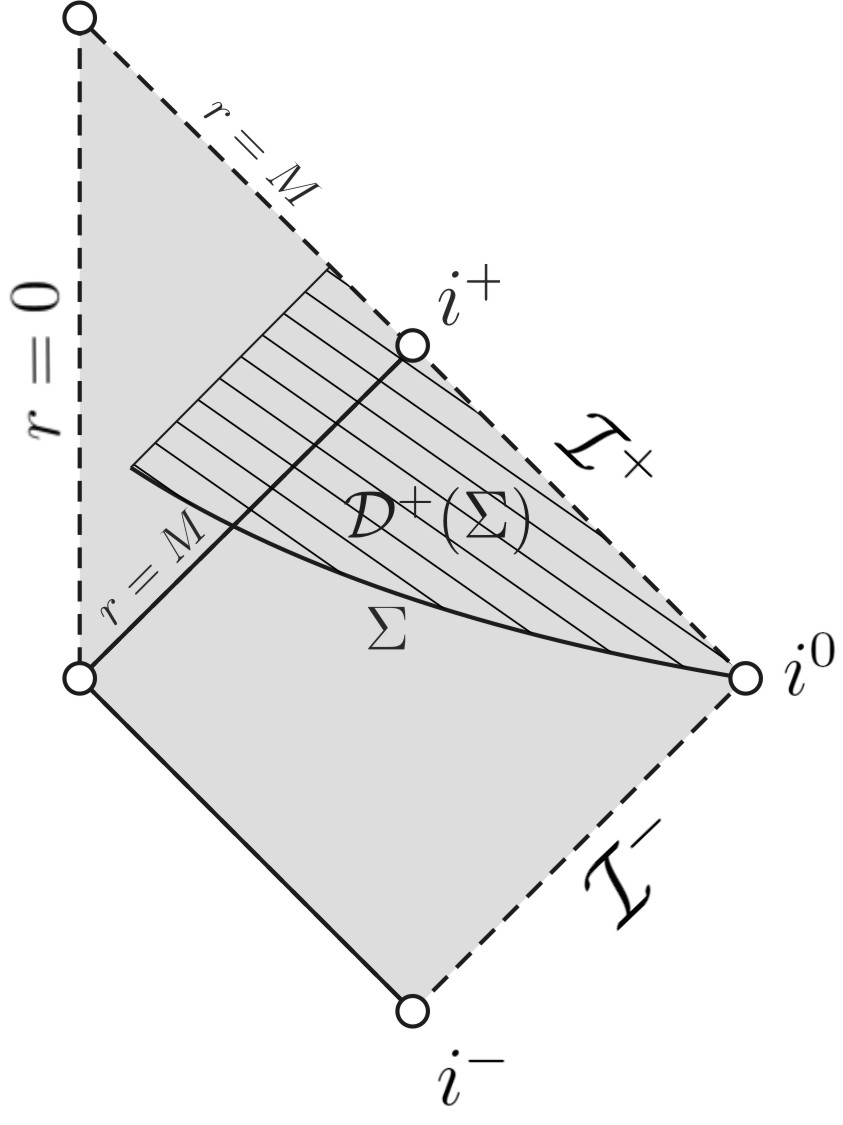}
\caption[]{Extremal Reissner-Nordstr\"om spacetime spacetime with Cauchy hypersurface ending in the non trapped interior region.}
\label{exRN}\end{figure}}
See upcoming work of Gajic \cite{dejan}.

Aretakis has further considered $\Box_g \phi$ on a fixed extremal Kerr background, c.f \cite{stef3}. Analogously to the extremal Reissner-Nordstr\"om case, he shows decay up to and including the event horizon for axisymmetric solutions $\phi$ and in \cite{stef4} shows instability properties when considering transverse derivatives. 
This has been further generalised in \cite{reall1, reall2}. 
In analogy with Conjecture \ref{outlookext}, we thus also make the following conjecture:
\begin{con}
In extremal Kerr spacetime $(\cM,g)$, with mass $M$ and angular momentum per unit mass $a$ and $M=|a|\neq 0$, let $\phi$ be 
an axisymmetric solution of the wave equation $\Box_g \phi=0$ arising from
sufficiently regular Cauchy data on $\Sigma$. Then, for $V_-$ regular at $\cC\cH^+$ 
\ben
|\phi|\leq C, \qquad |\partial_{V_-} \phi| \leq C
\een
globally in the maximal domain of dependence $\cD^+(\Sigma)$ of the hypersurface $\Sigma$.
\end{con}
The case of non axisymmetric $\phi$ seems to be rather complicated and we thus do not venture a conjecture here.

\subsection{Einstein vacuum equations}
\lb{EFvacuum}
We now return to the problem that originally motivated our work, namely the dynamics of the Einstein vacuum equations.
Our result and the forthcoming extension to subextremal Kerr backgrounds further support the expectation that the spherically symmetric toy models \cite{m_stab, m_interior, m_bh, amos, poisson, penrose}
are indeed indicative of what happens for the full nonlinear Einstein vacuum equations {\it without symmetry}. 
In particular our results support the following conjecture given in \cite{m_bh}.
\begin{con}
{\em (A. Ori)} Let $(\cM,g)$ be the maximal vacuum Cauchy development of sufficiently small perturbations of asymptotically flat two-ended Kerr data corresponding to parameters $0<|a|<M$. Then; 

(a) there exist both a future and past extension $(\tilde{\cM},\tilde{g})$ of $\cM$ with $C^0$ metric $\tilde{g}$ such that $\partial\cM$ is a bifurcate null cone in $\tilde{\cM}$ and {\bf all} future (past) incomplete geodesics in $\gamma$ pass into $\tilde{\cM}/\cM$.

(b) Moreover, for generic such perturbations, any $C^0$ extension $\tilde{\cM}$ will fail to have $L^2$ Christoffel symbols in a neighbourhood of any point of $\partial \cM$.
\end{con}
A proof of part (a) has recently been announced by Dafermos and Luk, {\it given the conjectured stability of Kerr exterior} (i.e.~ given the analog of Theorem \ref{anfang} for the full nonlinear Einstein vacuum equations). See the upcoming \cite{m_luk}.
Specific examples of vacuum spacetimes with null singularities as in (b) have been constructed by Luk \cite{luk3}. 
For a discussion of what all this means to Strong Cosmic Censorship see \cite{m_bh}.

\section*{Acknowledgements}
I am very grateful to Mihalis Dafermos for suggesting this problem to me and guiding me towards its completion.
Further, I would like to thank Gerard 't Hooft for his continuous encouragement and comments on the manuscript. 
Parts of this work were conducted at Cambridge University, ETH Z\"urich and Princeton University, whose hospitality is gratefully acknowledged.  
This work is part of the D-ITP consortium, a program of the Netherlands
Organisation for Scientific Research (NWO) that is funded by the Dutch
Ministry of Education, Culture and Science (OCW).

\begin{appendix}
\section{The $J$-currents and normal vectors}
\lb{Jcurrents}
In the following we will derive the $J$-currents on constant $r$, $u$ and $v$ hypersurfaces as well as the hypersurface $\gamma$, defined in \eqref{gross_h}.

We consider an arbitrary function
$F(u,v)$ independent of the angular coordinates. Let $\zeta$ be a levelset \mbox{$\zeta=\left\{F(u,v)=0\right\}$}. Then, the normal vector to the hypersurface $\zeta$ is given by
\bea
\lb{ngamma}
n^{\mu}_{\zeta}&=&\frac1{\sqrt{{\Omega^2}|{{\partial_u F \partial_v F}}}|}(\partial_v F\partial_u+\partial_u F\partial_v).
\eea
In particular, for the future directed normal vector of an $r(u,v)=const$ hypersurface we obtain
\bea
\lb{n_r}
n^{\mu}_{r=const}&=&\frac{1}{\sqrt{\Omega^2}}(\partial_u+\partial_v),
\eea
and on constant $u$ and $v$ null hypersurfaces with their related volume elements we have
\bea
\lb{n_u}
n^{\mu}_{u=const}&=&\frac{2}{\Omega^2}\partial_v,\quad \dV_{u=const}=r^2\frac{\Omega^2}{2}\md \sigma_{\mathbb S^2} \md v,\\
\lb{n_v}
n^{\mu}_{v=const}&=&\frac{2}{\Omega^2}\partial_u, \quad \dV_{v=const}=r^2\frac{\Omega^2}{2}\md \sigma_{\mathbb S^2} \md u.
\eea
For \eqref{n_u} and \eqref{n_v},
note that since vectors orthogonal to null hypersurfaces cannot be normalized, their proportionality has to be chosen consistent with an associated volume form in the application of the divergence theorem. Further, the volume element of the 4 dimensional spacetime is given by
\bea
\dV=r^2\frac{\Omega^2}{2}\md \sigma_{\mathbb S^2} \md u \md v.
\eea

According to \eqref{J}, using an arbitrary vector field \mbox{$X=X^u \partial_u+X^v \partial_v$} we then obtain: 
\bea
\lb{jgamma}
J_{\mu}^X(\phi) n^{\mu}_{\zeta}
&=&\frac{1}{\sqrt{\Omega^2}}\left[X^v\sqrt{\frac{\partial_u F}{\partial_v F}}
(\partial_v \phi)^2+X^u\sqrt{\frac{\partial_v F}{\partial_u F}}
(\partial_u \phi)^2\right]\nonumber\\
&&+\frac{\sqrt{\Omega^2}}{4}
\left[X^v\sqrt{\frac{\partial_v F}{\partial_u F}}+X^u\sqrt{\frac{\partial_u F}{\partial_v F}}\right]|\nabb \phi|^2,\\
\lb{jr}
J^{X}_{\mu}(\phi)n^{\mu}_{r=const}&=&\frac{1}{\sqrt{\Omega^2}}\left[X^v
(\partial_v \phi)^2+X^u
(\partial_u \phi)^2\right]
+\frac{\sqrt{\Omega^2}}{4}
\left[X^v+X^u\right]|\nabb \phi|^2, \\
\lb{jv}
J^{X}_{\mu}(\phi)n^{\mu}_{v=const}&=&\frac{2}{\Omega^2}X^u(\partial_u \phi)^2+\frac12 X^v|\nabb \phi|^2 ,\\
\lb{ju}
J^{X}_{\mu}(\phi)n^{\mu}_{u=const}&=&\frac{2}{\Omega^2}X^v(\partial_v \phi)^2+\frac12 X^u|\nabb \phi|^2 .
\eea

\section{The $K$-current}
\lb{Kcurrents}
In order to compute all scalar currents according to \eqref{K}
in $(u,v)$ coordinates we first derive the components of the deformation tensor which is given by
\bea
\lb{deftensor}
(\pi^X)^{\mu \nu}=\frac12(g^{\mu \lambda}\partial_{\lambda}X^{\nu}+g^{\nu \sigma}\partial_{\sigma}X^{\mu}+g^{\mu \lambda}g^{\nu \sigma}g_{\lambda \sigma, \delta}X^{\delta}),
\eea
where $X$ is an arbitrary vector field, \mbox{$X=X^u \partial_u+X^v \partial_v$} without angular components.\footnote{Recall that all our multipliers $N$, $-\partial_r$, ${S_0}$ and ${S}$ only contain $u$ and $v$ components.}
From this we obtain
\ben
(\pi^X)^{v v}&=&-\frac{2}{\Omega^2}\partial_u X^v,\\
(\pi^X)^{u u}&=&-\frac{2}{\Omega^2}\partial_v X^u,\\
(\pi^X)^{u v}&=&-\frac{1}{\Omega^2}(\partial_vX^v+\partial_uX^u)-\frac{2}{\Omega^2}\left( \frac{\partial_v \Omega }{\Omega}X^v+\frac{\partial_u \Omega }{\Omega}X^u\right) ,\\
(\pi^X)^{\theta \theta}&=&\frac1{r^2}\left( \frac{\partial_v r }{r}X^v+\frac{\partial_u r }{r}X^u\right) ,\\
(\pi^X)^{\phi \phi}&=&\frac1{r^2\sin^2 \theta}\left( \frac{\partial_v r }{r}X^v+\frac{\partial_u r }{r}X^u\right) .
\een

From \eqref{energymomentum} we calculate the components of the energy momentum tensor in $(u,v)$ coordinates as
\ben
T_{v v}&=&(\partial_v \phi)^2,\\
T_{u u}&=&(\partial_u \phi)^2,\\
T_{u v}&=&T_{v u}=\frac{\Omega^2}{4}|\nabb \phi|^2,\\
T_{\theta \theta}&=&(\partial_\theta \phi)^2+\frac{2 r^2}{\Omega^2} (\partial_u \phi \partial_v \phi)-\frac12r^2 |\nabb \phi|^2,\\
T_{\phi \phi}&=&(\partial_\phi \phi)^2+\frac{2 r^2 \sin^2 \theta}{\Omega^2} (\partial_u \phi \partial_v \phi) -\frac12r^2 \sin^2 \theta|\nabb \phi|^2.
\een

Multiplying the components according to \eqref{K} and using the relations \eqref{def_l_n} we obtain

\bea
\lb{Kplug}
K^X&=& -\frac{2}{\Omega^2}\left[\partial_u X^v (\partial_v \phi)^2+\partial_v X^u (\partial_u \phi)^2\right]\nonumber\\
&&-\frac{2}{r} \left[X^v+X^u\right](\partial_u \phi\partial_v \phi)\nonumber\\
&&-\left[\frac12 \left(\partial_v X^v+\partial_u X^u\right)+\left(\frac{\partial_v \Omega}{\Omega}X^v+\frac{\partial_u \Omega}{\Omega}X^u\right)\right]|\nabb \phi|^2.
\eea

\end{appendix}

 
\end{document}